\documentclass[12pt]{article}

\usepackage[total={7in,8.75in},
top=1in, left=0.9in, includefoot]{geometry}

\usepackage{dsfont}

\usepackage{psfrag}

\usepackage{calrsfs}
\usepackage{mathrsfs}
\usepackage{pdfsync}
\usepackage{amsmath, amsthm, amssymb, amsfonts}
\usepackage{mathtools}

\usepackage[shortlabels]{enumitem}

\usepackage[longnamesfirst, comma]{natbib}

\setcitestyle{round}
\usepackage{url}

\usepackage{float}
\usepackage{multirow}

\usepackage[usenames]{color}

\usepackage{tikz}

\usetikzlibrary{arrows,decorations.pathmorphing,backgrounds,positioning,fit,automata}
\usetikzlibrary{petri}
\usetikzlibrary{topaths}

\usepackage{indentfirst,calc,euscript}
\usepackage{setspace}
\usepackage[reals]{layout}
\usepackage{xr}
\usepackage{amscd}

\usepackage{sgame}
\usepackage{graphicx}
\usepackage{subcaption}

\usepackage[colorlinks=true,breaklinks=true,bookmarks=true,urlcolor=magenta,
     citecolor=blue,linkcolor=red,bookmarksopen=false,draft=false]{hyperref}

\usepackage{colortbl}
\newcommand{\card}{\operatorname{card}}

\newcommand{\ser}{\operatorname{ser}}

\newcommand{\BEq}{\operatorname{\mathsf{BEq}}}

\newcommand{\Opt}{\operatorname{\mathsf{Opt}}}

\newcommand{\PoA}{\operatorname{\mathsf{PoA}}}

\newcommand{\PoS}{\operatorname{\mathsf{PoS}}}

\newcommand{\WEq}{\operatorname{\mathsf{WEq}}}

\newcommand\BR{\,\mathsf{BR}}
\newcommand\Eq{\,\mathsf{Eq}}

\newcommand{\ignore}[1]{}

\newtheorem{theorem}{Theorem}

\newtheorem{claim}[theorem]{Claim}

\newtheorem{corollary}[theorem]{Corollary}

\newtheorem{lemma}[theorem]{Lemma}

\newtheorem{proposition}[theorem]{Proposition}

\theoremstyle{definition}

\newtheorem{definition}[theorem]{Definition}

\newtheorem{example}[theorem]{Example}

\newtheorem{remark}[theorem]{Remark}

\numberwithin{equation}{section}

\numberwithin{theorem}{section}

\begin{document}

\title{Dynamic Atomic Congestion Games\\ with Seasonal Flows}
\author{Marco Scarsini\thanks{This author is a member of GNAMPA-INdAM. The support of PRIN 20103S5RN3 and MOE2013-T2-1-158 is gratefully acknowledged.} \\
Dipartimento di Economia e Finanza \\
LUISS\\
Viale Romania 32\\
00197 Roma, Italy \\
\texttt{marco.scarsini@luiss.it}
\and
Marc Schr\"oder\thanks{The support of GSBE and MOE2013-T2-1-158 is gratefully acknowledged.}\\
Department of Quantitative Economics\\
Maastricht University\\
Tongersestraat 53\\
6211 LM Maastricht, The Netherlands\\
\texttt{m.schroeder@maastrichtuniversity.nl}
\and
Tristan Tomala\thanks{The support of the HEC foundation and of the Agence Nationale de la Recherche under grant ANR JEUDY, ANR-10-BLAN 0112 is gratefully acknowledged.}\\
HEC Paris and GREGHEC \\ 
1 rue de la Lib\'eration\\
78351 Jouy-en-Josas, France\\
\texttt{tomala@hec.fr}} 

\maketitle

\newpage

\begin{abstract}

\bigskip

We propose a model of discrete time dynamic congestion games with atomic players and a single source-destination pair. The latencies of edges are composed by free-flow transit times and possible queuing time due to capacity constraints. We give a precise description of the dynamics induced by the individual strategies of players and of the corresponding costs, either when the traffic is controlled by a planner, or when players act selfishly. In parallel networks, optimal and equilibrium behavior eventually coincides, but the selfish behavior of the first players has consequences that cannot be undone and are paid by all future generations. In more general topologies, our main contributions are three-fold. 

First, we show that equilibria are usually not unique. In particular, we prove that there exists a sequence of networks such that the price of anarchy is equal to $n-1$, where $n$ is the number of vertices, and the price of stability is equal to 1.

Second, we illustrate a new dynamic version of Braess's paradox: the presence of initial queues in a network may decrease the long-run costs in equilibrium. This paradox may arise even in networks for which no Braess's paradox was previously known.

Third, we propose an extension to model seasonalities by assuming that departure flows fluctuate periodically over time.  We introduce a measure that captures the queues induced by periodicity of inflows. This measure is the increase in costs compared to uniform departures for optimal and equilibrium flows in parallel networks.

\bigskip

\bigskip

\textbf{Keywords}: 
Network games, dynamic flows, price of seasonality, price of anarchy, max-flow min-cut.

\textit{OR/MS Subject Classification}: networks/graphs: multicommodity, theory; games/group decisions: noncooperative; transportation: models, network.
\end{abstract}

\newpage
\section{Introduction}\label{se:intro}

The analysis of transportation networks naturally leads to the consideration of congestion games, where each agent selfishly behaves as to minimize her own time on the road without regard for the effects that this behavior has on the other agents' traveling time. The outcome of the individual selfish behavior can be compared to the outcome that a social planner would choose. A way of comparison is, for instance, the price of anarchy \citep[see, e.g.][]{Rou:MIT2005, Rou:AGT2007, RouTar:AGT2007}, namely the ratio of the worst social cost induced by selfish behavior to the optimal social cost. 

Although the motivation for this theory is rooted in the study of traffic flows, most of the existing literature is actually static. The commonly adopted justification is that the static game represents the steady state of a dynamic model where the flow over the network is constant over time. Yet, for determining how the steady state is reached, a careful study of dynamic models is required. As we shall see, the behavior of agents in the transient phase may have an impact on the long-run outcome.

In this paper we study a dynamic model of congestion where the players have symmetric and unsplittable weights. This could be a high-level model for, e.g., traffic network, where each player is a car in a traffic network, or a telecommunication network, where each player is a data packet.
We characterize the optimal long-run flows and latencies, i.e., the ones induced by a benevolent long-lived social planner. 
When the players act selfishly in order to minimize their own traveling time, without heeding the planner's suggestion, the situation can be modeled as a noncooperative game. 
For some topologies of the network and when the inflow of players is uniform over time we are able to characterize the equilibria of this game. 
We consider the efficiency of its equilibria for various topologies and we show that some forms of Braess-type paradoxes are possible. Finally we devote our attention to the case where the inflow of players is periodic over time.

\subsection{Model} 

We analyze an atomic dynamic congestion game, based on the deterministic queuing model of \citet{KocSku:TCS2011}. 
Atomic models are typically more complicated to analyze than nonatomic models and have less nice properties. Nevertheless, they may be a better fit when the number of players is not huge and a nonatomic approximation is not justifiable. Atomic models have been used, for instance, in telecommunications \citep[see, for instance,][]{TekLiuSouHuaAhm:IEEEACM2012}.
The dynamics of the model is described as follows. Time is discrete and at each stage, a generation of finitely many players departs from the source with the goal of reaching the destination as fast as possible. 
We assume that each player has a unit weight and is unsplittable. 
We see this assumption of symmetry as a first order approximation when the size of the vehicles or of the data packet is not too dissimilar.
Each player chooses a route from source to destination, knowing the choice of the previous players. Each edge of the network is endowed with a free-flow transit time and a capacity. When a player enters an edge on the chosen route, she travels on that edge at a constant speed. When reaching the head of the edge, a queue might have formed since at most the capacity number of players can exit the edge at the same time. We assume that there is a global priority among players to determine who leaves the edge first. The latency suffered by the player on an edge is thus the sum of the transit and waiting times. The total latency suffered by the player is then the sum of the latencies suffered on all the edges she uses.

\subsection{Results}

We first study social optimality when the inflow is constant and at most the capacity of the network.  We prove that optimal flows exist and show that there is an optimal flow such that, at each stage, the current flow over routes minimizes the total cost among feasible static flows, i.e., the flows that satisfy all the capacity constraints. 

Then, we turn to the behavior of selfish players. In particular, we consider equilibria in which each player arrives at each intermediate vertex as fast as possible.  These are called uniformly-fastest-route equilibria. In general, such equilibria are not unique. 
For parallel networks, in all equilibria the flow coincides with the optimal flow from some stage on, and in the worst equilibrium, all players eventually pay the highest transit cost of the network. The intuition is that the first players all choose the fastest routes and induce congestion. Eventually, all routes get so congested that all latencies become equal. This result shows the impact of the dynamic nature of the model on latencies. While optimal and equilibrium behavior eventually coincide, the selfish behavior of the first players has consequences that cannot be undone and are paid by all future generations. 

In more general networks the results for equilibrium flows become more complicated. For chain-of-parallel networks, the equilibrium costs can be derived from the results for parallel networks, but the corresponding flows can be quite different and even aperiodic.

We also examine various efficiency measures of equilibria such as the price of anarchy, the price of stability and the Braess ratio. We demonstrate the following phenomena by examples. 
Firstly, there is a sequence of instances such that the price of anarchy is equal to $n-1$, where $n$ is the number of vertices, and the price of stability is equal to 1, illustrating the difference in long-run equilibrium costs. This is one of the main differences between the atomic and nonatomic cases, since in our atomic model there may exist multiple equilibria whose  behavior can be quite different, whereas in the nonatomic model, equilibrium is unique \citep[see][]{ComCorLar:ALP2011}. \citet{KocSku:TCS2011} show that the price of anarchy increases logarithmically in the number of edges if all edge capacities are equal to $1$. As a byproduct of our results, we obtain an example of a nonatomic game, where all capacities are $1$ and the price of anarchy is linear in the number of edges.
Secondly, we study Braess's paradox \citep[see][]{Bra:U1968,Bra:TS2005}, namely the decrease of the total equilibrium cost after deletion of an edge. This may happen even when the network does not contain the Wheatstone graph as a subnetwork \citep[see][]{MacLarSte:TCS2013}. We also obtain a variant of Braess's paradox: the equilibrium cost might decrease when there are initial queues in the network, or when the length of an edge increases. 
Thirdly, we study the Braess ratio \citep[see][]{Rou:JCSS2006}, namely, the largest factor by which the equilibrium cost can be improved by removal of an edge. We consider a set of networks with $n$ vertices where this ratio is $n-1$. A similar result appears in \citep{MacLarSte:TCS2013} in the nonatomic case.

In the last section we consider periodic inflows, we define a distance between two inflows, and we show that in parallel networks at capacity, periodicity adds the same cost in equilibrium and at the optimum, and this added cost is exactly the distance between the periodic and the uniform inflows.

\subsection{Related literature}\label{suse:existing}

Dynamic congestion games belong to the wider class of models of flows over time. \citet{ForFul:OR1958,ForFul:PUP1962}  introduced these models in a discrete time setting by considering the problem of maximizing the flow from source to destination in a given finite time horizon. \citet{Gal:MMJ1959} considered a refinement of the above problem, called earliest arrival flow, where the aim is to simultaneously maximize the flow for every time before the deadline; \citet{Wil:OR1971} and \citet{Min:OR1973} developed  algorithms for solving it. The continuous-time versions  were studied by \citet{Phi:MOR1990} and \citet{FleTar:ORL1998}, respectively. 
We refer the reader to \citet{Sku:RTCO2009} for a detailed analysis and an extensive bibliography. Equilibrium  concepts in dynamic network models date  back to \citet{Vic:AER1969}  in the economic literature and to \citet{Yag:TR1971} in the transportation literature.  
We refer the reader to \citet{Koc:PhD2012} for an extensive list of references on this topic. Recent mathematical formulations of the model resort to deterministic queueing theory, as introduced originally by \citet{Vic:AER1969} and later developed by \citet{HenKoc:TS1981}. In this stream of literature \citet{Aka:TRB2000,Aka:TS2001,AkaHey:TS2003}, \citet{Mou:TRB2006,Mou:TS2007}, \citet{AsnUkk:AGT2009}, \citet{HoeMirRogTen:P5IWINE:2009}, and especially \citet{KocSku:TCS2011} extended some results known for static congestion games to dynamic congestion games. The latter authors use a deterministic queueing model to study dynamic flows and  characterize Nash equilibria. They show the relation between dynamic and static models and they compute the price of anarchy for the dynamic model.
Along these lines,  \citet{ComCorLar:ALP2011,ComCorLar:OR2015} studied equilibria for flows over time in the single-source single-sink deterministic queuing model and  proved existence and uniqueness of equilibria when the inflow rate is piecewise constant. 
\citet{KocNasSku:MMOR2011} used measure-theoretic techniques to combine continuous and discrete time models of flow over time and, among other things, extended to this general setting the classical max-flow min-cut theorem.  \citet{BhaFleAns:GEB2015} considered a Stackelberg model with a network manager acting as a leader who chooses the capacity of each edge in a way that does not exceed its physical limit. They were able to bound the price of anarchy for this model.

Among this literature, our model belongs to the class of deterministic queueing models and is close to the one developed by \citet{KocSku:TCS2011}. 
There are some technical differences with this literature since our model is in discrete time and with atomic players, similar to 
\citet{WerHolKru:ORP2014}. This can induce, for instance, a possible multiplicity of equilibria.

Importantly, our  focus  differs from these papers. Many of them seek to characterize equilibrium flows (e.g. \citet{KocSku:TCS2011}), prove existence results \citep[see, e.g.,][]{ComCorLar:ALP2011,ComCorLar:OR2015}, and provide algorithms for computing equilibria. By contrast, we emphasize how the system evolves towards a steady state. Starting from an empty network, we study how the behavior of the first users impacts the equilibrium steady state and how this steady state is reached. The transient phase  that leads to the steady state is thus particularly important in our model. \citet{ShaShi:PACMSICMMCS2010} consider the transient phase of a dynamic network before a steady state equilibrium is reached. Although their model is stochastic, some of the questions they consider are close in spirit to our  model.
\citet{Com:MP2015} provides a nice survey of congestion models under uncertainty and describes a model of adaptive dynamics that gives a microfoundation for steady state traffic equilibrium models.

We now comment on some more accessory features of our work in relation to the existing literature.
Our model belongs to the class of congestion games with atomic players. In his fundamental paper, \citet{War:PICE1952} modelled the selfish behavior of a huge number of agents on a network as a nonatomic flow and introduced an equilibrium concept that has become the standard reference in the literature. \citet{ChaCoo:TTF1961} showed the relation between Nash and Wardrop equilibria and \citet{HauMar:N1985} proved that, under some conditions, the Wardrop equilibrium in a nonatomic model can be obtained as a limit of Nash equilibria of atomic models. The relation between atomic and nonatomic games has been recently studied by \citet{BhaFleHua:IPCO2010}. A nice survey on Wardrop equilibria can be found in \citet{CorSti:Wiley2010}. General congestion games with a finite number of players were introduced by \citet{Ros:IJGT1973}, who proved that they have pure Nash equilibria; they are actually  isomorphic to potential games \citep[see][]{MonSha:GEB1996}. The issue of multiplicity of equilibria in atomic congestion games was studied by \citet{Har:TS1988,BhaFleHoyHua:P20AACMSIAMSDA2009}. Consistent with this literature, we find multiple equilibria for our game.
 
In order to obtain well defined dynamics, we use a priority order. This approach can be found in earlier works.
\citet{FarOlvVet:CJTCS2008} introduced a routing model with a general priority  scheme for players on different edges. This allowed the authors to introduce a time dependence in the model and to define a cost for each player that depends on the actions of the players with a higher priority. Among the many possibilities, they considered a global priority scheme which is the same for every edge, and a time dependent priority scheme  where  priority is decided by who arrives first on an edge. 
A similar global priority scheme was exploited by \citet{HarHeiPfe:TCS2009}, who studied multicommodity flows where commodities are routed sequentially in a network. In their model demands for commodities are revealed in an online fashion and can be split along several paths. They framed the problem as an optimization problem and they studied online algorithms for its solution. In a related paper,  \citet{HarVeg:LNCS2007} considered a   model in which players' demands change over time and are released in $n$ sequential games in an online fashion. In each game, the new demands form a Nash equilibrium, and their routing remains unchanged afterwards. These three models do not explicitly take into account the dynamics of the flows of players over the network.
Our model retains the idea of a priority scheme, but it is dynamic. 

Finally, for measuring the efficiency of a game we use the now famous price of anarchy, i.e., the ratio between the worst Nash equilibrium latency and the socially optimal latency, and the price of stability, i.e., the ratio between the best Nash equilibrium latency and the socially optimal latency. These two measures were introduced by \citet{KouPap:STACS1999} and \citet{SchSti:P14SIAM2003}, respectively. Their names were coined by \citet{Pap:PACM2001} and \citet{AnsDasKleTarWexRou:SIAMJC2008}, respectively. Inefficiency of equilibria in routing games has been studied by several authors (see among others \citet{RouTar:JACM202, RouTar:GEB2004}, \citet{CorSchSti:MOR2004, CorSchSti:GEB2008, CorSchSti:OR2007}). 

\citet{Bra:U1968,Bra:TS2005} shows that removing an edge in a network can improve the equilibrium latency for all players in a static model. 
\citet{Dag:TS1998} shows similar paradoxical phenomena in traffic models when queues have physical magnitude.
\citet{Rou:JCSS2006} introduces a measure, called the Braess ratio, that quantifies the extent of Braess's paradox.  
A study of the network topologies for which the paradox may exists can be found in \citet{Mil:GEB2006} for static games and in \citet{MacLarSte:TCS2013} for dynamic games.
An analysis of a dynamic Braess-type paradox in communication networks is provided by \citet{XiaHil:IEEETCSIIEB2013}. 

\subsection{Organization of the paper}

The paper is organized as follows. Section~\ref{se:model} presents the model. 
Section~\ref{se:optimum} contains a characterization of the optimum strategy and cost for the case of constant inflows.
Section~\ref{se:topologies} studies the case of parallel networks and chain-of-parallel networks.  
Section~\ref{se:anarchy}  examines efficiency of equilibria and Braess-type phenomena.
Section~\ref{se:seasonal} proposes an extension to model seasonalities.
Section~\ref{se:conclusion} concludes and proposes some open problems.
All proofs are relegated to the Appendix or to the online Supplementary Material.

\section{The model}\label{se:model}

We study a dynamic congestion game on a general directed network with a single source-destination pair, where each edge has a transit cost and a capacity. Formally, consider a directed multigraph $\mathcal{G}=(V,E)$, where $V$ is a finite set of vertices and $E$ is a finite set of edges. We then define a network $\mathcal{N}=(\mathcal{G}, (\tau_{e})_{e \in E}, (\gamma_{e})_{e \in E})$, such that for each $e\in E$, the quantities $\tau_{e}\in\mathbb{N}$  and $\gamma_{e}\in \mathbb{N}$ are the \emph{free-flow transit cost} and the \emph{capacity} of edge $e$, respectively.

A \emph{path} in the network is a finite sequence of edges $(e_{1},\dots, e_{n})$ such that the head of $e_{i}$ coincides with the tail of $e_{i+1}$ for each $i=1,\dots, n-1$.

We make the following assumptions: 
\begin{enumerate}[(a)]

\item
There are two special vertices, the \emph{source} $s$, which has only outgoing edges, and the \emph{destination} $d$, which has only incoming edges (we use the symbol $d$ for destination, rather than the more common $t$, because we reserve $t$ for time). Source and destination are unique.

\item
For each vertex $v\in V\setminus\{s,d\}$, there exists at least a path from $s$ to $v$ and a path from $v$ to $d$.

\end{enumerate}

We call \emph{route}  a path from $s$ to $d$. The set of all routes is denoted by $\mathcal{R}$. The above assumptions guarantee that any path can be extended to a route.

Time is discrete and, at each stage $t$ finitely many players enter the network at the source and choose a route from $s$ to $d$. Each player represents a unit packet of traffic. For simplicity, we assume that all players have the same size, which we normalize to 1. The dynamics of the model is the following.
\begin{itemize}
\item 
At each stage $t\in\mathbb{N}_{+}$, a finite set $G_{t}$ of players, called the \emph{generation} at time $t$, departs from the source.
For all $t \in \mathbb{N}_{+}$ define,
\begin{equation*}
\delta_{t} = \card(G_{t}) \quad\text{and}\quad \mathcal{D} = \{\delta_{t}\}_{t\in \mathbb{N}_{+}}.
\end{equation*}
Denote $[it]$ the $i$-th player in generation $G_{t}$   (when there is no risk of confusion, the square brackets are removed). 

We thus have an infinite set of players $G:=\cup_{t} G_{t}$. We order this set (anti-lexicographically) by $\lhd$ as follows:
\[
[js] \lhd [it] \quad \text{iff}\quad s < t \quad\text{or}\quad (s = t \ \text{and}\ j < i). 
\]
This order represents priorities: if $[js] \lhd [it]$ and if these two players enter the  edge $e$ at the same time, then $[js]$ exits $e$ before $[it]$. Such a global priority is a natural choice for breaking ties in congestion games, see \citet{FarOlvVet:CJTCS2008} and \citet{WerHolKru:ORP2014}.

\item Each player chooses a route in $\mathcal{R}$.

\item At time $t$ player $[it]$ departs from the source $s$, takes the chosen route, and progresses with steps of size $1$ per unit of time along an edge $e$. After $\tau_e$ time units, she arrives at the head of the edge, where a queue may have formed.

\item The rules for exiting the queue are the following:

\begin{itemize}
\item
All players that entered edge $e$ before $[it]$ and those players $[js] \lhd [it]$ who entered $e$ at the same time as $[it]$ are ahead of $[it]$ in the queue.  That is, there is no over-taking in queues and players are ordered first by time of arrival, then by priority.

\item
At most $\gamma_{e}$ players can exit $e$ simultaneously. When player $[it]$ arrives at the end of $e$, if she finds less than $\gamma_{e}$ players in the queue, then she exits immediately; otherwise, only the first $\gamma_{e}$ players exit at this stage and player $[it]$ waits for one stage. This process repeats until there remain less than $\gamma_{e}$ players in the queue ahead of player $[it]$. Then, $[it]$ exits edge $e$ and continues along the chosen route.
\end{itemize}
\item
The process is repeated until player $[it]$ arrives at the destination $d$ and quits the system.
\end{itemize}

These rules define a \emph{dynamic congestion game} denoted $\Gamma(\mathcal{N}, \mathcal{D})$. Each strategy profile $\sigma \in \mathcal{R}^{G}$ induces queues on edges. We denote $\ell_{it}(\sigma)$ the \emph{latency} suffered by player $[it]$, defined as
\[
\ell_{it}(\sigma)=c_{it}(\sigma)+w_{it}(\sigma),
\]
where $c_{it}(\sigma):=\sum_{e\in r_{it}(\sigma)}\tau_{e}$ is the \emph{transit cost} paid by player $[it]$ and $w_{it}(\sigma)$ is the \emph{waiting cost} paid by player $[it]$, namely, the total number of stages that $[it]$ spends queueing, summed over the edges that she crosses. Both costs are additive, the total cost over the route is the sum of costs over the edges of the route.

We define the \emph{total transit cost} $c_{t}$, the \emph{total waiting cost} $w_{t}$, and the \emph{total latency} $\ell_{t}$ at stage $t$ as follows:
\begin{align*}
c_{t}(\sigma) &= \sum_{[it]\in G_{t}} c_{it}(\sigma), \\
w_{t}(\sigma) &= \sum_{[it]\in G_{t}} w_{it}(\sigma), \\
\ell_{t}(\sigma) &= \sum_{[it]\in G_{t}} \ell_{it}(\sigma) =c_{t}(\sigma)+w_{t}(\sigma).
\end{align*}
For each integer $T$, the \emph{average total latency} over the period $\{1,\dots, T\}$ is
\[
\bar{L}_{T}(\sigma)=\frac{1}{T}\sum_{t=1}^{T} \ell_{t}(\sigma).
\]
If $\lim_{T \to \infty}\bar{L}_{T}(\sigma)$ exists, then it
is called \emph{asymptotic average total latency} for the strategy $\sigma$.

\begin{definition}\label{de:SO} 
A strategy profile $\sigma$ is \emph{(socially) optimal} if
\begin{equation}\label{eq:optimum}
\liminf_{T\to\infty}\bar{L}_{T}(\sigma')\geq\limsup_{T\to\infty}\bar{L}_{T}(\sigma) \text{ for all }\sigma' \in \mathcal{R}^{G}.
\end{equation}
Call $\mathcal{O}(\mathcal{N}, \mathcal{D})$ the set of strategies $\sigma \in \mathcal{R}^G$ for which \eqref{eq:optimum} holds. Then  
\begin{equation}\label{eq:Opt}
\Opt(\mathcal{N}, \mathcal{D}):=\liminf_{T\to\infty}\bar{L}_{T}(\sigma)=\limsup_{T\to\infty}\bar{L}_{T}(\sigma) \quad\text{with }\sigma \in \mathcal{O}(\mathcal{N}, \mathcal{D})
\end{equation}
is called the \emph{optimal latency}.
\end{definition}

\begin{definition}\label{de:SPE} 
A strategy profile $\sigma$ is a \emph{Nash equilibrium} if
\[
\ell_{it}(\sigma)\leq\ell_{it}(\sigma'_{it},\sigma_{-it})\text{ for all }[it]\in G, \text{ for all }\sigma'_{it}\in \mathcal{R},
\]
where $\sigma_{-it}$ indicates the profile of strategies of all players different from $[it]$.

A Nash equilibrium $\sigma$ is a \emph{uniformly fastest route (UFR) equilibrium} if for every player $[it]$ and for every vertex $v$ on the route $\sigma_{it}$, there is no alternative route $\sigma'_{it}$ that allows player $[it]$ to arrive at $v$ earlier than under $\sigma_{it}$. 
\end{definition}

There is obviously an asymmetry between Definitions~\ref{de:SO} and \ref{de:SPE} which stems from the type of rationality driving the two concepts. An optimal strategy is the choice that a long-lived planner would like to take, in order to optimize the long-run social welfare. By contrast, an equilibrium is a strategy profile such that each finitely lived player optimizes her cost given the choices of the other players.

In non-atomic games \citep[see][]{KocSku:TCS2011}, the  Nash-flow-over-time definition is equivalent to assuming that all particles arrive at each intermediate vertex as early as possible, which is also equivalent to requiring that no flow overtakes any other flow. For atomic games, this equivalence does not hold. For instance, take a player who is the last in her generation. The immediate successor is the first in the next generation, and therefore there is a time difference between the departures of these two subsequent players. In that case, it might be that a player does not want to arrive at each intermediate vertex as early as possible. It might also be the case that a player overtakes her predecessor in equilibrium, while arriving at the destination at the same time point. In fact, the three notions which coincides in the non-atomic case, might yield different long-run latencies in the atomic case. Example~\ref{ex:verybadNash} in the online Supplementary Material illustrates this phenomenon.

We denote $\mathcal{E}(\mathcal{N}, \mathcal{D})$ the set of UFR equilibria of the game $\Gamma(\mathcal{N}, \mathcal{D})$. We argue that a UFR equilibrium exists, a similar argument can be found in \citet{WerHolKru:ORP2014}. In an empty network, there is always a shortest route with the property that every intermediate vertex is reached as early as possible, since this is equivalent to the static shortest path problem. If the first player chooses such a route, then that player cannot be overtaken. Taking this choice into account, the second player chooses a route that reaches every intermediate vertex as early as possible, so that   she cannot be overtaken either. Continuing  this procedure iteratively yields a UFR equilibrium. We present this result as a lemma whose proof is given in  the Appendix for completeness.

\begin{lemma}\label{le:existence}
The set $\mathcal{E}(\mathcal{N}, \mathcal{D})$ is not empty.
\end{lemma}

The quantity
\begin{equation}\label{eq:WEq}
\WEq(\mathcal{N}, \mathcal{D}):=\sup_{\sigma \in\mathcal{E}(\mathcal{N}, \mathcal{D})} \limsup_{T \to \infty}\bar{L}_{T}(\sigma)
\end{equation}
is called  the \emph{worst equilibrium latency}
and the quantity
\begin{equation}\label{eq:BEq}
\BEq(\mathcal{N}, \mathcal{D}):=\inf_{\sigma \in\mathcal{E}(\mathcal{N}, \mathcal{D})} \limsup_{T \to \infty}\bar{L}_{T}(\sigma)
\end{equation}
is called  the \emph{best equilibrium latency}.

\section{Socially optimal strategies}\label{se:optimum}

In this section  we characterize flows and costs generated by optimal strategies. 
Before stating the results, a simple observation is that the number of players entering the network over time has to be compared with the number that the network is able to absorb. 

\begin{definition}
A \emph{cut} in the network $\mathcal{N}$ is a subset of edges $C\subseteq E$ such that each route contains at least one element of $C$. The capacity of a cut $C$ is $\gamma_{C}=\sum_{e\in C}\gamma_{e}$. 
Call $\mathcal{C}(\mathcal{N})$ the set of all cuts in $\mathcal{N}$.
A \emph{minimum cut} is a cut $C$ such that
\[
\gamma_{C} \le \gamma_{C'} \quad \text{for all } C' \in \mathcal{C}(\mathcal{N}).
\]
The capacity $\gamma$ of the network $\mathcal{N}$ is the capacity of any minimum cut.
\end{definition}

Until further notice (see Section \ref{se:seasonal}), we assume that the number of players in each generation is uniform over time, i.e.,  $\delta_{t}=\delta$ for all $t \in \mathbb{N}_{+}$ (abusing notation, the departure sequence $\mathcal{D}$ will be denoted simply by $\delta$). From the max-flow min-cut theorem of  \citet{ForFul:OR1958}, if  $\delta >\gamma$, then the lengths of queues on the edges of the minimum cut diverge to infinity and thus the long-run average total cost is infinite under any strategy profile. Therefore, in this section, we assume $\delta \le \gamma$.

We first recall some usual  concepts of optimality for static flows over networks. A (static) \emph{network flow} $f$ assigns a non-negative flow value $f_{e}$ to each edge $e\in E$ (in our setting, these are  integers). The flow $f$ is \emph{feasible} if it obeys the capacity constraints, i.e., $f_{e}\leq \gamma_{e}$ for each $e\in E$, and flow conservation, i.e., the outflow minus the inflow at each vertex $v\in V\setminus\{s,d\}$ is $0$. The value of a feasible network flow is  the inflow at $d$, in our case this is 
$\delta$.
A  static flow may also be defined over routes. Consider a set of $\delta$ players and the routes that they choose. They induce a static flow over routes $F$ which assigns to each route $r$ an integer $F_{r}$, such that $\sum_{r\in\mathcal{R}}F_{r}=\delta$.
This in turn induces a flow over edges by letting $f_e=\sum_{\{r : e\in r\}}F_{r}$.

The min-cost flow (static) optimization problem is the minimization  of the total transit cost among all feasible flows with a value of $\delta$. Let $f^{*}$ be an optimal feasible solution and $F^{*}$ a corresponding optimal flow over routes. This is the optimal assignment that a planner would chose in a static framework with a single set of $\delta$ players, subject to feasibility.

Back to the dynamic problem, given a strategy profile $\sigma$ and a route $r$, denote $N^{r}_{t}(\sigma)$ the number of players who choose route $r$ at stage $t$ under the strategy profile $\sigma$.

\begin{theorem}\label{th:optimum}
Consider the game $\Gamma(\mathcal{N}, \delta)$, where $\delta \le \gamma$. Let $f^{*}$ be an optimal feasible min-cost network flow with a value of $\delta$ and let $F^{*}$ be the corresponding flow over routes.
Then there exists $\sigma \in\mathcal{O}(\mathcal{N}, \delta)$ such that for each stage $t\in\mathbb{N}_{+}$ and route $r\in\mathcal{R}$,
\begin{equation}\label{eq:optcost}
N^{r}_{t}(\sigma)=F^{*}_{r}.
\end{equation}
\end{theorem}

This result says that finding the optimal long-run latency boils down to computing a min-cost static flow. 
This  problem  is well studied in the literature and   algorithms for solving it efficiently are known \citep[see, for instance,][]{ForFul:PUP1962,AhuMagOrl:Prentice1993,Sch:Springer2003A,KorVyg:Springer2012}.

The detailed proof is in the Appendix. The main insight is as follows. Consider first the case of $\delta=\gamma$. Since in this case the inflow is equal to the capacity of the network, if the planner violates the capacity constraints at some stage $t$, then this creates a queue that will remain through time. An excess of players on some edge of the min-cut, can only be compensated by a future deficit on that edge, which entails an excess on some other edge of the min-cut. Consequently, queues can never be undone and the long-run planner is better-off never creating any queue. When $\delta<\gamma$ we can consider an augmented network obtained concatenating an edge of capacity $\delta$ before the origin of the original network. The optimum of this augmented network (now at capacity $\delta$) is the optimum of the original network with a flow $\delta<\gamma$.

Note that this problem is different from finding an earliest arrival flow, where a given set of players (or particles) has to be shipped to the destination with the requirement that each particle arrives as fast as possible
\citep[see, for instance,][]{Gal:MMJ1959,HopTar:MOR2000}. \citet{JarRat:MS1982} has shown that this is equivalent to the problem of having as many players as possible to reach the destination in a prescribed amount of time, or to the problem of minimizing the average time to evacuate the system. In our setting instead, players enter the system infinitely often over time, and the goal of the planner is to minimize the average traveling time, which is achieved by not creating queues.

%\begin{color}{red}
%I HAVE SUPPRESSED THE FOLLOWING SEQUENCE. This problems is motivated by evacuation situations, where, in case of a fire, all people in a building have to reach some safe destination as early as possible. 
%\end{color}

\section{Equilibria for simple network topologies}\label{se:topologies}
This section describes the impact of the dynamic nature of the model on equilibrium latencies. We  first consider  parallel networks for which we are able to give sharp characterizations of equilibria. Then, we extend the results to chain-of-parallel networks.

\subsection{Parallel networks}\label{suse:parallel}

In a \emph{parallel network}, each route contains a single edge (see Figure \ref{fi:parallelnetwork}).
For such networks, we can compute exactly the optimum and equilibrium costs for uniform departure inflow ($\delta_t=\delta$ for all $t\in\mathbb{N}_+$), with $\delta\leq\gamma$. 

\begin{figure}[h]
\centering
\begin{tikzpicture}[->,>=stealth',shorten >=1pt,auto,node distance=6cm,
  thick,main node/.style={circle,fill=blue!20,draw,minimum size=25pt,font=\sffamily\Large\bfseries},source node/.style={circle,fill=green!20,draw,minimum size=25pt,font=\sffamily\Large\bfseries},dest node/.style={circle,fill=red!20,draw,minimum size=25pt,font=\sffamily\Large\bfseries}]

  \node[source node] (1) {$s$};
  \node[dest node] (2) [right of=1] {$d$};

  \path[every node/.style={font=\sffamily\small}]
    (1) edge [bend right = 60] node[above] {$e_{1}$} (2)
        edge [bend right = 30] node[above] {$e_{2}$} (2)
        edge node [above] {$e_{3}$} (2)
        edge [bend left = 30] node[above] {$e_{4}$} (2)
        edge [bend left = 60] node[above] {$e_{5}$} (2);

\end{tikzpicture}
~\vspace{0cm} \caption{\label{fi:parallelnetwork} Parallel network.}
\end{figure}
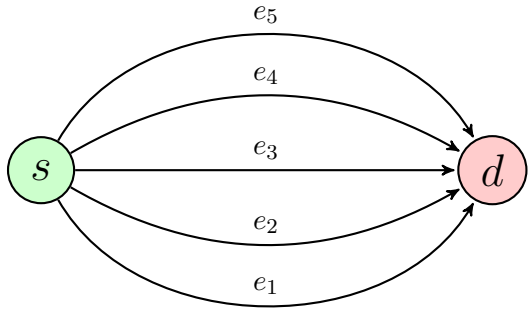

For convenience, we impose an order $\prec$ on the edges such that their lengths are weakly increasing along this order: $i < j \implies e_{i} \prec e_{j} \implies \tau_{e_{i}}\leq\tau_{e_j}$.

Observe that a parallel network admits a unique cut and thus its capacity is simply the sum of the capacities of its edges $\gamma=\sum_{e\in E}\gamma_{e}$.

For each $f \in E$, denote 
$f^{\prec}:=\{e \in E:e \prec f\}\quad\text{and}\quad f^{\precsim}:=f^{\prec} \cup \{f\}$,
where $e_{1}^{\prec} = \varnothing$.
For each $\delta \le \gamma$, there is a unique $f_{\delta}\in E$ such that
\[
\sum_{e \in f_{\delta}^{\prec}} \gamma_{e} < \delta \leq\sum_{e \in f_{\delta}^{\precsim}}\gamma_{e}.
\]
Define the (sub-)network $\mathcal{N}_{\delta}$ with set of edges $f_{\delta}^{\precsim}$, such that each edge $e \in f_{\delta}^{\prec}$ has transit cost $\tau_{e}$ and capacity $\gamma_{e}$, and edge $f_{\delta}$ has transit cost $\tau_{f_{\delta}}$ and capacity $\delta-\sum_{e \in f_{\delta}^{\prec}} \gamma_{e}\le \gamma_{f_{\delta}}$. The total capacity of $\mathcal{N}_{\delta}$ is precisely $\delta$. 

\begin{theorem}\label{th:parallelunif} 
Consider the game $\Gamma(\mathcal{N},\delta)$, where $\mathcal{N}$ is a parallel network and $\delta \le \gamma$. Then
\begin{align*}
\Opt(\mathcal{N}, \delta)&=\sum_{e \in f_{\delta}^{\prec}} \gamma_{e}\tau_{e} + \left(\delta- \sum_{e \in f_{\delta}^{\prec}}  \gamma_{e}\right)\tau_{f_{\delta}},\\
\WEq(\mathcal{N}, \delta)&= \delta \tau_{f_{\delta}},
\end{align*}
and there exists a time $t_{0}$ such that for each $t\geq t_{0}$, 
\[
N^{e}_{t}(\sigma^{\Opt})=N^{e}_{t}(\sigma^{\WEq})=\gamma_{e}\quad \text{for } e \prec f_{\delta} \quad \text{and}\quad N^{f_{\delta}}_{t}(\sigma^{\Opt})=N^{f_{\delta}}_{t}(\sigma^{\WEq})= \delta -\sum_{e \in f_{\delta}^{\prec}}  \gamma_{e}.
\]

\end{theorem}

The intuition for the proof is simple.  First, it is clear that in a social optimum, the planner uses only the sub-network 
$\mathcal{N}_{\delta}$, and from Theorem~\ref{th:optimum}, no queues are created. Thus, any optimal strategy sends exactly $\gamma_{e}$ players on each edge $e$ of $\mathcal{N}_\delta$ at each stage. Regarding equilibria, the idea is that the selfish players first fill short edges, thereby creating
queues. As a result, the latencies of these edges increase for future generations, and eventually,
all latencies become equal to the highest transit cost for that (sub)network. From that point on, players are basically indifferent and, as in an optimal strategies, exactly $\gamma_{e}$ players choose edge $e$ at each stage. The formal proof is in the Appendix.

\subsection{Chain-of-parallel networks}\label{suse:chainofparallel}

Let $\mathcal{N}_{1}, \mathcal{N}_{2}$ be two networks with respective source-destination pairs $(s_{1},d_{1})$, 
$(s_{2},d_{2})$. The \emph{series composition} of $\mathcal{N}_{1}$ and $\mathcal{N}_{2}$ is the network $\mathcal{N}=\mathcal{N}_{1}\oplus\mathcal{N}_{2}$ with source $s_{1}$, destination $d_{2}$ and where $d_{1}$ and $s_{2}$ are merged together. A \emph{chain-of-parallel network} is  obtained by composing parallel networks in series. For $h \in \{1, \dots, H\}$, let $\mathcal{N}^{(h)}=\left(E^{(h)}, (\tau_{e})_{e \in E^{(h)}}, (\gamma_{e})_{e \in E^{(h)}}\right)$ be a parallel network and  consider the network $\mathcal{N}_{\ser}(H)$ obtained by composing $\mathcal{N}^{(1)}, \dots, \mathcal{N}^{(H)}$ in series. Clearly, any subnetwork $\mathcal{N}^{(h)}$ is a cut of $\mathcal{N}_{\ser}(H)$. Let $\mathcal{N}^{(*)}$ be a minimum cut of $\mathcal{N}_{\ser}(H)$ and let $\gamma^{(*)}$ be the capacity of $\mathcal{N}^{(*)}$.

We obtain the following characterization for optimal and equilibrium values.

\begin{theorem}\label{th:chain}
Consider the game $\Gamma(\mathcal{N}_{\ser}(H), \gamma^{(*)})$. Then
\begin{align*}
\Opt\left(\mathcal{N}_{\ser}(H), \gamma^{(*)}\right)&=\sum_{h=1}^{H} \Opt\left(\mathcal{N}^{(h)}, \gamma^{(*)}\right),\\
\WEq\left(\mathcal{N}_{\ser}(H), \gamma^{(*)}\right)&=\sum_{h=1}^{H} \WEq\left(\mathcal{N}^{(h)}, \gamma^{(*)}\right).
\end{align*}
\end{theorem}

The insights are as follows. First, for optimal strategies the modular structure of the graph implies that each subnetwork can be analyzed separately in such a way that no queues are created. Second, modularity implies that the worst equilibrium latency has to be at least the sum of the worst equilibrium latencies of each  subnetwork. The uniformly fastest route property guarantees that the latency cannot be worse.

The above result may  seem straightforward. An important point to consider is that, although the number of players departing from the source is uniform over time, the number of players who exit a module may actually be non-uniform (periodic, or even aperiodic), at equilibrium. The following example illustrates this phenomenon.

\begin{example}\label{ex:chainofparallel}
Consider the chain-of-parallel network given in Figure~\ref{fi:chainofparallel}, where the capacity of each edge is 1 and the transit costs are indicated on the edges. The capacity $\gamma^{(*)}$ of the network is $2$. 

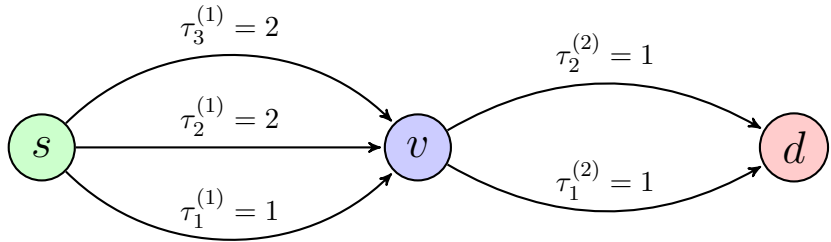
\begin{figure}[h]
\centering
\begin{tikzpicture}[->,>=stealth',shorten >=1pt,auto,node distance=5cm,
  thick,main node/.style={circle,fill=blue!20,draw,minimum size=25pt,font=\sffamily\Large\bfseries},source node/.style={circle,fill=green!20,draw,minimum size=25pt,font=\sffamily\Large\bfseries},dest node/.style={circle,fill=red!20,draw,minimum size=25pt,font=\sffamily\Large\bfseries}]

  \node[source node] (1) {$s$};
  \node[main node] (2) [right of=1] {$v$};
  \node[dest node] (3) [right of=2] {$d$};
  
  \path[every node/.style={font=\sffamily\small}]
    (1) edge [bend right = 45] node[above] {$\tau_{1}^{(1)}=1$} (2)
        edge node [above] {$\tau_{2}^{(1)}=2$} (2)
        edge [bend left = 45] node[above] {$\tau_{3}^{(1)}=2$} (2);

  \path[every node/.style={font=\sffamily\small}]
    (2) edge [bend right = 30] node[above] {$\tau_{1}^{(2)}=1$} (3)
        edge [bend left = 30] node[above] {$\tau_{2}^{(2)}=1$} (3);

\end{tikzpicture}
~\vspace{0cm} 
\caption{\label{fi:chainofparallel} Chain-of-parallel network.}
\end{figure}

Equilibria of this game are described in detail in the online Supplementary Material.

Consider the following strategy profile.
\begin{equation}\label{eq:1stequilibrium}
\sigma_{it}^{\Eq}=
\begin{cases}
e_{1}^{(1)} e_{1}^{(2)} &\text{ for } [it]=[11],\\
e_{1}^{(1)} e_{2}^{(2)} &\text{ for } [it]=[21],\\
e_{1}^{(1)} e_{1}^{(2)} &\text{ for } [it]=[1t] \text{ and } t\geq2,\\
e_{2}^{(1)} e_{2}^{(2)} &\text{ for } [it]=[2t] \text{ and } t\geq2.
\end{cases}
\end{equation}

It is easy to check that this is an equilibrium. The first player $[11]$ takes the fastest route $e_{1}^{(1)} e_{1}^{(2)}$. The second player $[21]$ cannot pay less than a total cost of 3. She does so by taking $e_{1}^{(1)}$ first and queuing after $[11]$ (a cost of 2), then taking $e_{2}^{(2)}$. This choice of the first generation leaves a queue of size 1 on edge  $e_{1}^{(1)} $ for the next generation. The next two players have to pay at least 3 each. They do so by choosing $e_{1}^{(1)} e_{1}^{(2)}$ and $e_{2}^{(1)} e_{2}^{(2)}$. The queue on edge $e_{1}^{(1)} $ is thus re-created for the next generation.

The average total latency of this equilibrium is 6. Due to the indifferences, the  same average total latency can be achieved with the following periodic equilibrium strategy profile.

\begin{equation}\label{eq:2ndequilibrium}
\tilde{\sigma}_{it}^{\Eq}=
\begin{cases}
e_{1}^{(1)} e_{1}^{(2)} &\text{ for } [it]=[1t] \text{ and } t \text{  odd},\\
e_{1}^{(1)} e_{2}^{(2)} &\text{ for } [it]=[2t] \text{ and } t \text{  odd},\\
e_{2}^{(1)} e_{1}^{(2)} &\text{ for } [it]=[1t] \text{ and } t \text{  even},\\
e_{3}^{(1)} e_{2}^{(2)} &\text{ for } [it]=[2t] \text{ and } t \text{  even}.
\end{cases}
\end{equation}

Under this profile, the second player of each odd generation creates a queue on $e_{1}^{(1)}$.  As both players of the even generation take a long route, none of these two players waits in a queue and thus the queue on $e_{1}^{(1)}$ disappears. Since the first player in the following odd generation uses the fast route $e_{1}^{(1)}$  again, she arrives at $v$ at the same time as the previous two players. Therefore, she waits in the queue on $e_{1}^{(2)}$. So, the first player of each odd generation waits in the queue on  $e_{1}^{(2)}$, and the second player waits on $e_{1}^{(1)}$ (except for the very first player). Therefore, the strategy profile $\tilde{\sigma}^{\Eq}$ yields an average latency of $6$, which is the worst equilibrium latency. However, the queues vary periodically with time (one can even exploit the indifferences to construct a more complex equilibrium where queues vary with time in an aperiodic manner).
\end{example}

Even though such periodicities can occur in an UFR equilibrium, we prove that the worst equilibrium cost can always be obtained with stationary strategies.

\section{Efficiency of equilibria and  complex topologies}\label{se:anarchy}

Efficiency of equilibria is a central issue in the theory of congestion games. Several efficiency measures have been proposed, among them, the price of anarchy, the price of stability, and the Braess ratio. Here we use these measures to show how inefficient equilibria can be for dynamic congestion games, and we look at the possible sources of inefficiencies.  We first look at parallel networks. Then, we consider more complex topologies. In this section, we consider a uniform inflow in heavy traffic, i.e., $\delta=\gamma$.

\begin{definition}
Given a game $\Gamma (\mathcal{N}, \gamma)$,
\begin{enumerate}[(a)]
\item
its \emph{price of anarchy} is defined as
\[
\PoA(\mathcal{N}, \gamma):=\frac{\WEq(\mathcal{N}, \gamma)}{\Opt(\mathcal{N}, \gamma)},
\]

\item
its \emph{price of stability}  is defined as
\[
\PoS(\mathcal{N}, \gamma):=\frac{\BEq(\mathcal{N}, \gamma)}{\Opt(\mathcal{N}, \gamma)},
\]

\item
its \emph{Braess ratio} $\BR(\mathcal{N}, \gamma)$ is defined as  the largest factor by which the removal of one or more edges can improve
the latency of traffic in an equilibrium flow.

\end{enumerate}
\end{definition}

\subsection{Parallel networks}

A direct consequence of Theorem \ref{th:parallelunif} is a computation of the price of anarchy for parallel networks.
\begin{corollary}\label{co:poaunif}
Consider the game $\Gamma (\mathcal{N}, \gamma)$, where $\mathcal{N}$ is a parallel network. Then
\[
\PoA(\mathcal{N}, \gamma)\le\frac{\max_{e}\tau_{e}}{\min_{e}\tau_{e}}.
\]
\end{corollary}
The inequality is straightforward and shows that the price of anarchy admits an upper bound the does not depend on capacities but only on the relative lengths of edges. To see that the bound is tight, consider a parallel network with two parallel edges such that the first is short and wide, $\tau_{1}=1$, $\gamma_{1}=N^{p}$, where $p\in\mathbb{N}_+$, and the second is long and narrow, $\tau_{2}=N$, $\gamma_{2}=1$.   The number of players per stage is the capacity of the network $N^{p}+1$.  This instance is similar to the classical example of Pigou where in equilibrium, the congestion on the fast edge creates a latency which matches the latency of the slow edge. The price of anarchy for this network is $(N^{p}+1)/(N^{p-1}+1)$, which is roughly $N= \max_{e}\tau_{e}/\min_{e}\tau_{e}$ for $p$ sufficiently large.

\subsection{Series-parallel networks}\label{sususe:seriesparallel}

Let $\mathcal{N}_{1}, \mathcal{N}_{2}$ be two networks with respective source-destination pairs $(s_{1},d_{1})$, 
$(s_{2},d_{2})$. The \emph{parallel composition} of $\mathcal{N}_{1}$ and $\mathcal{N}_{2}$ is the network $\mathcal{N}=\mathcal{N}_{1}\vee\mathcal{N}_{2}$ where the sources (resp. destinations) of $\mathcal{N}_{1}$  and $\mathcal{N}_{2}$ are merged together and the set of edges is the disjoint union of $E_{1}$ and $E_{2}$. A \emph{series-parallel network} is a network which can be obtained by iterated parallel and series compositions of networks, starting with a network containing only one edge.

The well-known paradox due to \citet{Bra:U1968,Bra:TS2005}  arises when adding a new edge to a network  increases the worst equilibrium latency. In static games, this paradox can only occur if the networks contains a Wheatstone subnetwork (see Figure~\ref{fi:Wheatstone}), or in other words is not series-parallel \citep{Mil:GEB2006}. \citet{MacLarSte:TCS2013} noticed that in non-atomic dynamic congestion games, Braess's paradox can arise in networks that are series-parallel.

In dynamic congestion games a different sort of Braess's paradox can arise: the presence of initial queues in the network, or increasing the transit costs of an edge may decrease the worst equilibrium latency. The following example  is an adjustment of one of the networks considered in \citet{MacLarSte:TCS2013}  .

\begin{example}\label{ex:strictineq}
Consider the series-parallel network in Figure~\ref{fi:seriesparallel} where the  associated free-flow transit costs and capacities are given. The network has two minimum cuts $\{e_{1},e_{4}\}$ and $\{e_{2},e_{3}, e_{4}\}$ with a capacity of 3, and each edge is part of one cut. Deleting one edge would cause the total cost to explode, thus Braess's paradox cannot happen in its usual form.

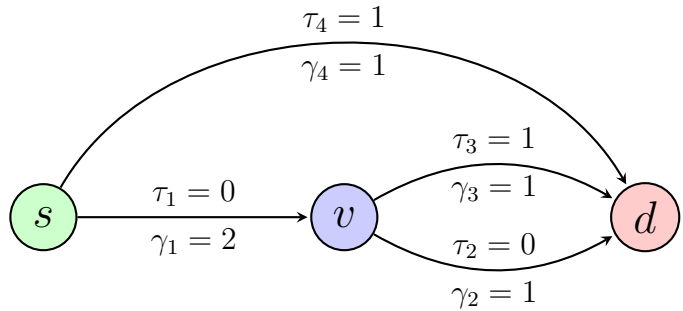
\begin{figure}[H]
\centering
\begin{tikzpicture}[->,>=stealth',shorten >=1pt,auto,node distance=4.5cm,
  thick,main node/.style={circle,fill=blue!20,draw,minimum size=25pt,font=\sffamily\Large\bfseries},source node/.style={circle,fill=green!20,draw,minimum size=25pt,font=\sffamily\Large\bfseries},dest node/.style={circle,fill=red!20,draw,minimum size=25pt,font=\sffamily\Large\bfseries}]  
  
\node[source node] (source) at (0,0) {$s$};
\node[main node] (main) at (4,0) {$v$};
\node[dest node] (dest) at (8,0) {$d$};
\draw[-stealth] (source)[out=60,in=120] to node[above]{$\tau_{4}=1$} node[below]{$\gamma_{4}=1$} (dest);
\draw[-stealth] (source) to node[above]{$\tau_{1}=0$} node[below]{$\gamma_{1}=2$} (main);
\draw[-stealth] (main)[out=30,in=150] to node[above]{$\tau_{3}=1$} node[below]{$\gamma_{3}=1$} (dest);
\draw[-stealth] (main)[out=-30,in=-150] to node[above]{$\tau_{2}=0$} node[below]{$\gamma_{2}=1$} (dest);
\end{tikzpicture}
\caption{Series-parallel network where each edge is part of a minimum cut.}
\label{fi:seriesparallel}
\end{figure}

Consider the following equilibrium strategy.
\begin{equation}\label{eq:equilMacko}
\sigma_{it}^{\Eq}=
\begin{cases}
e_{1} e_{2} &\text{ for } [it]=[11],\\
e_{1} e_{3} &\text{ for } [it]=[21],\\
e_{1} e_{2} &\text{ for } [it]=[31],\\
e_{1} e_{2} &\text{ for } [it]=[1t], t\geq2,\\
e_{4} &\text{ for } [it]=[2t], t\geq2,\\
e_{1} e_{3} &\text{ for } [it]=[3t], t\geq2.
\end{cases}
\end{equation}

To verify that this is indeed an equilibrium, the reader is referred to the online Supplementary Material. The strategy $\sigma^{\Eq}$ yields a latency of 4, which means that the last player of each generation pays more than the maximum free flow transit costs. 

If we have an initial queue of length one on $e_{2}$, or if we increase the transit cost of $e_{2}$ from $0$ to $1$, then in equilibrium one player of each generation chooses $e_{1}e_{2}$, one player chooses $e_{1}e_{3}$ and one player chooses $e_{4}$. This equilibrium is efficient and the latency is $1$ for every player: a paradox! Recall that this paradox occurs in a network for which no Braess's paradox is possible in its classical form. Related paradoxical phenomena have been studied by \citet{Dag:TS1998}, who showed that, in a model with physical queues, decreasing the capacity of an edge can improve the equilibrium flow of a network. Our example shows that this paradox can occur also with point queues that do not spill over preceding edges.

Notice that Figure~\ref{fi:seriesparallel} corresponds to a series-parallel network considered in \citet{MacLarSte:TCS2013}, where the capacity of edge $e_{3}$ is decreased. So this form of Braess' paradox may occur in all series-parallel networks considered in their work.

\end{example}

\subsection{Wheatstone networks}\label{sususe:braessnetwork}

One open question is the impact of the multiplicity of equilibria. The following example will show that for the Wheatstone network, the worst equilibrium has latency costs that are three times higher than the latency costs of the best equilibrium.

\begin{example}\label{ex:Wheatstone}
Consider the Wheatstone network in Figure~\ref{fi:Wheatstone} with associated free-flow transit costs and capacity equal to $1$ for all edges. The capacity of the network is $2$.
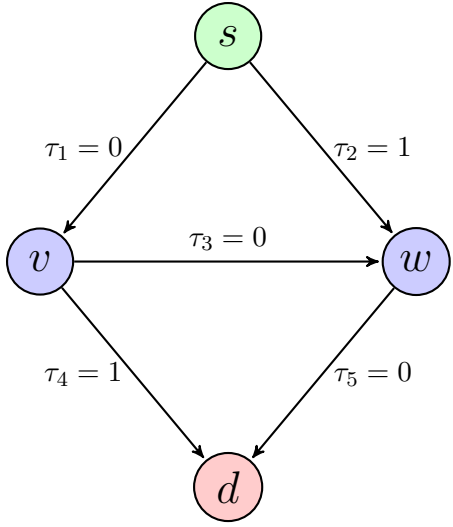
\begin{figure}[h]
\centering
\begin{tikzpicture}[->,>=stealth',shorten >=1pt,auto,node distance=4.5cm,
  thick,main node/.style={circle,fill=blue!20,draw,minimum size=25pt,font=\sffamily\Large\bfseries},source node/.style={circle,fill=green!20,draw,minimum size=25pt,font=\sffamily\Large\bfseries},dest node/.style={circle,fill=red!20,draw,minimum size=25pt,font=\sffamily\Large\bfseries}]  

  \node[source node] (1)  at (0,3) {$s$};
  \node[main node] (2) at (-2.5,0) {$v$};
  \node[dest node] (3) at (0,-3) {$d$};
  \node[main node] (4) at (2.5,0) {$w$};

  \path[every node/.style={font=\sffamily\small}]
    (1) edge node [right] {$\tau_{2}=1$} (4)
        edge node[left,color=black] {$\tau_{1}=0$}  (2)
    (2) edge node [left] {$\tau_{4}=1$} (3)
        edge node {$\tau_{3}=0$} (4)
    (4) edge node [right,color=black] {$\tau_{5}=0$}  (3);

\end{tikzpicture}
~\vspace{0cm} \caption{\label{fi:Wheatstone} Wheatstone network.}
\end{figure}

Consider the following UFR equilibrium strategy 
\begin{equation}\label{eq:Wheatstoneequil}
\sigma_{it}^{\Eq}=\begin{cases}
e_{1} e_{3} e_{5} &\text{ for } [it]=[i1],i=1,2,\\
e_{1} e_{3} e_{5} &\text{ for } [it]=[12],\\
e_{2} e_{5} &\text{ for } [it]=[22],\\
e_{1} e_{3} e_{5} &\text{ for } [it]=[13],\\
e_{1} e_{4} &\text{ for } [it]=[23],\\
e_{2} e_{5} &\text{ for } [it]=[14],\\
e_{1} e_{3} e_{5} &\text{ for } [it]=[24],\\
e_{1} e_{4} &\text{ for } [it]=[1t],t\geq5,\\
e_{2} e_{5} &\text{ for } [it]=[2t],t\geq5.
\end{cases}
\end{equation}
The strategy $\sigma^{\Eq}$ yields a latency of 6. If we remove edge $e_{3}$ from the network, then the worst equilibrium latency improves by a factor of three (from $6$ to $2$) and is equal to the optimum latency of the network. This appears even more paradoxical than in a static model, since the edge $e_{3}$ is not used in equilibrium in steady state.
\end{example}

We can generalize the results of the Wheatstone network by considering  Braess's graphs as defined by \citet{Rou:JCSS2006}. Using these graphs, we can construct an example  where multiple equilibria exist, some of them efficient and some others unboundedly bad. Details can be found in the online Supplementary Material.

\begin{proposition}\label{pr:pospoa}
For every even integer $n$, there exists a network $\mathcal{N}=(\mathcal{G},(\tau_{e})_{e\in E},(\gamma_{e})_{e\in E})$ in which $\mathcal{G}$ has $n$ vertices such that
\begin{align*}
\PoA(\mathcal{N},\gamma)&=\BR(\mathcal{N},\gamma)=n-1,\\
\PoS(\mathcal{N},\gamma)&=1.
\end{align*}
\end{proposition}

\section{Seasonal inflows on parallel network}\label{se:seasonal}

In this section, we consider an inflow sequence which is a periodic function of time: there exists an integer $K$ such that $\delta_{t+K}=\delta_t$ for all $t$. Considering parallel networks once more, we provide a characterization of optimum and equilibrium costs with periodic inflow. 

From the max-flow min-cut theorem of  \citet{ForFul:OR1958}, if the average number of players $(\sum_{k=1}^K\delta_{k})/K$ exceeds the capacity $\gamma$ of the network, then the lengths of queues on the edges of the minimum cut diverge to infinity and thus the long-run average total cost is infinite under any strategy profile. We assume from now on that $\gamma=\frac{1}{K}\sum_{k=1}^K\delta_{k}$.

Let $\mathbb{N}_{K}(\gamma)$ be the set of $K$-dimensional integer vectors $\boldsymbol{\delta}=(\delta_{1},\dots, \delta_{K})$ such that $\sum_{k=1}^{K} \delta_{k} = \gamma K$. A $K$-periodic inflow sequence will be identified with a vector $\boldsymbol{\delta}\in \mathbb{N}_{K}(\gamma)$.   
We denote $\Gamma(\mathcal{N}, K, \boldsymbol{\delta})$ the game with $K$-periodic inflow sequence given by $\boldsymbol{\delta}$.

For each integer $p$, the total latency over the period $\{pK+1,\dots, (p+1)K\}$ is
\[
L_{p}(\sigma)=\sum_{t=pK+1}^{(p+1)K}\ell_{t}(\sigma).
\]

The \emph{average total latency} over $P$ periods is
\[
\widetilde{L}_{P}(\sigma)=\frac1P\sum_{p=1}^P L_{p}(\sigma).
\]
If $\lim_{P \to \infty}\widetilde{L}_{P}(\sigma)$ exists, then it
is called \emph{asymptotic average total latency} for the strategy $\sigma$.\\

We want to capture how optimum and equilibrium costs are affected by seasonality. The main point is that at peak hours, the inflow exceeds the capacity, and therefore queues build up, even if all the flow is controlled by the planner. This is exemplified as follows.

\begin{example}\label{ex:period600}
Consider a parallel network having two edges $e_{1}, e_{2}$ each connecting the source to the destination. We assume  $\gamma_{e_{1}}=\gamma_{e_{2}}=1$, $\tau_{e_{1}}=1$ and  $\tau_{e_{2}}=2$. The capacity of the network is thus $2$. 

Consider the $3$-periodic sequence of departures  $\boldsymbol{\delta}= (6,0,0)$. Then, the following strategy profile that allocates three player per period to each edge is optimal. 
\[
\sigma^{\Opt}_{it}=
\begin{cases}
e_{1} & \text{for } i \text{ odd,}\\
e_{2} & \text{for } i \text{ even.}
\end{cases}
\]
To see it, consider the  first two players and send the first one to $e_{1}$, the second one to $e_{2}$. Then, the next two players have to queue at least for one period, so it is as if they had departed one period later and it is optimal to send one of them to $e_{1}$ and the other to $e_{2}$. Now, the remaining two players have to queue at least two periods, so it is as if they had departed two periods later, and it is again optimal to  send one of them to $e_{1}$ and the other to $e_{2}$.

The total latency over a period of time $\{1,2,3\}$ (modulo 3) is $15$, that is, $3$ times the single-period optimal total latency that we would have if departures were uniform $(2,2,2)$ plus the added cost of $6$ induced by the waiting times: two players pay an extra cost of $1$ and two players pay an extra cost of $2$.

Consider now the following equilibrium strategy $\sigma^{\Eq}$. For $t = 1$ we let
\[
\sigma^{\Eq}_{it}=
\begin{cases}
e_{1} & \text{for $i=1$ or $i$ even,} \\
e_{2} & \text{for $i>2$, odd,}
\end{cases}
\]
therefore the latencies for the first six players are
\begin{align*}
\ell_{11}(\sigma^{\Eq})&=1, \ \ell_{21}(\sigma^{\Eq})=2, \ \ell_{31}(\sigma^{\Eq})=2, \\ \ell_{41}(\sigma^{\Eq})&=3,\ \ell_{51}(\sigma^{\Eq})=3, \ \ell_{61}(\sigma^{\Eq})=4.
\end{align*}
For $t \ge 4$ and $t=1 \mod 3$, 
\[
\sigma^{\Eq}_{it}=
\begin{cases}
e_{1} & \text{for $i$ odd,}\\
e_{2} & \text{for $i$ even,}
\end{cases}
\]
and
\begin{align*}
\ell_{1t}(\sigma^{\Eq})&=2, \ \ell_{2t}(\sigma^{\Eq})=2, \ \ell_{3t}(\sigma^{\Eq})=3, \\ \ell_{4t}(\sigma^{\Eq})&=3,\ \ell_{5t}(\sigma^{\Eq})=4, \ \ell_{6t}(\sigma^{\Eq})=4.
\end{align*}
It is easy to check that this is an equilibrium. It is constructed in such a way that each player chooses $e_1$ when he is indifferent between the two edges. This choice makes it the worst equilibrium. In the steady state, the total equilibrium payoff over a $3$-period is $18$, that is $3$ times the single-period equilibrium total latency when departures are uniform $(2,2,2)$ plus the added cost of $6$ induced by the waiting times.
\end{example}

We define now a quantity that measures the non-uniformity of the inflow.
Define the following binary relation on $\mathbb{N}_{K}(\gamma)$.

\begin{definition} For any two elements $\boldsymbol{\delta}, \boldsymbol{\delta}' \in \mathbb{N}_{K}(\gamma)$, we say that \emph{$\boldsymbol{\delta}'$ is obtained from $\boldsymbol{\delta}$ by an elementary operation} (denote it $\boldsymbol{\delta}\to \boldsymbol{\delta}'$), if there exists a stage $t$  such that 
\begin{align*}
\delta_{t} &> \gamma,\\
\delta'_{t}&=\delta_{t}-1, \\
\delta'_{t+1}&=\delta_{t+1}+1, \\
\delta'_{k}&=\delta_{k} \text{ for $k\notin\{t,t+1\}$},
\end{align*}
where indices are considered modulo $K$. 
\end{definition}

Denote $\boldsymbol{\gamma}_{K} = (\gamma, \dots, \gamma) \in \mathbb{N}_{K}(\gamma)$ the uniform vector. Consider the directed graph representing the above binary relation $\to$ and denote $D(\boldsymbol{\delta})$ the distance in this graph from $\boldsymbol{\delta}$ to  $\boldsymbol{\gamma}_{K}$. An  elementary operation $\boldsymbol{\delta}\to \boldsymbol{\delta}'$ consists in moving one unit from a slot where the capacity is over-filled, to the next slot. Note that indices are considered modulo $K$, so this definition is invariant under circular permutation.  Any $\boldsymbol{\delta}\neq \boldsymbol{\gamma}_{K}$ has at least one successor in the graph and $\boldsymbol{\gamma}_{K}$ is the only element with no successor. Then, $D(\boldsymbol{\delta})$ is the minimum number of elementary operations needed to transform $\boldsymbol{\delta}$ into $\boldsymbol{\gamma}_{K}$. See Figure~\ref{fi:600to222}.

\begin{figure}[h]
  \tikzstyle{vertex}=[ball color=black!10,circle,minimum size=17pt,inner sep=0pt]
  \tikzstyle{redvertex}=[ball color=red,circle,minimum size=17pt,inner sep=0pt]
  \tikzstyle{emptyvertex}=[circle,draw=red!75,minimum size=17pt,inner sep=0pt]  
  \tikzstyle{vertextime}=[minimum size=17pt,inner sep=0pt]
  
\centering
\begin{tikzpicture}[shorten >=1pt,->]
%  \tikzstyle{vertex}=[circle,shade=black!25,minimum size=17pt,inner sep=0pt]
%  \tikzstyle{vertextime}=[minimum size=17pt,inner sep=0pt]

  \foreach \name/\x in {1/1, 2/2, 3/3}
    \node[vertextime] (T-\name) at (\x, 0) {$\name$};
  
  \foreach \name/\x in {1/1, 2/2, 3/3, 4/4, 5/5, 6/6}
    \node[vertex] (G-\name) at (1, \x) {};
    
\end{tikzpicture}
\qquad
\begin{tikzpicture}[shorten >=1pt,->]
%  \tikzstyle{vertex}=[circle,shade=black!25,minimum size=17pt,inner sep=0pt]
%  \tikzstyle{redvertex}=[circle,shade=red!50,minimum size=17pt,inner sep=0pt]
%  \tikzstyle{emptyvertex}=[circle,draw=red!75,minimum size=17pt,inner sep=0pt]  
%  \tikzstyle{vertextime}=[minimum size=17pt,inner sep=0pt]

  \foreach \name/\x in {1/1, 2/2, 3/3}
    \node[vertextime] (T-\name) at (\x, 0) {$\name$};
  
  \foreach \name/\x in {1/1, 2/2, 3/3, 4/4, 5/5}
    \node[vertex] (G-\name) at (1, \x) {};
  \foreach \name/\x in {6/6}
    \node[emptyvertex] (G-\name) at (1, \x) {};    
  \foreach \name/\x in {1/1}
    \node[redvertex] (F-\name) at (2, \x) {};  
  \draw (G-6) .. controls +(-30:1cm) .. (F-1);    
\end{tikzpicture}
\qquad
\begin{tikzpicture}[shorten >=1pt,->]
%  \tikzstyle{vertex}=[circle,fill=black!25,minimum size=17pt,inner sep=0pt]
%  \tikzstyle{redvertex}=[circle,fill=red!25,minimum size=17pt,inner sep=0pt]
%  \tikzstyle{emptyvertex}=[circle,draw=red!75,minimum size=17pt,inner sep=0pt]  
%  \tikzstyle{vertextime}=[minimum size=17pt,inner sep=0pt]

  \foreach \name/\x in {1/1, 2/2, 3/3}
    \node[vertextime] (T-\name) at (\x, 0) {$\name$};
  
  \foreach \name/\x in {1/1, 2/2, 3/3, 4/4}
    \node[vertex] (G-\name) at (1, \x) {};
  \foreach \name/\x in {5/5}
    \node[emptyvertex] (G-\name) at (1, \x) {}; 
  \foreach \name/\x in {1/1}
    \node[vertex] (F-\name) at (2, \x) {};       
  \foreach \name/\x in {2/2}
    \node[redvertex] (F-\name) at (2, \x) {};  
  \draw (G-5) .. controls +(-30:1cm) .. (F-2);    
\end{tikzpicture}
\qquad
\begin{tikzpicture}[shorten >=1pt,->]
%  \tikzstyle{vertex}=[circle,fill=black!25,minimum size=17pt,inner sep=0pt]
%  \tikzstyle{redvertex}=[circle,fill=red!25,minimum size=17pt,inner sep=0pt]
%  \tikzstyle{emptyvertex}=[circle,draw=red!75,minimum size=17pt,inner sep=0pt]  
%  \tikzstyle{vertextime}=[minimum size=17pt,inner sep=0pt]

  \foreach \name/\x in {1/1, 2/2, 3/3}
    \node[vertextime] (T-\name) at (\x, 0) {$\name$};
  
  \foreach \name/\x in {1/1, 2/2, 3/3}
    \node[vertex] (G-\name) at (1, \x) {};
  \foreach \name/\x in {4/4}
    \node[emptyvertex] (G-\name) at (1, \x) {};     
  \foreach \name/\x in {1/1, 2/2}
    \node[vertex] (F-\name) at (2, \x) {};       
   \foreach \name/\x in {3/3}
    \node[redvertex] (F-\name) at (2, \x) {};      
  \draw (G-4) .. controls +(-30:1cm) .. (F-3);    
\end{tikzpicture}

\bigskip
\bigskip

\begin{tikzpicture}[shorten >=1pt,->]
%  \tikzstyle{vertex}=[circle,fill=black!25,minimum size=17pt,inner sep=0pt]
%  \tikzstyle{redvertex}=[circle,fill=red!25,minimum size=17pt,inner sep=0pt]
%  \tikzstyle{emptyvertex}=[circle,draw=red!75,minimum size=17pt,inner sep=0pt]  
%  \tikzstyle{vertextime}=[minimum size=17pt,inner sep=0pt]

  \foreach \name/\x in {1/1, 2/2, 3/3}
    \node[vertextime] (T-\name) at (\x, 0) {$\name$};
  
  \foreach \name/\x in {1/1, 2/2, 3/3}
    \node[vertex] (G-\name) at (1, \x) {};
  \foreach \name/\x in {1/1, 2/2}
    \node[vertex] (F-\name) at (2, \x) {}; 
  \foreach \name/\x in {3/3}
    \node[emptyvertex] (F-\name) at (2, \x) {};           
  \foreach \name/\x in {1/1}
    \node[redvertex] (H-\name) at (3, \x) {};  
  \draw (F-3) .. controls +(-30:1cm) .. (H-1);    
\end{tikzpicture}
\qquad
\begin{tikzpicture}[shorten >=1pt,->]
%  \tikzstyle{vertex}=[circle,fill=black!25,minimum size=17pt,inner sep=0pt]
%  \tikzstyle{redvertex}=[circle,fill=red!25,minimum size=17pt,inner sep=0pt]
%  \tikzstyle{emptyvertex}=[circle,draw=red!75,minimum size=17pt,inner sep=0pt]  
%  \tikzstyle{vertextime}=[minimum size=17pt,inner sep=0pt]

  \foreach \name/\x in {1/1, 2/2, 3/3}
    \node[vertextime] (T-\name) at (\x, 0) {$\name$};
  
  \foreach \name/\x in {1/1, 2/2}
    \node[vertex] (G-\name) at (1, \x) {};
  \foreach \name/\x in {3/3}
    \node[emptyvertex] (G-\name) at (1, \x) {};     
  \foreach \name/\x in {1/1, 2/2}
    \node[vertex] (F-\name) at (2, \x) {};       
  \foreach \name/\x in {3/3}
    \node[redvertex] (F-\name) at (2, \x) {};  
  \foreach \name/\x in {1/1}
    \node[vertex] (H-\name) at (3, \x) {};      
  \draw (G-3) -- (F-3);    
\end{tikzpicture}
\qquad
\begin{tikzpicture}[shorten >=1pt,->]
%  \tikzstyle{vertex}=[circle,fill=black!25,minimum size=17pt,inner sep=0pt]
%  \tikzstyle{redvertex}=[circle,fill=red!25,minimum size=17pt,inner sep=0pt]
%  \tikzstyle{emptyvertex}=[circle,draw=red!75,minimum size=17pt,inner sep=0pt]  
%  \tikzstyle{vertextime}=[minimum size=17pt,inner sep=0pt]

  \foreach \name/\x in {1/1, 2/2, 3/3}
    \node[vertextime] (T-\name) at (\x, 0) {$\name$};
  
  \foreach \name/\x in {1/1, 2/2}
    \node[vertex] (G-\name) at (1, \x) {};
  \foreach \name/\x in {1/1, 2/2}
    \node[vertex] (F-\name) at (2, \x) {};  
  \foreach \name/\x in {3/3}
    \node[emptyvertex] (F-\name) at (2, \x) {};          
  \foreach \name/\x in {1/1}
    \node[vertex] (H-\name) at (3, \x) {};  
  \foreach \name/\x in {2/2}
    \node[redvertex] (H-\name) at (3, \x) {};         
  \draw (F-3) .. controls +(-30:1cm) .. (H-2);    
\end{tikzpicture}
\qquad
\begin{tikzpicture}[shorten >=1pt,->]
%  \tikzstyle{vertex}=[circle,fill=black!25,minimum size=17pt,inner sep=0pt]
%  \tikzstyle{redvertex}=[circle,fill=red!25,minimum size=17pt,inner sep=0pt]
%  \tikzstyle{emptyvertex}=[circle,draw=red!75,minimum size=17pt,inner sep=0pt]  
%  \tikzstyle{vertextime}=[minimum size=17pt,inner sep=0pt]

  \foreach \name/\x in {1/1, 2/2, 3/3}
    \node[vertextime] (T-\name) at (\x, 0) {$\name$};
  
  \foreach \name/\x in {1/1, 2/2}
    \node[vertex] (G-\name) at (1, \x) {};  
  \foreach \name/\x in {1/1, 2/2}
    \node[vertex] (F-\name) at (2, \x) {};       
  \foreach \name/\x in {1/1, 2/2}
    \node[vertex] (H-\name) at (3, \x) {};      
\end{tikzpicture}
~\vspace{0cm} \caption{\label{fi:600to222} Operations needed to transform $(6,0,0)$ into $(2,2,2)$.}
\end{figure}
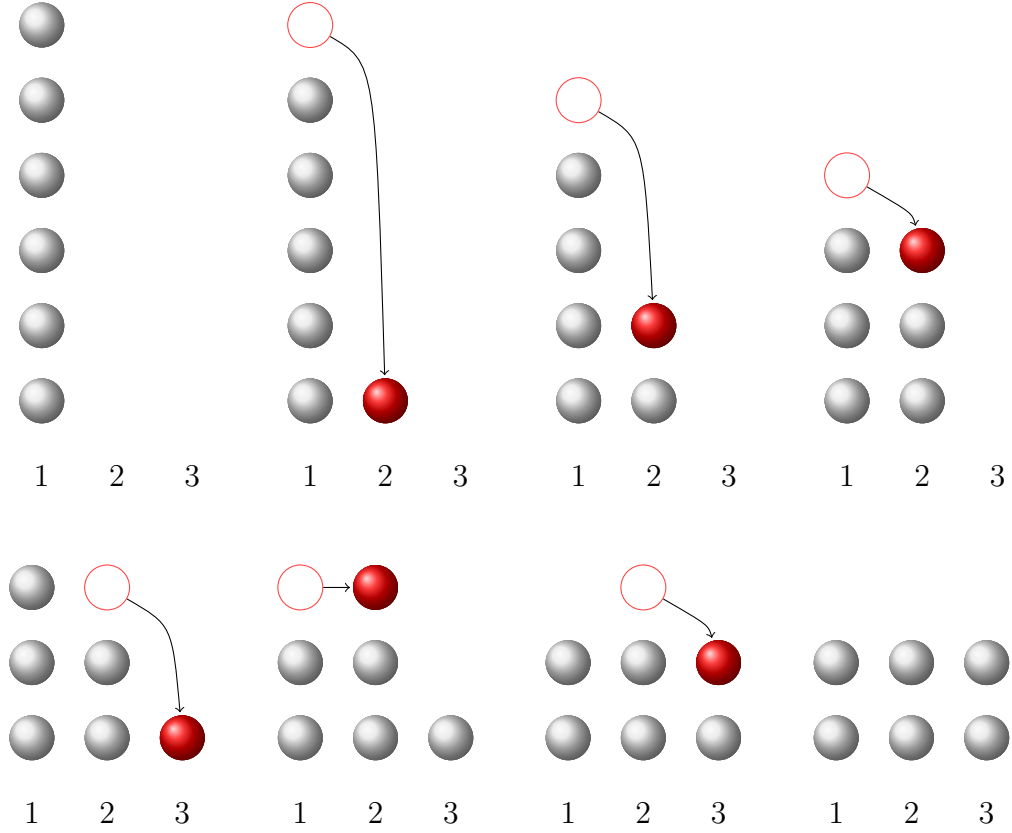

Proposition~\ref{pr:parallelperiod} below states that the quantity  $D(\boldsymbol{\delta})$ measures the total waiting time incurred at the optimum and at the worst equilibrium by the players due to non-uniform departures.

\begin{proposition}\label{pr:parallelperiod} 
Consider the game $\Gamma(\mathcal{N},K, \boldsymbol{\delta})$, where $\mathcal{N}$ is a parallel network and $\boldsymbol{\delta}\in\mathbb{N}_{K}(\gamma)$. Then
\begin{align*}
\Opt(\mathcal{N},K, \boldsymbol{\delta})&=K\sum_{e\in E}\gamma_{e}\tau_{e}+
D(\boldsymbol{\delta}),\\
\WEq(\mathcal{N}, K, \boldsymbol{\delta})&=K\gamma\max_{e\in E}\tau_{e}+D(\boldsymbol{\delta}).
\end{align*}
\end{proposition}

The main idea of the proof is that if $\boldsymbol{\delta}'$ is obtained from $\boldsymbol{\delta}$ by an elementary operation, then the optimum and the worst equilibrium under $\boldsymbol{\delta}'$ are obtained from the optimum and the worst equilibrium under $\boldsymbol{\delta}$ by letting a player postpone her departure by one unit of time. In other words, since a player departing over capacity has to queue anyhow, it would save one unit of total cost if this player's departure were postponed to the next unit of time. The formal proof is in the online Supplementary Material.

Regarding the impact on efficiency, it is easy to see that the ratio $\WEq(\mathcal{N}, K, \boldsymbol{\delta})/\Opt(\mathcal{N},K, \boldsymbol{\delta})$ is decreasing in $D(\boldsymbol{\delta})$. Intuitively, when seasonality is high, the planner has to create queues, and thus the optimum tends to resemble the equilibrium.

\section{Conclusions and open problems}\label{se:conclusion}

In this paper, we have considered dynamic congestion games with atomic players and common source-destination pair. 
We have shown that when the inflow rate is uniform over time and does not exceed the capacity of the network, an optimal dynamic flow never create queues. This result is independent of the topology of the network.
For special topologies such as parallel network, we have provided exact computation of optimum and equilibrium costs. An important insight is that optimum and equilibrium flows eventually coincide, but the transient phase before reaching the steady state induces an important difference in costs. 

We have studied efficiency of equilibria and have shown that the price of anarchy is unbounded, even for parallel networks. 
We also found that there exist networks that admit efficient equilibria, but for which both the price of anarchy and the Braess ratio are arbitrarily large. This shows that multiplicity of equilibria in atomic games may have a significant impact.

We have shown that several Braess-type paradoxes can occur in atomic dynamic network games. First of all, we have the usual Braess's paradox according to which removing an edge from a network can improve the equilibrium cost. Unlike what happens in static games, this paradox can occur also in networks that do not include a Wheatstone subnetwork. Moreover, we can have another paradoxical phenomenon for which initial queues in the network reduce the equilibrium cost. Alternatively, increasing the transit cost of an edge may reduce the equilibrium cost.

Finally, we have studied the impact of seasonality of inflow by considering parallel networks and periodic inflow sequences. The main result is that the optimum cost and the equilibrium cost are shifted upwards by the same amount which is interpreted as a measure of  seasonality.

We think of this work as  a first attempt to understand atomic dynamic congestion games. Several problems remain open, among them:

\begin{enumerate}[(a)]

\item
Are queues always bounded in equilibrium? If yes, how much worse can the equilibrium costs be, compared to the socially optimal costs, and what would be a characterization of this cost for a given network?

\item
We found a new kind of paradoxical phenomenon: the presence of an initial queue improves the equilibrium latency. In which networks does such paradox exist? And by how much can the latency improve?

\item
Based on  many examples that we have solved numerically, we conjecture that $D(\boldsymbol{\delta})$ is always an upper bound of the extra equilibrium cost due to seasonality. Is this true?
\end{enumerate}

\subsection*{Acknowledgments}

We thank three referees, an Associate Editor, and the Area Editor Professor Asuman Ozdaglar, whose insightful observations allowed us to expand the scope of our analysis. We are indebted to Ludovic Renou, who provided helpful comments and to Thomas Rivera whose suggestions helped us improve the exposition. We thank Roberto Cominetti and Jos\'e Correa for some insightful conversations.

\appendix

\section{Proofs}

\subsection*{Proofs of Section~\ref{se:model}}

\begin{proof}[Proof of Lemma~\ref{le:existence}]

We prove the existence of a uniformly fastest route equilibrium of the game $\Gamma(\mathcal{N}, \mathcal{D})$. Define the strategy profile $\sigma\in\mathcal{R}^G$ as follows. In an empty network there is always a shortest route with the property that every intermediate vertex is reached as early as possible, since that case is equivalent to the static shortest path problem. Let player $[11]$ choose a route with that property. We define the strategy for each other player $[it]$ iteratively. Given the choices of players $[js] \lhd [it]$, let player $[it]$ choose a route such that each intermediate vertex is reached as early as possible. A slight modification of Dijkstra's algorithm can be used to compute such path. Let us argue that the above strategy profile $\sigma$ is a UFR equilibrium.

By definition of $\sigma$, a player $[js]$ does not influence the costs of a player $[it]$ with $[it] \lhd [js]$, since $[js]$ does not overtake $[it]$. Hence the latency of a player $[it]$ does not depend on a player $[js]$ with $[it] \lhd [js]$. So player $[it]$ has the same latency value she had when she chose her route. Since she chose a  shortest route with the property that every intermediate vertex is reached as early as possible, the strategy profile is a UFR equilibrium.
\end{proof}

\subsection*{Proofs of Section~\ref{se:optimum}}

\begin{proof}[Proof of Theorem \ref{th:optimum}]

We start proving the theorem for the case $\delta=\gamma$, i.e., when the inflow is at capacity.
The proof starts with two lemmas. The first lemma actually holds for any $\delta \le \gamma$.

For edge $e$ and each stage $t$, denote $y_{t}^{e}(\sigma)$ the number of players who \emph{enter} edge $e$ at stage $t$ under strategy $\sigma$ and 
\[
\bar{y}_{T}^{e}(\sigma) = \frac{1}{T} \sum_{t\leq T}y_{t}^{e}(\sigma).
\]

\begin{lemma}\label{le:limsupLbar}
Let $C$ be a minimum cut and $c\in C$. If a strategy $\sigma$ is such that $\limsup_{T}\bar{y}_{T}^{c}(\sigma)>\gamma_{c}$, then $\limsup_{T}\bar{L}_{T}(\sigma)=+\infty$.
\end{lemma}

\begin{proof} 

Take a strategy $\sigma$,  a minimum cut $C$  and an edge $c\in C$ such that $\limsup_{T}\bar{y}_{T}^{c}(\sigma)>\gamma_{c}$. Then, there exists $\alpha>0$ and a subsequence $\{T_{k}\}$ such that along this subsequence we have $\bar{y}_{T_{k}}^{c}(\sigma)\geq \gamma_{c}+\alpha$. Since at most $\gamma_{c}$ players can exit edge $c$ at any given time, this implies that there exists a player who has a waiting time of $w=\lfloor \alpha T_{k}/\gamma_c\rfloor$. This in turn implies that for each integer $s< w$, there exist $\gamma_c$ players who have waiting time $s$. Thus, the total waiting time adds up to at least
\[
\gamma_c\cdot(1+\dots+w-1)=\frac{\gamma_c\cdot(w-1)\cdot w}{2}
\]
and  the average waiting time at stage $T_{k}$ is such that
\[
\bar{w}_{T_{k}}(\sigma)\geq \frac{(\lfloor \alpha T_{k}/\gamma_c\rfloor-1)\cdot\lfloor \alpha T_{k}/\gamma_c\rfloor}{2T_{k}}.
\]
The r.h.s.\ diverges as $k\to\infty$, which concludes the proof. 
\end{proof}

For each stage $T$ and edge $e$, denote $x_{T}^{e}(\sigma)$ the number of players who \emph{exit} edge $e$ at stage $T$ under strategy $\sigma$ and 
\[
\bar{x}_{T}^{e}(\sigma) = \frac{1}{T} \sum_{t\leq T}x_{t}^{e}(\sigma).
\]

\begin{lemma}\label{co:FF}
If a strategy $\sigma$ is such that  $\limsup_{T}\bar{L}_{T}(\sigma)<+\infty$, then, for any  minimum cut $C$ and $c\in C$, we have
\[
\lim_{T}\bar{y}_{T}^{c}(\sigma)=\lim_{T}\bar x_{T}^{c}(\sigma)=\gamma_{c}.
\]
\end{lemma}

\begin{proof}
Consider such a strategy $\sigma$ and a minimum cut $C$. Thanks to Lemma~\ref{le:limsupLbar}, for each edge $c\in C$, $\limsup_{T}\bar{y}_{T}^{c}(\sigma)\leq \gamma_{c}$. If there exists an edge $c\in C$ such that $\liminf_{T}\bar{y}_{T}^{c}(\sigma) < \gamma_{c}$, then there must be another edge $c'\in C$ such that $\limsup_{T}\bar{y}_{T}^{c'}(\sigma) > \gamma_{c}$. This is a consequence of the Ford and Fulkerson theorem. If for $\alpha>0$, $\bar{y}_{T_{k}}^{c}(\sigma)\leq \gamma_{c}-\alpha$ along a subsequence $\{T_{k}\}$, then there is a deficit of players on edge $c$ which has to be compensated by an an excess of players on another edge $c'$ of the minimum cut $C$. As a consequence of Lemma~\ref{le:limsupLbar}, this results in an unbounded average cost.

Finally, if the inflows satisfy  
\[
\lim_{T}\bar{y}_{T}^{c}(\sigma)=\gamma_{c},\quad \forall c\in C,
\]
then the outflows $\bar{x}_{T}^{c}(\sigma), c\in C$, satisfy it as well.
\end{proof}

We may now conclude the proof of Theorem \ref{th:optimum} when $\delta=\gamma$. Lemmas~\ref{le:limsupLbar} and \ref{co:FF} imply that, to guarantee $\limsup_{T}\bar{L}_{T}(\sigma)<+\infty$,  flows have to match capacities on every minimal cut. Thus, in order to minimize the asymptotic average latency, there should remain no queues, i.e., at the optimum the total waiting time is zero. A simple way to achieve that is to repeat a static flow with no queues at each stage. By construction, the min-cost flow $f^{*}$ has a value $\gamma$, the capacity of the network, and satisfies $f^{*}_e\leq\gamma_e$ for each edge $e$. Thus, repeating the assignment $F^{*}$ at each stage yields an asymptotic average latency $L^{*}$, which is the value of the min-cost static flow problem. This is clearly the best that can be achieved without creating queues and therefore this is the optimal asymptotic average latency. 
Notice that under this assignment \emph{all} edges are queue free, not just on the edges of the minimum cut.

Now, the case  $\delta < \gamma$ can be treated by augmenting the network with a fictitious edge $f$ of capacity $\delta$ and length $0$, whose tail is the new source and whose head is the old source. This new edge $f$ is clearly the unique minimum cut of the new network. This way, we obtain a game where the inflow is at capacity, and so we can apply the first part of the proof. Since the output  of $f$ is the input of the original network, and is constantly $\delta$, the result follows.
\end{proof}

\subsection*{Proofs of Section~\ref{se:topologies}}

\begin{proof}[Proof of Proposition~\ref{th:parallelunif}]
First, by Theorem~\ref{th:optimum}, the optimal  latency can be computed by sending the capacity number of players on each edge of the subnetwork $\mathcal{N}_{\delta}$. Hence the result follows.

Second, we show that there exists an equilibrium such that each player pays the transit cost of $\tau_{f_{\delta}}$. To simplify notation, the proof below assumes that $\delta=\gamma$. A similar proof can be given if $\delta<\gamma$. Call $n$ the number of edges in $E$.

We start by defining several times. Let $T_0=0$ and, for all $j\in\{1,\ldots,n-1\}$,
\begin{equation*}
T_j=\sum_{k=1}^j\left(\frac{\sum_{i=1}^k\gamma_i}{\delta-\sum_{i=1}^k\gamma_i}\cdot(\tau_{k+1}-\tau_k)\right).
\end{equation*}
Denote $\underline{T}_j=\left\lfloor T_j\right\rfloor$ for all $j\in\{1,\ldots,n-1\}$.

Define a strategy profile $\sigma\in\mathcal{R}^G$ for $\Gamma(\mathcal{N},\delta)$ as follows. For all $[it]\in G$, choose $e\in E$ with minimum latency, and if there are multiple edges with minimum latency, then choose among these  the first one in the order $\prec$.

We divide the proof into three parts: (i) stages $t$ with $t\leq \underline{T}_1$, (ii) stages $t$ with $\underline{T}_{j-1}<t\leq \underline{T}_j$ for $j\in\{2,\ldots,n-1\}$ and (iii) stages $t$ with $t>\underline{T}_{n-1}$. Note that part (ii) is redundant if $n=2$.

\medskip

\noindent
(i) Each player who sees a queue of size $\gamma_1\cdot(\tau_2-\tau_1)=(\delta-\gamma_1)\cdot T_1$  on $e_1$, faces a waiting cost of $\tau_2-\tau_1$, and consequently is indifferent between $e_1$ and $e_2$. If all $[it]$ with $t<\underline{T}_1$ choose $e_1$, then the queue on $e_1$ contains $(\delta-\gamma_1)\cdot \underline{T}_1$ players at the start of stage $\underline{T}_1$. Define 
\begin{equation*}
\alpha_1=(\delta-\gamma_1)\cdot(T_1- \underline{T}_1).
\end{equation*}
This  is the number of players of $G_{\underline{T}_1}$ needed before a player is indifferent between $e_1$ and $e_2$. 
Since $0\leq\alpha_1<(\delta-\gamma_1)$, we know that player $[\alpha_1+1,\underline{T}_1]$ sees a queue of size $(\delta-\gamma_1)\cdot T_1\cdot$  on $e_1$ (and is indifferent between $e_1$ and $e_2$), and that player $[\alpha_1+\gamma_1+1,\underline{T}_1]$ is the first player to choose $e_2$. In other words, at all stages $t$ with $t< \underline{T}_1$ all players choose $e_1$. 

\medskip

\noindent
(ii) For $j\in\{2,\ldots,n-1\}$, the following analysis holds true iteratively. Consider the sum of the queues on $e_1, \ldots, e_j$ that have grown  starting from the first  player who was indifferent between $e_1, \ldots, e_j$ onwards. We call this the joint queue of $e_1, \ldots, e_j$. The joint queue of $e_1, \ldots, e_j$ contains $\max\{0,\delta-\sum_{i=1}^j\gamma_i-\alpha_{j-1}\}$ players at the start of stage $\underline{T}_{j-1}+1$.

Each player who sees a queue of size $\sum_{i=1}^j\gamma_i\cdot(\tau_{j+1}-\tau_j)=(\delta-\sum_{i=1}^j\gamma_i)\cdot(T_j-T_{j-1})$  on the joint queue of $e_1, \ldots, e_j$, is indifferent between $e_1$, \ldots, $e_j$ and $e_{j+1}$. If all $[it]$ with $\underline{T}_{j-1}<t<\underline{T}_j$ choose one edge in $\{e_1, \ldots, e_j\}$, then the joint queue of $e_1, \ldots, e_j$ contains $\max\{0,\delta-\sum_{i=1}^j\gamma_i-\alpha_{j-1}\}+(\delta-\sum_{i=1}^j\gamma_i)\cdot(\underline{T}_j-\underline{T}_{j-1}-1)$ players at the start of stage $\underline{T}_j$. Define 
\begin{equation*}
\alpha_j=\left(\delta-\sum_{i=1}^j\gamma_i\right)\cdot(T_j- \underline{T}_j+\underline{T}_{j-1}+1-T_{j-1})-\max\left\{0,\delta-\sum_{i=1}^j\gamma_i-\alpha_{j-1}\right\}.
\end{equation*}
Thus $\alpha_j$ is the number of players needed in $G_{\underline{T}_j}$ before a player is indifferent between $e_1$, \ldots, $e_j$ and $e_{j+1}$.

\begin{claim}
We have
$0\leq\alpha_j<\delta-\gamma_1$ for all $j\in\{2,\ldots,n-1\}$.
\end{claim}
\begin{proof} 
First, we show that $\alpha_j\geq 0$ for all $j\in\{2,\ldots,n-1\}$. Notice that if $\delta-\sum_{i=1}^j\gamma_i-\alpha_{j-1}\leq0$, then the result follows. So assume that $\delta-\sum_{i=1}^j\gamma_i-\alpha_{j-1}>0$. We prove by induction that $\delta-\sum_{i=1}^j\gamma_i-\alpha_{j-1}\leq(\delta-\sum_{i=1}^j\gamma_i)\cdot(\underline{T}_{j-1}+1-T_{j-1})$ for all $j\in\{2,\ldots,n-1\}$.

For $j=2$,
\begin{align*}
\delta-\sum_{i=1}^2\gamma_i-\alpha_1&=\left(\delta-\sum_{i=1}^2\gamma_i\right)\cdot(\underline{T}_1+1-T_1)-\gamma_2\cdot(T_1-\underline{T}_1)\\
&\leq\left(\delta-\sum_{i=1}^2\gamma_i\right)\cdot(\underline{T}_1+1-T_1).
\end{align*} 

Suppose the inequality holds true for $j\in\{2,\ldots,n-2\}$. Then
\begin{align*}
\delta-\sum_{i=1}^{j+1}\gamma_i-\alpha_j&=\left(\delta-\sum_{i=1}^{j+1}\gamma_i\right)\cdot(\underline{T}_j-T_j+T_{j-1}-\underline{T}_{j-1})\\
&\quad-\gamma_{j+1}\cdot(T_j-\underline{T}_j+\underline{T}_{j-1}+1-T_{j-1})+\max\left\{0,\delta-\sum_{i=1}^j\gamma_i-\alpha_{j-1}\right\}\\
&\leq\left(\delta-\sum_{i=1}^{j+1}\gamma_i\right)\cdot(\underline{T}_j+1-T_j)-\gamma_{j+1}\cdot(T_j-\underline{T}_j)\\
&\leq\left(\delta-\sum_{i=1}^{j+1}\gamma_i\right)\cdot(\underline{T}_j+1-T_j),
\end{align*}
where the first inequality follows from the induction hypothesis.

The above result implies
\begin{equation*}
\alpha_j\geq\left(\delta-\sum_{i=1}^{j+1}\gamma_i\right)\cdot(T_j-\underline{T}_j)\geq 0.
\end{equation*}

Second, we show that $\alpha_j<\delta-\gamma_1$ for all $j\in\{2,\ldots,n-1\}$. We prove by induction that $\delta-\sum_{i=1}^j\gamma_i-\alpha_{j-1}\geq(\delta-\sum_{i=1}^j\gamma_i)\cdot(\underline{T}_{j-1}+1-T_{j-1})-\sum_{i=2}^j\gamma_i\cdot(T_{i-1}-\underline{T}_{i-1})$ for all $j\in\{2,\ldots,n-1\}$.

For $j=2$,
\begin{align*}
\delta-\sum_{i=1}^2\gamma_i-\alpha_1&=\left(\delta-\sum_{i=1}^2\gamma_i\right)\cdot(\underline{T}_1+1-T_1)-\gamma_2\cdot(T_1-\underline{T}_1).
\end{align*} 

Suppose the inequality holds true for $j\in\{2,\ldots,n-2\}$. Then
\begin{align*}
\delta-\sum_{i=1}^{j+1}\gamma_i-\alpha_j&=\left(\delta-\sum_{i=1}^{j+1}\gamma_i\right)\cdot(\underline{T}_j-T_j+T_{j-1}-\underline{T}_{j-1})\\
&\quad-\gamma_{j+1}\cdot(T_j-\underline{T}_j+\underline{T}_{j-1}+1-T_{j-1})+\max\left\{0,\delta-\sum_{i=1}^j\gamma_i-\alpha_{j-1}\right\}\\
&\geq\left(\delta-\sum_{i=1}^{j+1}\gamma_i\right)\cdot(\underline{T}_j+1-T_j)-\sum_{i=2}^{j+1}\gamma_i\cdot(T_{i-1}-\underline{T}_{i-1}),
\end{align*}
where the inequality follows from the induction hypothesis.

The above result implies
\begin{equation*}
\alpha_j\leq\left(\delta-\sum_{i=1}^{j+1}\gamma_i\right)\cdot(T_j-\underline{T}_j)+\sum_{i=2}^{j+1}\gamma_i\cdot(T_{i-1}-\underline{T}_{i-1})<\delta-\gamma_1.
\end{equation*}
This concludes the proof of the claim.
\end{proof}

Since $0\leq\alpha_j<\delta-\gamma_1$, we know that player $[\alpha_j+1,\underline{T}_j]$ sees a queue of size $\sum_{i=j+1}^n\gamma_i\cdot(T_j-T_{j-1})\cdot$  on the joint queue of $e_1, \ldots, e_j$ (and is indifferent between $e_1, \ldots, e_j$ and $e_{j+1}$). So, if player $[\alpha_j+\sum_{i=1}^{j}\gamma_i+1,\underline{T}_j]$ exists, then she  is the first player to choose $e_{j+1}$, and if player $[\alpha_j+\sum_{i=1}^{j}\gamma_i+1,\underline{T}_j]$ does not exist, then player $[\sum_{i=1}^j\gamma_i+1,\underline{T}_j+1]$ is the first player to choose $e_{j+1}$.

\medskip

\noindent
(iii) For all stages $t$ with $t>\underline{T}_{n-1}$, player $[1t]$ faces a latency of $\tau_n$ on each $e\in E$ and therefore is indifferent between $e_1$, \ldots, and $e_n$. So the first $\gamma_1$ players choose $e_1$, the second $\gamma_2$ players choose $e_2$, \ldots, and the last $\gamma_n$ players choose $e_n$. This implies no additional queue is created during this stage.
Since at most $\gamma$ players arrive at each stage, individual costs cannot become higher than $\tau_n$. 

\medskip

Third, notice that both in the socially optimal strategy and the worst equilibrium flows on each edge of the subnetwork $\mathcal{N}_{\delta}$ are eventually equal to capacity.
\end{proof}

%%%%%%%%%%%%%%%%%%%%%%%%%%%%%%%%%%%%%%%%%%%%
\bibliographystyle{artbibst}
\bibliography{bibdynamiccongestion}

\begin{thebibliography}{64}
\expandafter\ifx\csname natexlab\endcsname\relax\def\natexlab#1{#1}\fi
\expandafter\ifx\csname url\endcsname\relax
  \def\url#1{\texttt{#1}}\fi
\expandafter\ifx\csname urlprefix\endcsname\relax\def\urlprefix{URL }\fi

\bibitem[{Ahuja et~al.(1993)Ahuja, Magnanti, and
  Orlin}]{AhuMagOrl:Prentice1993}
\textsc{Ahuja, R.~K.}, \textsc{Magnanti, T.~L.}, and \textsc{Orlin, J.~B.}
  (1993) \emph{Network Flows. Theory, Algorithms, and Applications}.
\newblock Prentice Hall, Inc., Englewood Cliffs, NJ.

\bibitem[{Akamatsu(2000)}]{Aka:TRB2000}
\textsc{Akamatsu, T.} (2000) A dynamic traffic equilibrium assignment paradox.
\newblock \emph{Transportation Res. Part B} \textbf{34}, 515--531.
\newblock
  \newline\urlprefix\url{http://dx.doi.org/10.1016/S0191-2615(99)00036-3}.

\bibitem[{Akamatsu(2001)}]{Aka:TS2001}
\textsc{Akamatsu, T.} (2001) An efficient algorithm for dynamic traffic
  equilibrium assignment with queues.
\newblock \emph{Transportation Sci.} \textbf{35}, 389--404.
\newblock
  \newline\urlprefix\url{http://dx.doi.org/10.1287/trsc.35.4.389.10435}.

\bibitem[{Akamatsu and Heydecker(2003)}]{AkaHey:TS2003}
\textsc{Akamatsu, T.} and \textsc{Heydecker, B.} (2003) Detecting dynamic
  traffic assignment capacity paradoxes in saturated networks.
\newblock \emph{Transportation Sci.} \textbf{37}, 123--138.
\newblock
  \newline\urlprefix\url{http://dx.doi.org/10.1287/trsc.37.2.123.15245}.

\bibitem[{Anshelevich et~al.(2008)Anshelevich, Dasgupta, Kleinberg, Tardos,
  Wexler, and Roughgarden}]{AnsDasKleTarWexRou:SIAMJC2008}
\textsc{Anshelevich, E.}, \textsc{Dasgupta, A.}, \textsc{Kleinberg, J.},
  \textsc{Tardos, {\'E}.}, \textsc{Wexler, T.}, and \textsc{Roughgarden, T.}
  (2008) The price of stability for network design with fair cost allocation.
\newblock \emph{SIAM J. Comput.} \textbf{38}, 1602--1623.
\newblock \newline\urlprefix\url{http://dx.doi.org/10.1137/070680096}.

\bibitem[{Anshelevich and Ukkusuri(2009)}]{AsnUkk:AGT2009}
\textsc{Anshelevich, E.} and \textsc{Ukkusuri, S.} (2009) Equilibria in dynamic
  selfish routing.
\newblock In \textsc{Mavronicolas, M.} and \textsc{Papadopoulou, V.~G.} (eds.),
  \emph{Algorithmic Game Theory}, volume 5814, 171--182. Springer Berlin
  Heidelberg.
\newblock
  \newline\urlprefix\url{http://dx.doi.org/10.1007/978-3-642-04645-2_16}.

\bibitem[{Bhaskar et~al.(2015)Bhaskar, Fleischer, and
  Anshelevich}]{BhaFleAns:GEB2015}
\textsc{Bhaskar, U.}, \textsc{Fleischer, L.}, and \textsc{Anshelevich, E.}
  (2015) A {S}tackelberg strategy for routing flow over time.
\newblock \emph{Games Econom. Behav.} \textbf{92}, 232--247.
\newblock \newline\urlprefix\url{http://dx.doi.org/10.1016/j.geb.2013.09.004}.

\bibitem[{Bhaskar et~al.(2009)Bhaskar, Fleischer, Hoy, and
  Huang}]{BhaFleHoyHua:P20AACMSIAMSDA2009}
\textsc{Bhaskar, U.}, \textsc{Fleischer, L.}, \textsc{Hoy, D.}, and
  \textsc{Huang, C.-C.} (2009) Equilibria of atomic flow games are not unique.
\newblock In \emph{Proceedings of the {T}wentieth {A}nnual {ACM}-{SIAM}
  {S}ymposium on {D}iscrete {A}lgorithms}, 748--757. SIAM, Philadelphia, PA.

\bibitem[{Bhaskar et~al.(2010)Bhaskar, Fleischer, and
  Huang}]{BhaFleHua:IPCO2010}
\textsc{Bhaskar, U.}, \textsc{Fleischer, L.}, and \textsc{Huang, C.-C.} (2010)
  The price of collusion in series-parallel networks.
\newblock In \emph{Integer Programming and Combinatorial Optimization}, volume
  6080 of \emph{Lecture Notes in Comput. Sci.}, 313--326. Springer, Berlin.
\newblock
  \newline\urlprefix\url{http://dx.doi.org/10.1007/978-3-642-13036-6_24}.

\bibitem[{Braess(1968)}]{Bra:U1968}
\textsc{Braess, D.} (1968) \"{U}ber ein {P}aradoxon aus der {V}erkehrsplanung.
\newblock \emph{Unternehmensforschung} \textbf{12}, 258--268.

\bibitem[{Braess(2005)}]{Bra:TS2005}
\textsc{Braess, D.} (2005) On a paradox of traffic planning.
\newblock \emph{Transportation Science} \textbf{39}, 446--450.
\newblock \newline\urlprefix\url{http://dx.doi.org/10.1287/trsc.1050.0127}.
\newblock Translation of the original German (1968) article by D. Braess, A.
  Nagurney and T. Wakolbinger.

\bibitem[{Charnes and Cooper(1961)}]{ChaCoo:TTF1961}
\textsc{Charnes, A.} and \textsc{Cooper, W.~W.} (1961) {Multicopy traffic
  network models}.
\newblock In \textsc{Herman, R.} (ed.), \emph{Theory of Traffic Flow}.
  Elsevier, Amsterdam.
\newblock Proceedings of the Symposium on the Theory of Traffic Flow held at
  the General Motors Research Laboratories, Warren, MI.

\bibitem[{Cominetti(2015)}]{Com:MP2015}
\textsc{Cominetti, R.} (2015) Equilibrium routing under uncertainty.
\newblock \emph{Math. Program.} \textbf{151}, 117--151.
\newblock \newline\urlprefix\url{http://dx.doi.org/10.1007/s10107-015-0889-y}.

\bibitem[{Cominetti et~al.(2015)Cominetti, Correa, and
  Larr{\'e}}]{ComCorLar:OR2015}
\textsc{Cominetti, R.}, \textsc{Correa, J.}, and \textsc{Larr{\'e}, O.} (2015)
  Dynamic equilibria in fluid queueing networks.
\newblock \emph{Oper. Res.} \textbf{63}, 21--34.
\newblock \newline\urlprefix\url{http://dx.doi.org/10.1287/opre.2015.1348}.

\bibitem[{Cominetti et~al.(2011)Cominetti, Correa, and
  Larr{\'e}}]{ComCorLar:ALP2011}
\textsc{Cominetti, R.}, \textsc{Correa, J.~R.}, and \textsc{Larr{\'e}, O.}
  (2011) Existence and uniqueness of equilibria for flows over time.
\newblock In \emph{Automata, Languages and Programming. {P}art {II}}, volume
  6756 of \emph{Lecture Notes in Comput. Sci.}, 552--563. Springer, Heidelberg.
\newblock
  \newline\urlprefix\url{http://dx.doi.org/10.1007/978-3-642-22012-8\_44}.

\bibitem[{Correa et~al.(2004)Correa, Schulz, and
  Stier-Moses}]{CorSchSti:MOR2004}
\textsc{Correa, J.~R.}, \textsc{Schulz, A.~S.}, and \textsc{Stier-Moses, N.~E.}
  (2004) Selfish routing in capacitated networks.
\newblock \emph{Math. Oper. Res.} \textbf{29}, 961--976.
\newblock \newline\urlprefix\url{http://dx.doi.org/10.1287/moor.1040.0098}.

\bibitem[{Correa et~al.(2007)Correa, Schulz, and
  Stier-Moses}]{CorSchSti:OR2007}
\textsc{Correa, J.~R.}, \textsc{Schulz, A.~S.}, and \textsc{Stier-Moses, N.~E.}
  (2007) Fast, fair, and efficient flows in networks.
\newblock \emph{Oper. Res.} \textbf{55}, 215--225.
\newblock \newline\urlprefix\url{http://dx.doi.org/10.1287/opre.1070.0383}.

\bibitem[{Correa et~al.(2008)Correa, Schulz, and
  Stier-Moses}]{CorSchSti:GEB2008}
\textsc{Correa, J.~R.}, \textsc{Schulz, A.~S.}, and \textsc{Stier-Moses, N.~E.}
  (2008) A geometric approach to the price of anarchy in nonatomic congestion
  games.
\newblock \emph{Games Econom. Behav.} \textbf{64}, 457--469.
\newblock \newline\urlprefix\url{http://dx.doi.org/10.1016/j.geb.2008.01.001}.

\bibitem[{Correa and Stier-Moses(2010)}]{CorSti:Wiley2010}
\textsc{Correa, J.~R.} and \textsc{Stier-Moses, N.~E.} (2010) Wardrop
  equilibria.
\newblock In \textsc{Cochran, J.~J.}, \textsc{Cox, L.~A.}, \textsc{Keskinocak,
  P.}, \textsc{Kharoufeh, J.~P.}, and \textsc{Smith, J.~C.} (eds.), \emph{Wiley
  Encyclopedia of Operations Research and Management Science}. John Wiley \&
  Sons, Inc.
\newblock
  \newline\urlprefix\url{http://dx.doi.org/10.1002/9780470400531.eorms0962}.

\bibitem[{Daganzo(1998)}]{Dag:TS1998}
\textsc{Daganzo, C.~F.} (1998) Queue spillovers in transportation networks with
  a route choice.
\newblock \emph{Transportation Sci.} \textbf{32}, 3--11.
\newblock \newline\urlprefix\url{http://dx.doi.org/10.1287/trsc.32.1.3}.

\bibitem[{Farzad et~al.(2008)Farzad, Olver, and Vetta}]{FarOlvVet:CJTCS2008}
\textsc{Farzad, B.}, \textsc{Olver, N.}, and \textsc{Vetta, A.} (2008) A
  priority-based model of routing.
\newblock \emph{Chic. J. Theoret. Comput. Sci.} Article 1, 29.
\newblock \newline\urlprefix\url{http://dx.doi.org/10.4086/cjtcs.2008.001}.

\bibitem[{Fleischer and Tardos(1998)}]{FleTar:ORL1998}
\textsc{Fleischer, L.} and \textsc{Tardos, {\'E}.} (1998) Efficient
  continuous-time dynamic network flow algorithms.
\newblock \emph{Oper. Res. Lett.} \textbf{23}, 71--80.
\newblock
  \newline\urlprefix\url{http://dx.doi.org/10.1016/S0167-6377(98)00037-6}.

\bibitem[{Ford and Fulkerson(1958)}]{ForFul:OR1958}
\textsc{Ford, Jr., L.~R.} and \textsc{Fulkerson, D.~R.} (1958) Constructing
  maximal dynamic flows from static flows.
\newblock \emph{Oper. Res.} \textbf{6}, 419--433.
\newblock \newline\urlprefix\url{http://dx.doi.org/10.1287/opre.6.3.419}.

\bibitem[{Ford and Fulkerson(1962)}]{ForFul:PUP1962}
\textsc{Ford, Jr., L.~R.} and \textsc{Fulkerson, D.~R.} (1962) \emph{Flows in
  Networks}.
\newblock Princeton University Press, Princeton, N.J.

\bibitem[{Gale(1959)}]{Gal:MMJ1959}
\textsc{Gale, D.} (1959) Transient flows in Networks.
\newblock \emph{Michigan Math. J.} \textbf{6}, 59--63.
\newblock \newline\urlprefix\url{http://dx.doi.org/10.1307/mmj/1028998140}.

\bibitem[{Harker(1988)}]{Har:TS1988}
\textsc{Harker, P.~T.} (1988) Multiple equilibrium behaviors on networks.
\newblock \emph{Transportation Sci.} \textbf{22}, 39--46.
\newblock \newline\urlprefix\url{http://dx.doi.org/10.1287/trsc.22.1.39}.

\bibitem[{Harks et~al.(2009)Harks, Heinz, and Pfetsch}]{HarHeiPfe:TCS2009}
\textsc{Harks, T.}, \textsc{Heinz, S.}, and \textsc{Pfetsch, M.~E.} (2009)
  Competitive online multicommodity routing.
\newblock \emph{Theory Comput. Syst.} \textbf{45}, 533--554.
\newblock \newline\urlprefix\url{http://dx.doi.org/10.1007/s00224-009-9187-5}.

\bibitem[{Harks and V\'egh(2007)}]{HarVeg:LNCS2007}
\textsc{Harks, T.} and \textsc{V\'egh, L.~A.} (2007) Nonadaptive selfish
  routing with online demands.
\newblock In \textsc{Janssen, J.} and \textsc{Pra{\l}at, P.} (eds.),
  \emph{Combinatorial and Algorithmic Aspects of Networking}, volume 4852 of
  \emph{Lecture Notes in Computer Science}, 27--45. Springer Berlin Heidelberg.
\newblock
  \newline\urlprefix\url{http://dx.doi.org/10.1007/978-3-540-77294-1_5}.

\bibitem[{Haurie and Marcotte(1985)}]{HauMar:N1985}
\textsc{Haurie, A.} and \textsc{Marcotte, P.} (1985) On the relationship
  between {N}ash-{C}ournot and {W}ardrop equilibria.
\newblock \emph{Networks} \textbf{15}, 295--308.
\newblock \newline\urlprefix\url{http://dx.doi.org/10.1002/net.3230150303}.

\bibitem[{Hendrickson and Kocur(1981)}]{HenKoc:TS1981}
\textsc{Hendrickson, C.} and \textsc{Kocur, G.} (1981) Schedule delay and
  departure time decisions in a deterministic model.
\newblock \emph{Transportation Sci.} \textbf{15}, 62--77.
\newblock \newline\urlprefix\url{http://dx.doi.org/10.1287/trsc.15.1.62}.

\bibitem[{Hoefer et~al.(2009)Hoefer, Mirrokni, R\"{o}glin, and
  Teng}]{HoeMirRogTen:P5IWINE:2009}
\textsc{Hoefer, M.}, \textsc{Mirrokni, V.~S.}, \textsc{R\"{o}glin, H.}, and
  \textsc{Teng, S.-H.} (2009) Competitive routing over time.
\newblock In \emph{Proceedings of the 5th International Workshop on Internet
  and Network Economics}, WINE '09, 18--29. Springer-Verlag, Berlin,
  Heidelberg.
\newblock
  \newline\urlprefix\url{http://dx.doi.org/10.1007/978-3-642-10841-9_4}.

\bibitem[{Hoppe and Tardos(2000)}]{HopTar:MOR2000}
\textsc{Hoppe, B.} and \textsc{Tardos, {\'E}.} (2000) The quickest
  transshipment problem.
\newblock \emph{Math. Oper. Res.} \textbf{25}, 36--62.
\newblock \newline\urlprefix\url{http://dx.doi.org/10.1287/moor.25.1.36.15211}.

\bibitem[{Jarvis and Ratliff(1982)}]{JarRat:MS1982}
\textsc{Jarvis, J.~J.} and \textsc{Ratliff, H.~D.} (1982) Some equivalent
  objectives for dynamic network flow problems.
\newblock \emph{Management Sci.} \textbf{28}, 106--109.
\newblock \newline\urlprefix\url{http://dx.doi.org/10.1287/mnsc.28.1.106}.

\bibitem[{Koch(2012)}]{Koc:PhD2012}
\textsc{Koch, R.} (2012) Routing games over time.
\newblock Ph.D. thesis, Technische Universit\"at Berlin.
\newblock
  \newline\urlprefix\url{http://opus4.kobv.de/opus4-tuberlin/frontdoor/index/index/docId/3463}.

\bibitem[{Koch et~al.(2011)Koch, Nasrabadi, and Skutella}]{KocNasSku:MMOR2011}
\textsc{Koch, R.}, \textsc{Nasrabadi, E.}, and \textsc{Skutella, M.} (2011)
  Continuous and discrete flows over time: a general model based on measure
  theory.
\newblock \emph{Math. Methods Oper. Res.} \textbf{73}, 301--337.
\newblock \newline\urlprefix\url{http://dx.doi.org/10.1007/s00186-011-0357-2}.

\bibitem[{Koch and Skutella(2011)}]{KocSku:TCS2011}
\textsc{Koch, R.} and \textsc{Skutella, M.} (2011) Nash equilibria and the
  price of anarchy for flows over time.
\newblock \emph{Theory Comput. Syst.} \textbf{49}, 71--97.
\newblock \newline\urlprefix\url{http://dx.doi.org/10.1007/s00224-010-9299-y}.

\bibitem[{Korte and Vygen(2012)}]{KorVyg:Springer2012}
\textsc{Korte, B.} and \textsc{Vygen, J.} (2012) \emph{Combinatorial
  Optimization}.
\newblock Springer, Heidelberg, fifth edition.
\newblock \newline\urlprefix\url{http://dx.doi.org/10.1007/978-3-642-24488-9}.

\bibitem[{Koutsoupias and Papadimitriou(1999)}]{KouPap:STACS1999}
\textsc{Koutsoupias, E.} and \textsc{Papadimitriou, C.} (1999) Worst-case
  equilibria.
\newblock In \emph{S{TACS} 99 ({T}rier)}, volume 1563 of \emph{Lecture Notes in
  Comput. Sci.}, 404--413. Springer, Berlin.
\newblock \newline\urlprefix\url{http://dx.doi.org/10.1007/3-540-49116-3_38}.

\bibitem[{Macko et~al.(2013)Macko, Larson, and Steskal}]{MacLarSte:TCS2013}
\textsc{Macko, M.}, \textsc{Larson, K.}, and \textsc{Steskal, L.} (2013)
  Braess's paradox for flows over time.
\newblock \emph{Theory Comput. Syst.} \textbf{53}, 86--106.
\newblock \newline\urlprefix\url{http://dx.doi.org/10.1007/s00224-013-9462-3}.

\bibitem[{Milchtaich(2006)}]{Mil:GEB2006}
\textsc{Milchtaich, I.} (2006) Network topology and the efficiency of
  equilibrium.
\newblock \emph{Games Econom. Behav.} \textbf{57}, 321--346.
\newblock \newline\urlprefix\url{http://dx.doi.org/10.1016/j.geb.2005.09.005}.

\bibitem[{Minieka(1973)}]{Min:OR1973}
\textsc{Minieka, E.} (1973) Maximal, lexicographic, and dynamic network flows.
\newblock \emph{Operations Res.} \textbf{21}, 517--527.
\newblock \newline\urlprefix\url{http://dx.doi.org/10.1287/opre.21.2.517}.

\bibitem[{Monderer and Shapley(1996)}]{MonSha:GEB1996}
\textsc{Monderer, D.} and \textsc{Shapley, L.~S.} (1996) Potential games.
\newblock \emph{Games Econom. Behav.} \textbf{14}, 124--143.
\newblock \newline\urlprefix\url{http://dx.doi.org/10.1006/game.1996.0044}.

\bibitem[{Mounce(2006)}]{Mou:TRB2006}
\textsc{Mounce, R.} (2006) Convergence in a continuous dynamic queueing model
  for traffic networks.
\newblock \emph{Transportation Res. Part B} \textbf{40}, 779--791.
\newblock \newline\urlprefix\url{http://dx.doi.org/10.1016/j.trb.2005.10.004}.

\bibitem[{Mounce(2007)}]{Mou:TS2007}
\textsc{Mounce, R.} (2007) Convergence to equilibrium in dynamic traffic
  networks when route cost is decay monotone.
\newblock \emph{Transportation Sci.} \textbf{41}, 409--414.
\newblock \newline\urlprefix\url{http://dx.doi.org/10.1287/trsc.1070.0202}.

\bibitem[{Papadimitriou(2001)}]{Pap:PACM2001}
\textsc{Papadimitriou, C.} (2001) Algorithms, games, and the internet.
\newblock In \emph{Proceedings of the {T}hirty-{T}hird {A}nnual {ACM}
  {S}ymposium on {T}heory of {C}omputing}, 749--753. ACM, New York.
\newblock \newline\urlprefix\url{http://dx.doi.org/10.1145/380752.380883}.

\bibitem[{Philpott(1990)}]{Phi:MOR1990}
\textsc{Philpott, A.~B.} (1990) Continuous-time flows in networks.
\newblock \emph{Math. Oper. Res.} \textbf{15}, 640--661.
\newblock \newline\urlprefix\url{http://dx.doi.org/10.1287/moor.15.4.640}.

\bibitem[{Rosenthal(1973)}]{Ros:IJGT1973}
\textsc{Rosenthal, R.~W.} (1973) A class of games possessing pure-strategy
  {N}ash equilibria.
\newblock \emph{Internat. J. Game Theory} \textbf{2}, 65--67.
\newblock \newline\urlprefix\url{http://dx.doi.org/10.1007/BF01737559}.

\bibitem[{Roughgarden(2005)}]{Rou:MIT2005}
\textsc{Roughgarden, T.} (2005) \emph{Selfish Routing and the Price of
  Anarchy}.
\newblock MIT Press, Cambridge, Mass.

\bibitem[{Roughgarden(2006)}]{Rou:JCSS2006}
\textsc{Roughgarden, T.} (2006) On the severity of {B}raess's {P}aradox:
  designing networks for selfish users is hard.
\newblock \emph{J. Comput. System Sci.} \textbf{72}, 922--953.
\newblock \newline\urlprefix\url{http://dx.doi.org/10.1016/j.jcss.2005.05.009}.

\bibitem[{Roughgarden(2007)}]{Rou:AGT2007}
\textsc{Roughgarden, T.} (2007) Routing games.
\newblock In \emph{Algorithmic Game Theory}, 461--486. Cambridge Univ. Press,
  Cambridge.

\bibitem[{Roughgarden and Tardos(2002)}]{RouTar:JACM202}
\textsc{Roughgarden, T.} and \textsc{Tardos, {\'E}.} (2002) How bad is selfish
  routing?
\newblock \emph{J. ACM} \textbf{49}, 236--259 (electronic).
\newblock \newline\urlprefix\url{http://dx.doi.org/10.1145/506147.506153}.

\bibitem[{Roughgarden and Tardos(2004)}]{RouTar:GEB2004}
\textsc{Roughgarden, T.} and \textsc{Tardos, {\'E}.} (2004) Bounding the
  inefficiency of equilibria in nonatomic congestion games.
\newblock \emph{Games Econom. Behav.} \textbf{47}, 389--403.
\newblock \newline\urlprefix\url{http://dx.doi.org/10.1016/j.geb.2003.06.004}.

\bibitem[{Roughgarden and Tardos(2007)}]{RouTar:AGT2007}
\textsc{Roughgarden, T.} and \textsc{Tardos, {\'E}.} (2007) Introduction to the
  inefficiency of equilibria.
\newblock In \emph{Algorithmic Game Theory}, 443--459. Cambridge Univ. Press,
  Cambridge.

\bibitem[{Schrijver(2003)}]{Sch:Springer2003A}
\textsc{Schrijver, A.} (2003) \emph{Combinatorial Optimization. {P}olyhedra and
  Efficiency.}
\newblock Springer-Verlag, Berlin.

\bibitem[{Schulz and Stier~Moses(2003)}]{SchSti:P14SIAM2003}
\textsc{Schulz, A.~S.} and \textsc{Stier~Moses, N.} (2003) On the performance
  of user equilibria in traffic networks.
\newblock In \emph{Proceedings of the {F}ourteenth {A}nnual {ACM}-{SIAM}
  {S}ymposium on {D}iscrete {A}lgorithms ({B}altimore, {MD}, 2003)}, 86--87.
  ACM, New York.

\bibitem[{Shah and Shin(2010)}]{ShaShi:PACMSICMMCS2010}
\textsc{Shah, D.} and \textsc{Shin, J.} (2010) Dynamics in congestion games.
\newblock In \emph{Proceedings of the ACM SIGMETRICS International Conference
  on Measurement and Modeling of Computer Systems - SIGMETRICS '10}.
  Association for Computing Machinery / ACM-Sigmetrics, New York, New York.
\newblock \newline\urlprefix\url{http://dx.doi.org/10.1145/1811099.1811052}.

\bibitem[{Skutella(2009)}]{Sku:RTCO2009}
\textsc{Skutella, M.} (2009) An introduction to network flows over time.
\newblock In \emph{Research Trends in Combinatorial Optimization}, 451--482.
  Springer, Berlin.
\newblock
  \newline\urlprefix\url{http://dx.doi.org/10.1007/978-3-540-76796-1\_21}.

\bibitem[{Tekin et~al.(2012)Tekin, Liu, Southwell, Huang, and
  Ahmad}]{TekLiuSouHuaAhm:IEEEACM2012}
\textsc{Tekin, C.}, \textsc{Liu, M.}, \textsc{Southwell, R.}, \textsc{Huang,
  J.}, and \textsc{Ahmad, S. H.~A.} (2012) Atomic congestion games on graphs
  and their applications in networking.
\newblock \emph{IEEE/ACM Trans. Netw.} \textbf{20}, 1541--1552.
\newblock \newline\urlprefix\url{http://dx.doi.org/10.1109/TNET.2012.2182779}.

\bibitem[{Vickrey(1969)}]{Vic:AER1969}
\textsc{Vickrey, W.~S.} (1969) Congestion theory and transport investment.
\newblock \emph{Amer. Econ. Rev.} \textbf{59}, 251--260.
\newblock
  \newline\urlprefix\url{http://ideas.repec.org/a/aea/aecrev/v59y1969i2p251-60.html}.

\bibitem[{Wardrop(1952)}]{War:PICE1952}
\textsc{Wardrop, J.~G.} (1952) Some theoretical aspects of road traffic
  research.
\newblock In \emph{Proceedings of the Institute of Civil Engineers, Pt. II},
  volume~1, 325--378.
\newblock
  \newline\urlprefix\url{http://www.icevirtuallibrary.com/content/article/10.1680/ipeds.1952.11362}.

\bibitem[{Werth et~al.(2014)Werth, Holzhauser, and Krumke}]{WerHolKru:ORP2014}
\textsc{Werth, T.~L.}, \textsc{Holzhauser, M.}, and \textsc{Krumke, S.~O.}
  (2014) Atomic routing in a deterministic queuing model.
\newblock \emph{Oper. Res. Perspect.} \textbf{1}, 18--41.
\newblock \newline\urlprefix\url{http://dx.doi.org/10.1016/j.orp.2014.05.001}.

\bibitem[{Wilkinson(1971)}]{Wil:OR1971}
\textsc{Wilkinson, W.~L.} (1971) An algorithm for universal maximal dynamic
  flows in a network.
\newblock \emph{Operations Res.} \textbf{19}, 1602--1612.
\newblock \newline\urlprefix\url{http://dx.doi.org/10.1287/opre.19.7.1602}.

\bibitem[{Xia and Hill(2013)}]{XiaHil:IEEETCSIIEB2013}
\textsc{Xia, Y.} and \textsc{Hill, D.} (2013) Dynamic Braess's paradox in
  complex communication networks.
\newblock \emph{IEEE Transactions on Circuits and Systems II: Express Briefs}
  \textbf{60}, 172--176.
\newblock \newline\urlprefix\url{http://dx.doi.org/10.1109/TCSII.2013.2240912}.

\bibitem[{Yagar(1971)}]{Yag:TR1971}
\textsc{Yagar, S.} (1971) Dynamic traffic assignment by individual path
  minimization and queuing.
\newblock \emph{Transp. Res.} \textbf{5}, 179--196.
\newblock
  \newline\urlprefix\url{http://dx.doi.org/10.1016/0041-1647(71)90027-X}.

\end{thebibliography}

%%%%%%%%%%%%%%%%%%%%%%%%%%%%%%%%%%%%%%%%%%%%

\newpage

\section{Supplementary material}

\subsection*{Section~\ref{se:topologies}}

\begin{proof}[Proof of Theorem~\ref{th:chain}]
\noindent \emph{Optimum}.
It is clear that a minimum cut of $\mathcal{N}_{\ser}(H)$ is a subnetwork $\mathcal{N}^{(h)}$ with minimum capacity. 
Denote $\mathcal{N}^{(*)}$  such a minimum cut of $\mathcal{N}_{\ser}(H)$ and  $\gamma^{(*)}$  its capacity  which is also the capacity of the whole network $\mathcal{N}_{\ser}(H)$. If the size of each generation is $\delta=\gamma^{(*)}$, then each subnetwork $\mathcal{N}^{(h)}$ has a capacity at least $\delta$. 
Thus, the planner can choose the global flow in order to minimize the cost on each subnetwork separately, which is clearly the best achievable total cost.

\medskip

\noindent\emph{Equilibrium}.
Consider again $\mathcal{N}^{(*)}$, a minimum cut of $\mathcal{N}_{\ser}(H)$ with capacity $\gamma^{(*)}$. 
First, we show that there is an equilibrium with corresponding latency equal to the sum of the worst latencies of each module. Consider the subnetwork $\mathcal{N}^{(1)}$. 
The worst equilibrium cost on $\mathcal{N}^{(1)}$ with corresponding strategy profile is given by  Theorem~\ref{th:parallelunif}.  Now, from the structure of this equilibrium, from some time onwards, there are $\gamma^{(*)}$ players outgoing from $\mathcal{N}^{(1)}$ at each stage. 
Since the output of $\mathcal{N}^{(1)}$ is at most $\gamma^{(*)}$ in earlier stages, the long-run worst equilibrium cost for the next modules is the same as under a constant inflow of  $\gamma^{(*)}$. 
Hence the latency of this equilibrium is the sum of the worst latencies of each module.

Second, we show that the sum of the worst latencies is the worst equilibrium latency for this network. 
If the inflow of each module is constant from some point onwards, then the above equilibrium is the worst equilibrium.  The UFR property assures that within each module an equilibrium is played. Therefore, on each (sub)module, costs are at most the maximum transit costs. 

However, a module $\mathcal{N}^{(h)}$ with a capacity larger than $\gamma^{(*)}$ is able to produce a non-uniform outflow. As long as this outflow is below $\gamma^{(*)}$, the latency of the following module $\mathcal{N}^{(h+1)}$ is at most the worst latency of $\mathcal{N}^{(h+1)}$. 
Let $t^{*}$ be the first period in which the outflow of $\mathcal{N}^{(h)}$ is above $\gamma^{(*)}$. 
Each player that departs after $\gamma^{(*)}$ players already departed (potentially) faces an additional queue in $\mathcal{N}^{(h+1)}$. However, in order to obtain an outflow above $\gamma^{(*)}$, players from two different generations must leave at the same moment. 
This implies that all players from the second generation have a latency which is one unit below the latency of the first generation. So the additional queue that will be created in $\mathcal{N}^{(h+1)}$ is offset by the decrease in latency in $\mathcal{N}^{(h)}$. 
A similar idea applies to subsequent periods in which this additional queue is maintained. Hence overall the equilibrium latency cannot be worse than the sum of the worst latencies. 
\end{proof}

\noindent
\emph{Details of Example~\ref{ex:chainofparallel}.}
Consider the chain-of-parallel network given in Figure~\ref{fi:chainofparallelSM}, where the capacity of each edge is 1 and the transit costs are indicated on the edges. The capacity $\gamma^{(*)}$ of the network is $2$.

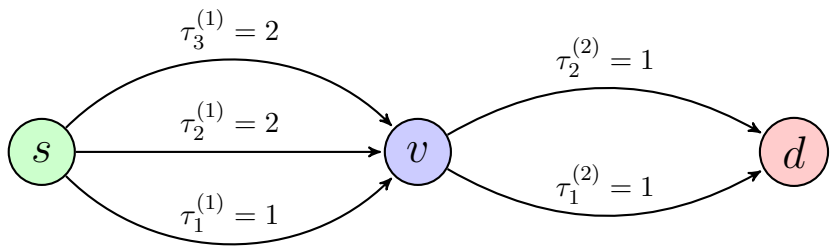
\begin{figure}[h]
\centering
\begin{tikzpicture}[->,>=stealth',shorten >=1pt,auto,node distance=5cm,
  thick,main node/.style={circle,fill=blue!20,draw,minimum size=25pt,font=\sffamily\Large\bfseries},source node/.style={circle,fill=green!20,draw,minimum size=25pt,font=\sffamily\Large\bfseries},dest node/.style={circle,fill=red!20,draw,minimum size=25pt,font=\sffamily\Large\bfseries}]

  \node[source node] (1) {$s$};
  \node[main node] (2) [right of=1] {$v$};
  \node[dest node] (3) [right of=2] {$d$};
  
  \path[every node/.style={font=\sffamily\small}]
    (1) edge [bend right = 45] node[above] {$\tau_{1}^{(1)}=1$} (2)
        edge node [above] {$\tau_{2}^{(1)}=2$} (2)
        edge [bend left = 45] node[above] {$\tau_{3}^{(1)}=2$} (2);

  \path[every node/.style={font=\sffamily\small}]
    (2) edge [bend right = 30] node[above] {$\tau_{1}^{(2)}=1$} (3)
        edge [bend left = 30] node[above] {$\tau_{2}^{(2)}=1$} (3);

\end{tikzpicture}
~\vspace{0cm}  
\caption{\label{fi:chainofparallelSM}Chain-of-parallel network.}
\end{figure}

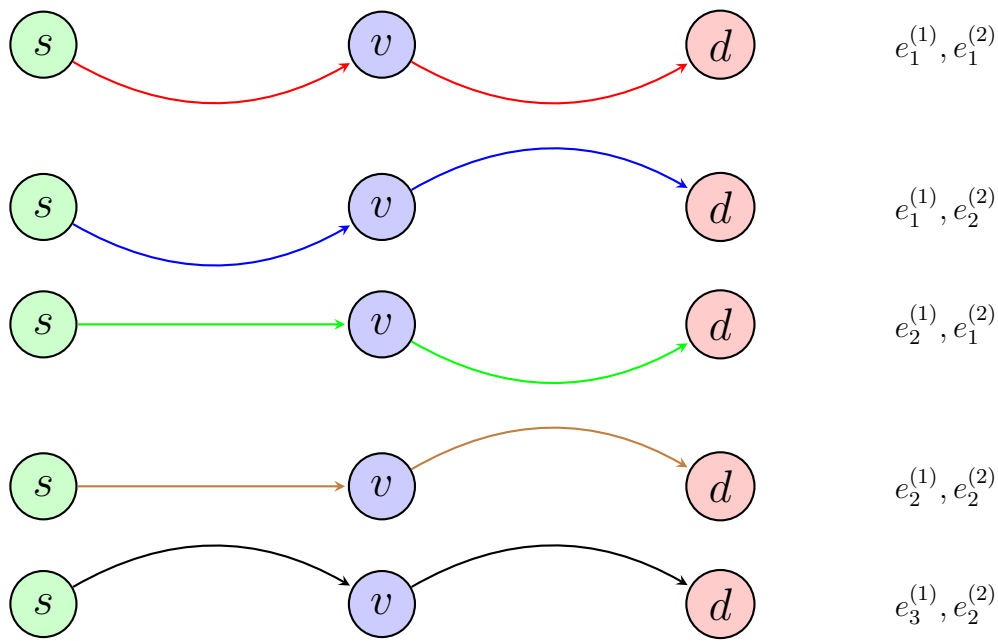
\begin{figure}[H]
\centering
\tikzstyle{main node}=[circle,fill=blue!20,draw,minimum size=25pt,font=\sffamily\Large\bfseries]
\tikzstyle{source node}=[circle,fill=green!20,draw,minimum size=25pt,font=\sffamily\Large\bfseries]
\tikzstyle{dest node}=[circle,fill=red!20,draw,minimum size=25pt,font=\sffamily\Large\bfseries]
\tikzstyle{fake node}=[circle,minimum size=22pt]
\tikzstyle{small blue node}=[blue,fill=white,font=\sffamily\bfseries]
\tikzstyle{small red node}=[red,fill=white,font=\sffamily\bfseries]
\tikzstyle{small black node}=[black,fill=white,font=\sffamily\bfseries]
\tikzstyle{small node}=[font=\sffamily\bfseries]

\begin{tikzpicture}[->,>=stealth',shorten >=1pt,auto,node distance=4.5cm,thick]  

  \node[source node] (source) {$s$};
  \node[main node] (main) [right of=source] {$v$};
  \node[dest node] (dest) [right of=main] {$d$};
  
  \draw[-stealth] (source)[red, out=-30,in=-150] to node[above]{} node[below]{} (main);
  \draw[-stealth] (main)[red, out=-30,in=-150] to node[above]{} node[below]{} (dest);
\node at (12,0) {$e_{1}^{(1)},e_{1}^{(2)}$};

\end{tikzpicture}

\begin{tikzpicture}[->,>=stealth',shorten >=1pt,auto,node distance=4.5cm,thick]  

  \node[source node] (source) {$s$};
  \node[main node] (main) [right of=source] {$v$};
  \node[dest node] (dest) [right of=main] {$d$};
  
\draw[-stealth] (source)[blue, out=-30,in=-150] to node[above]{} node[below]{} (main);
\draw[-stealth] (main)[blue, out=30,in=150] to node[above]{} node[below]{} (dest);
\node at (12,0) {$e_{1}^{(1)},e_{2}^{(2)}$};
\end{tikzpicture}

\begin{tikzpicture}[->,>=stealth',shorten >=1pt,auto,node distance=4.5cm,thick]  
  
  \node[source node] (source) {$s$};
  \node[main node] (main) [right of=source] {$v$};
  \node[dest node] (dest) [right of=main] {$d$};

\draw[-stealth] (source)[green] to node[above]{} node[below]{} (main);
\draw[-stealth] (main)[green, out=-30,in=-150] to node[above]{} node[below]{} (dest);
\node at (12,0) {$e_{2}^{(1)},e_{1}^{(2)}$};
\end{tikzpicture}

\begin{tikzpicture}[->,>=stealth',shorten >=1pt,auto,node distance=4.5cm,thick]  
  
  \node[source node] (source) {$s$};
  \node[main node] (main) [right of=source] {$v$};
  \node[dest node] (dest) [right of=main] {$d$};

\draw[-stealth] (source)[brown] to node[above]{} node[below]{} (main);
\draw[-stealth] (main)[brown, out=30,in=150] to node[above]{} node[below]{} (dest);
\node at (12,0) {$e_{2}^{(1)},e_{2}^{(2)}$};
\end{tikzpicture}

\begin{tikzpicture}[->,>=stealth',shorten >=1pt,auto,node distance=4.5cm,thick]  
  
  \node[source node] (source) {$s$};
  \node[main node] (main) [right of=source] {$v$};
  \node[dest node] (dest) [right of=main] {$d$};

\draw[-stealth] (source)[out=30,in=150] to node[above]{} node[below]{} (main);
\draw[-stealth] (main)[out=30,in=150] to node[above]{} node[below]{} (dest);
\node at (12,0) {$e_{3}^{(1)},e_{2}^{(2)}$};
\end{tikzpicture}

~\vspace{0cm} 
\caption{Route color code.}
\label{fi:chainofparallelcolorcode} 
\end{figure}

Consider the following  strategy profile.
\begin{equation}\label{eq:1stequilibriumSM}
\sigma_{it}^{\Eq}=
\begin{cases}
e_{1}^{(1)} e_{1}^{(2)} &\text{ for } [it]=[11],\\
e_{1}^{(1)} e_{2}^{(2)} &\text{ for } [it]=[21],\\
e_{1}^{(1)} e_{1}^{(2)} &\text{ for } [it]=[1t] \text{ and } t\geq2,\\
e_{2}^{(1)} e_{2}^{(2)} &\text{ for } [it]=[2t] \text{ and } t\geq2.
\end{cases}
\end{equation}

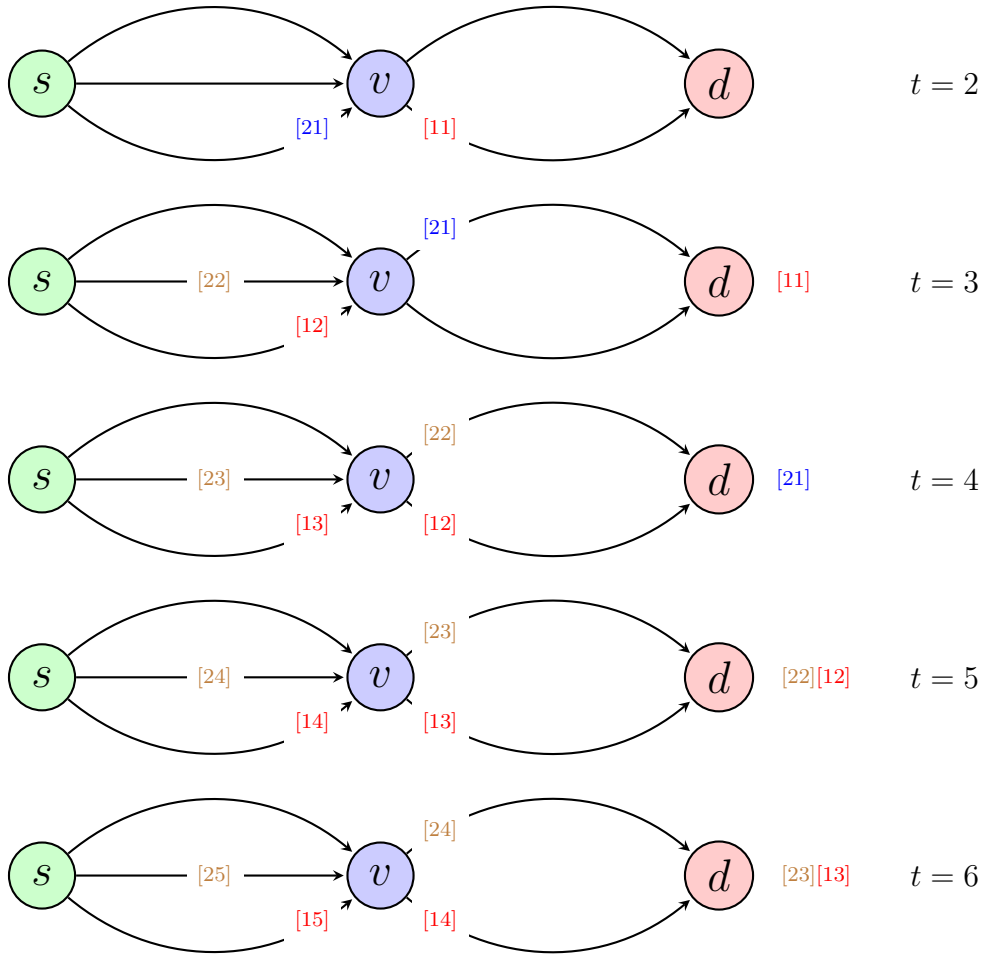
\begin{figure}[H]
\centering
\tikzstyle{main node}=[circle,fill=blue!20,draw,minimum size=25pt,font=\sffamily\Large\bfseries]
\tikzstyle{source node}=[circle,fill=green!20,draw,minimum size=25pt,font=\sffamily\Large\bfseries]
\tikzstyle{dest node}=[circle,fill=red!20,draw,minimum size=25pt,font=\sffamily\Large\bfseries]
\tikzstyle{fake node}=[circle,minimum size=22pt]
\tikzstyle{small blue node}=[blue,fill=white,font=\sffamily\bfseries]
\tikzstyle{small red node}=[red,fill=white,font=\sffamily\bfseries]
\tikzstyle{small green node}=[green,fill=white,font=\sffamily\bfseries]
\tikzstyle{small black node}=[black,fill=white,font=\sffamily\bfseries]
\tikzstyle{small brown node}=[brown,fill=white,font=\sffamily\bfseries]
\tikzstyle{small node}=[font=\sffamily\bfseries]

\begin{tikzpicture}[->,>=stealth',shorten >=1pt,auto,node distance=4.5cm,thick]  

  \node[source node] (1) {$s$};
  \node[main node] (2) [right of=1] {$v$};
  \node[dest node] (3) [right of=2] {$d$};
 
\draw[-stealth] (source)[] to node[above]{} node[below]{} (main);
\draw[-stealth] (source)[out=40,in=140] to node[above]{} node[below]{} (main);
\draw[-stealth] (source)[out=-40,in=-140] to node[above]{} node[below]{} (main);

\draw[-stealth] (main)[out=40,in=140] to node[above]{} node[below]{} (dest);  
\draw[-stealth] (main)[out=-40,in=-140] to node[above]{} node[below]{} (dest);  

\node[small red node] at (5.3,-0.6) {${}_{[11]}$};
\node[small blue node] at (3.6,-0.6) {${}_{[21]}$};

\node at (12,0) {$t=2$};

\end{tikzpicture}

\begin{tikzpicture}[->,>=stealth',shorten >=1pt,auto,node distance=4.5cm,thick]  

  \node[source node] (1) {$s$};
  \node[main node] (2) [right of=1] {$v$};
  \node[dest node] (3) [right of=2] {$d$};
 
\draw[-stealth] (source)[] to node[above]{} node[below]{} (main);
\draw[-stealth] (source)[out=40,in=140] to node[above]{} node[below]{} (main);
\draw[-stealth] (source)[out=-40,in=-140] to node[above]{} node[below]{} (main);

\draw[-stealth] (main)[out=40,in=140] to node[above]{} node[below]{} (dest);  
\draw[-stealth] (main)[out=-40,in=-140] to node[above]{} node[below]{} (dest);  

\node[small red node] at (10,0) {${}_{[11]}$};
\node[small blue node] at (5.3,0.7) {${}_{[21]}$};
\node[small red node] at (3.6,-0.6) {${}_{[12]}$};
\node[small brown node] at (2.3,0) {${}_{[22]}$};

\node at (12,0) {$t=3$};

\end{tikzpicture}

\begin{tikzpicture}[->,>=stealth',shorten >=1pt,auto,node distance=4.5cm,thick]  

  \node[source node] (1) {$s$};
  \node[main node] (2) [right of=1] {$v$};
  \node[dest node] (3) [right of=2] {$d$};
 
\draw[-stealth] (source)[] to node[above]{} node[below]{} (main);
\draw[-stealth] (source)[out=40,in=140] to node[above]{} node[below]{} (main);
\draw[-stealth] (source)[out=-40,in=-140] to node[above]{} node[below]{} (main);

\draw[-stealth] (main)[out=40,in=140] to node[above]{} node[below]{} (dest);  
\draw[-stealth] (main)[out=-40,in=-140] to node[above]{} node[below]{} (dest);  

\node[small blue node] at (10,0) {${}_{[21]}$};
\node[small red node] at (5.3,-0.6) {${}_{[12]}$};
\node[small brown node] at (5.3,0.6) {${}_{[22]}$};
\node[small red node] at (3.6,-0.6) {${}_{[13]}$};
\node[small brown node] at (2.3,0) {${}_{[23]}$};

\node at (12,0) {$t=4$};

\end{tikzpicture}

\begin{tikzpicture}[->,>=stealth',shorten >=1pt,auto,node distance=4.5cm,thick]  

  \node[source node] (1) {$s$};
  \node[main node] (2) [right of=1] {$v$};
  \node[dest node] (3) [right of=2] {$d$};
 
\draw[-stealth] (source)[] to node[above]{} node[below]{} (main);
\draw[-stealth] (source)[out=40,in=140] to node[above]{} node[below]{} (main);
\draw[-stealth] (source)[out=-40,in=-140] to node[above]{} node[below]{} (main);

\draw[-stealth] (main)[out=40,in=140] to node[above]{} node[below]{} (dest);  
\draw[-stealth] (main)[out=-40,in=-140] to node[above]{} node[below]{} (dest);  

\node[small brown node] at (10.3,0) {${}_{[22]\textcolor{red}{[12]}}$};
\node[small red node] at (5.3,-0.6) {${}_{[13]}$};
\node[small brown node] at (5.3,0.6) {${}_{[23]}$};
\node[small red node] at (3.6,-0.6) {${}_{[14]}$};
\node[small brown node] at (2.3,0) {${}_{[24]}$};

\node at (12,0) {$t=5$};

\end{tikzpicture}

\begin{tikzpicture}[->,>=stealth',shorten >=1pt,auto,node distance=4.5cm,thick]  

  \node[source node] (1) {$s$};
  \node[main node] (2) [right of=1] {$v$};
  \node[dest node] (3) [right of=2] {$d$};
 
\draw[-stealth] (source)[] to node[above]{} node[below]{} (main);
\draw[-stealth] (source)[out=40,in=140] to node[above]{} node[below]{} (main);
\draw[-stealth] (source)[out=-40,in=-140] to node[above]{} node[below]{} (main);

\draw[-stealth] (main)[out=40,in=140] to node[above]{} node[below]{} (dest);  
\draw[-stealth] (main)[out=-40,in=-140] to node[above]{} node[below]{} (dest);  

\node[small brown node] at (10.3,0) {${}_{[23]\textcolor{red}{[13]}}$};
\node[small red node] at (5.3,-0.6) {${}_{[14]}$};
\node[small brown node] at (5.3,0.6) {${}_{[24]}$};
\node[small red node] at (3.6,-0.6) {${}_{[15]}$};
\node[small brown node] at (2.3,0) {${}_{[25]}$};

\node at (12,0) {$t=6$};

\end{tikzpicture}

~\vspace{0cm} 
\caption{Example~\ref{ex:chainofparallel}, equilibrium \eqref{eq:1stequilibriumSM}}
\label{fi:chainofparallel1steq}
\end{figure}

Figure~\ref{fi:chainofparallel1steq} shows that this is an equilibrium. The first player $[11]$ takes the fastest route $e_{1}^{(1)} e_{1}^{(2)}$. The second player $[21]$ cannot pay less than a total cost of 3. She does so by taking $e_{1}^{(1)}$ first and queuing after $[11]$ (a cost of 2), then taking $e_{2}^{(2)}$. This choice of the first generation leaves a queue of size 1 on edge  $e_{1}^{(1)} $ for the next generation. The next two players have to pay at least 3 each. They do so by choosing $e_{1}^{(1)} e_{1}^{(2)}$ and $e_{2}^{(1)} e_{2}^{(2)}$. The queue on edge $e_{1}^{(1)} $ is thus re-created for the next generation.

The same average total latency of 6 can be achieved with the following periodic equilibrium strategy profile (see Figure~\ref{fi:chainofparallel2ndeq}).

\begin{equation}\label{eq:2ndequilibriumSM}
\tilde{\sigma}_{it}^{\Eq}=
\begin{cases}
e_{1}^{(1)} e_{1}^{(2)} &\text{ for } [it]=[1t] \text{ and } t \text{  odd},\\
e_{1}^{(1)} e_{2}^{(2)} &\text{ for } [it]=[2t] \text{ and } t \text{  odd},\\
e_{2}^{(1)} e_{1}^{(2)} &\text{ for } [it]=[1t] \text{ and } t \text{  even},\\
e_{3}^{(1)} e_{2}^{(2)} &\text{ for } [it]=[2t] \text{ and } t \text{  even}.
\end{cases}
\end{equation}

\begin{figure}[H]
\centering
\tikzstyle{main node}=[circle,fill=blue!20,draw,minimum size=25pt,font=\sffamily\Large\bfseries]
\tikzstyle{source node}=[circle,fill=green!20,draw,minimum size=25pt,font=\sffamily\Large\bfseries]
\tikzstyle{dest node}=[circle,fill=red!20,draw,minimum size=25pt,font=\sffamily\Large\bfseries]
\tikzstyle{fake node}=[circle,minimum size=22pt]
\tikzstyle{small blue node}=[blue,fill=white,font=\sffamily\bfseries]
\tikzstyle{small red node}=[red,fill=white,font=\sffamily\bfseries]
\tikzstyle{small green node}=[green,fill=white,font=\sffamily\bfseries]
\tikzstyle{small black node}=[black,fill=white,font=\sffamily\bfseries]
\tikzstyle{small node}=[font=\sffamily\bfseries]

\begin{tikzpicture}[->,>=stealth',shorten >=1pt,auto,node distance=4.5cm,thick]  

  \node[source node] (1) {$s$};
  \node[main node] (2) [right of=1] {$v$};
  \node[dest node] (3) [right of=2] {$d$};
 
\draw[-stealth] (source)[] to node[above]{} node[below]{} (main);
\draw[-stealth] (source)[out=40,in=140] to node[above]{} node[below]{} (main);
\draw[-stealth] (source)[out=-40,in=-140] to node[above]{} node[below]{} (main);

\draw[-stealth] (main)[out=40,in=140] to node[above]{} node[below]{} (dest);  
\draw[-stealth] (main)[out=-40,in=-140] to node[above]{} node[below]{} (dest);  

\node[small red node] at (5.3,-0.6) {${}_{[11]}$};
\node[small blue node] at (3.6,-0.6) {${}_{[21]}$};

\node at (12,0) {$t=2$};

\end{tikzpicture}

\begin{tikzpicture}[->,>=stealth',shorten >=1pt,auto,node distance=4.5cm,thick]  

  \node[source node] (1) {$s$};
  \node[main node] (2) [right of=1] {$v$};
  \node[dest node] (3) [right of=2] {$d$};
 
\draw[-stealth] (source)[] to node[above]{} node[below]{} (main);
\draw[-stealth] (source)[out=40,in=140] to node[above]{} node[below]{} (main);
\draw[-stealth] (source)[out=-40,in=-140] to node[above]{} node[below]{} (main);

\draw[-stealth] (main)[out=40,in=140] to node[above]{} node[below]{} (dest);  
\draw[-stealth] (main)[out=-40,in=-140] to node[above]{} node[below]{} (dest);  

\node[small red node] at (10,0) {${}_{[11]}$};
\node[small blue node] at (5.3,0.7) {${}_{[21]}$};
\node[small green node] at (2.3,0) {${}_{[12]}$};
\node[small black node] at (2.3,1) {${}_{[22]}$};

\node at (12,0) {$t=3$};

\end{tikzpicture}

\begin{tikzpicture}[->,>=stealth',shorten >=1pt,auto,node distance=4.5cm,thick]  

  \node[source node] (1) {$s$};
  \node[main node] (2) [right of=1] {$v$};
  \node[dest node] (3) [right of=2] {$d$};
 
\draw[-stealth] (source)[] to node[above]{} node[below]{} (main);
\draw[-stealth] (source)[out=40,in=140] to node[above]{} node[below]{} (main);
\draw[-stealth] (source)[out=-40,in=-140] to node[above]{} node[below]{} (main);

\draw[-stealth] (main)[out=40,in=140] to node[above]{} node[below]{} (dest);  
\draw[-stealth] (main)[out=-40,in=-140] to node[above]{} node[below]{} (dest);  

\node[small green node] at (5.5,-0.6) {${}_{\textcolor{red}{[13]}[12]}$};
\node[small blue node] at (10,0) {${}_{[21]}$};
\node[small black node] at (5.3,0.7) {${}_{[22]}$};
\node[small red node] at (3.6,-0.6) {${}_{[13]}$};
\node[small blue node] at (3.6,-0.6) {${}_{[23]}$};

\node at (12,0) {$t=4$};

\end{tikzpicture}

\begin{tikzpicture}[->,>=stealth',shorten >=1pt,auto,node distance=4.5cm,thick]  

  \node[source node] (1) {$s$};
  \node[main node] (2) [right of=1] {$v$};
  \node[dest node] (3) [right of=2] {$d$};
 
\draw[-stealth] (source)[] to node[above]{} node[below]{} (main);
\draw[-stealth] (source)[out=40,in=140] to node[above]{} node[below]{} (main);
\draw[-stealth] (source)[out=-40,in=-140] to node[above]{} node[below]{} (main);

\draw[-stealth] (main)[out=40,in=140] to node[above]{} node[below]{} (dest);  
\draw[-stealth] (main)[out=-40,in=-140] to node[above]{} node[below]{} (dest);  

\node[small black node] at (10.3,0) {${}_{[22]\textcolor{green}{[12]}}$};
\node[small red node] at (8.1,-0.7) {${}_{[13]}$};
\node[small blue node] at (5.3,0.7) {${}_{[23]}$};
\node[small green node] at (2.3,0) {${}_{[14]}$};
\node[small black node] at (2.3,1) {${}_{[24]}$};

\node at (12,0) {$t=5$};

\end{tikzpicture}

\begin{tikzpicture}[->,>=stealth',shorten >=1pt,auto,node distance=4.5cm,thick]  

  \node[source node] (1) {$s$};
  \node[main node] (2) [right of=1] {$v$};
  \node[dest node] (3) [right of=2] {$d$};
 
\draw[-stealth] (source)[] to node[above]{} node[below]{} (main);
\draw[-stealth] (source)[out=40,in=140] to node[above]{} node[below]{} (main);
\draw[-stealth] (source)[out=-40,in=-140] to node[above]{} node[below]{} (main);

\draw[-stealth] (main)[out=40,in=140] to node[above]{} node[below]{} (dest);  
\draw[-stealth] (main)[out=-40,in=-140] to node[above]{} node[below]{} (dest);  

\node[small blue node] at (10.3,0) {${}_{[23]\textcolor{red}{[13]}}$};
\node[small green node] at (5.5,-0.6) {${}_{\textcolor{red}{[15]}[14]}$};
\node[small black node] at (5.3,0.7) {${}_{[24]}$};
\node[small blue node] at (3.6,-0.6) {${}_{[25]}$};

\node at (12,0) {$t=6$};

\end{tikzpicture}

\begin{tikzpicture}[->,>=stealth',shorten >=1pt,auto,node distance=4.5cm,thick]  

  \node[source node] (1) {$s$};
  \node[main node] (2) [right of=1] {$v$};
  \node[dest node] (3) [right of=2] {$d$};
 
\draw[-stealth] (source)[] to node[above]{} node[below]{} (main);
\draw[-stealth] (source)[out=40,in=140] to node[above]{} node[below]{} (main);
\draw[-stealth] (source)[out=-40,in=-140] to node[above]{} node[below]{} (main);

\draw[-stealth] (main)[out=40,in=140] to node[above]{} node[below]{} (dest);  
\draw[-stealth] (main)[out=-40,in=-140] to node[above]{} node[below]{} (dest);  

\node[small black node] at (10.3,0) {${}_{[24]\textcolor{green}{[14]}}$};
\node[small red node] at (8.1,-0.7) {${}_{[15]}$};
\node[small blue node] at (5.3,0.7) {${}_{[25]}$};
\node[small green node] at (2.3,0) {${}_{[16]}$};
\node[small black node] at (2.3,1) {${}_{[26]}$};

\node at (12,0) {$t=7$};

\end{tikzpicture}

~\vspace{0cm} 
\caption{Example~\ref{ex:chainofparallel}, equilibrium \eqref{eq:2ndequilibriumSM}}
\label{fi:chainofparallel2ndeq}
\end{figure}
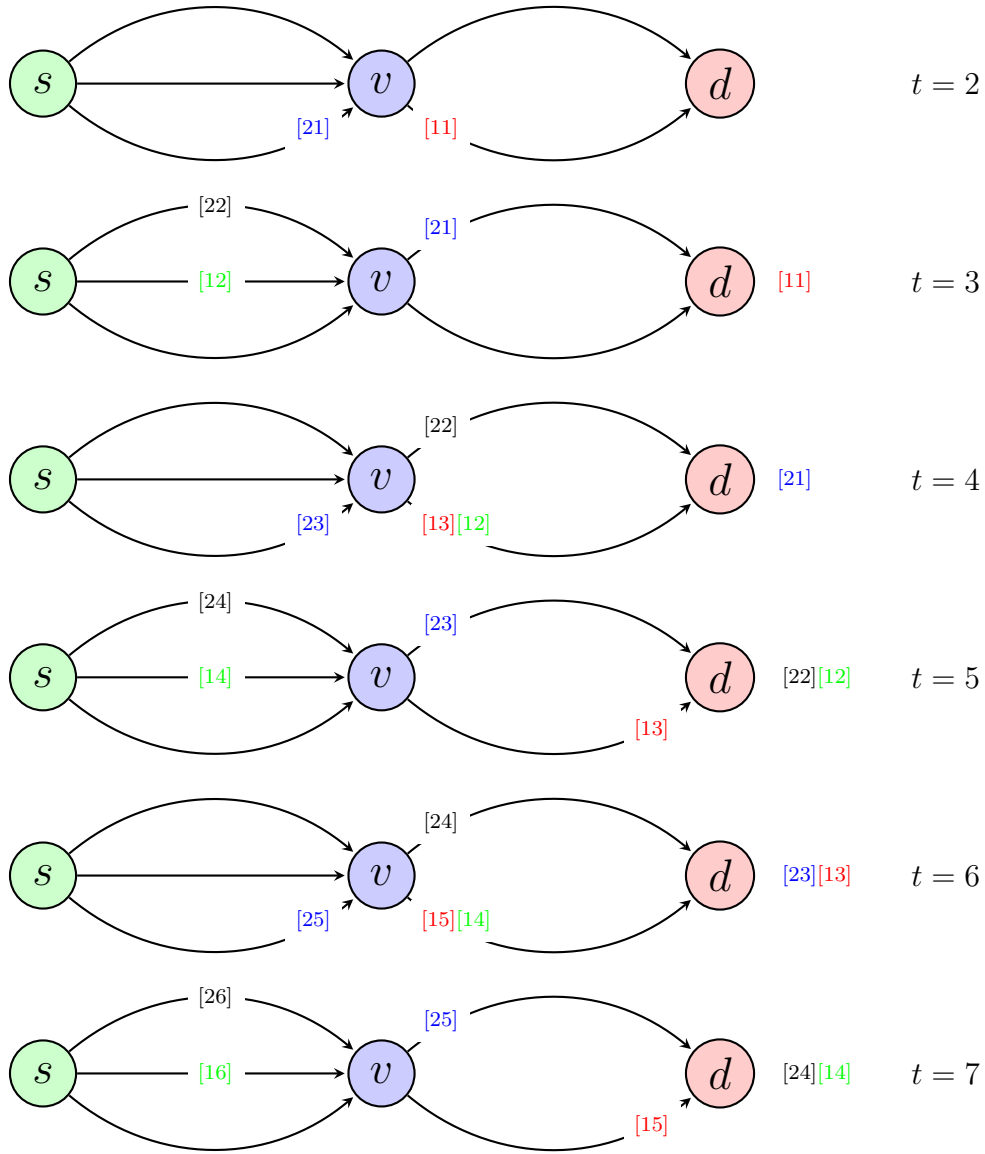

Under this profile, the second player of each odd generation creates a queue on $e_{1}^{(1)}$.  As both players of the even generation take a long route, none of these two players waits in a queue and thus the queue on $e_{1}^{(1)}$ disappears. Since the first player in the following odd generation uses the fast route $e_{1}^{(1)}$  again, she arrives at $v$ at the same time as the previous two players. Therefore, she waits in the queue on $e_{1}^{(2)}$. So, the first player of each odd generation waits in the queue on  $e_{1}^{(2)}$, and the second player waits on $e_{1}^{(1)}$ (except for the very first player).

\subsection*{Section~\ref{se:anarchy}}

\bigskip
\noindent
\emph{Details of Example~\ref{ex:strictineq}.}
Consider the series-parallel network in Figure~\ref{fi:seriesparallelSM} where the  associated free-flow transit costs and capacities are given. The network has two minimum cuts $\{e_{1},e_{4}\}$ and $\{e_{2},e_{3}, e_{4}\}$ with a capacity of 3, and each edge is part of one cut.

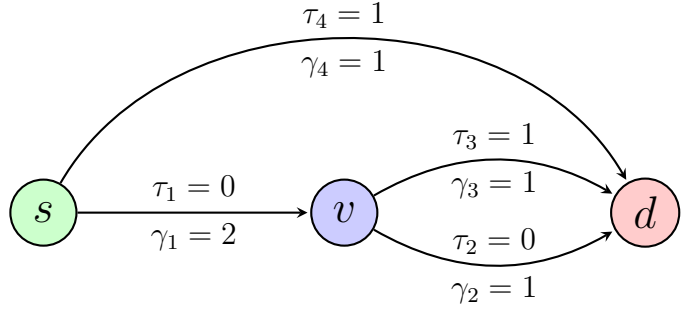
\begin{figure}[H]
\centering
\begin{tikzpicture}[->,>=stealth',shorten >=1pt,auto,node distance=4.5cm,
  thick,main node/.style={circle,fill=blue!20,draw,minimum size=25pt,font=\sffamily\Large\bfseries},source node/.style={circle,fill=green!20,draw,minimum size=25pt,font=\sffamily\Large\bfseries},dest node/.style={circle,fill=red!20,draw,minimum size=25pt,font=\sffamily\Large\bfseries}]  
  
\node[source node] (source) at (0,0) {$s$};
\node[main node] (main) at (4,0) {$v$};
\node[dest node] (dest) at (8,0) {$d$};
\draw[-stealth] (source)[out=60,in=120] to node[above]{$\tau_{4}=1$} node[below]{$\gamma_{4}=1$} (dest);
\draw[-stealth] (source) to node[above]{$\tau_{1}=0$} node[below]{$\gamma_{1}=2$} (main);
\draw[-stealth] (main)[out=30,in=150] to node[above]{$\tau_{3}=1$} node[below]{$\gamma_{3}=1$} (dest);
\draw[-stealth] (main)[out=-30,in=-150] to node[above]{$\tau_{2}=0$} node[below]{$\gamma_{2}=1$} (dest);
\end{tikzpicture}
\caption{Series-parallel network where each cut is a minimum cut.}
\label{fi:seriesparallelSM}
\end{figure}

Consider the following equilibrium strategy.
\begin{equation}\label{eq:equilMackoSM}
\sigma_{it}^{\Eq}=
\begin{cases}
e_{1} e_{2} &\text{ for } [it]=[11],\\
e_{1} e_{3} &\text{ for } [it]=[21],\\
e_{1} e_{2} &\text{ for } [it]=[31],\\
e_{1} e_{2} &\text{ for } [it]=[1t], t\geq2,\\
e_{4} &\text{ for } [it]=[2t], t\geq2,\\
e_{1} e_{3} &\text{ for } [it]=[3t], t\geq2.
\end{cases}
\end{equation}
To verify that this is indeed an equilibrium, the reader is referred to Figures~\ref{fi:seriesparallelcolorcode} and \ref{fi:Macko}.

\begin{figure}[H]
\centering
\tikzstyle{main node}=[circle,fill=blue!20,draw,minimum size=25pt,font=\sffamily\Large\bfseries]
\tikzstyle{source node}=[circle,fill=green!20,draw,minimum size=25pt,font=\sffamily\Large\bfseries]
\tikzstyle{dest node}=[circle,fill=red!20,draw,minimum size=25pt,font=\sffamily\Large\bfseries]
\tikzstyle{fake node}=[circle,minimum size=22pt]
\tikzstyle{small blue node}=[blue,fill=white,font=\sffamily\bfseries]
\tikzstyle{small red node}=[red,fill=white,font=\sffamily\bfseries]
\tikzstyle{small black node}=[black,fill=white,font=\sffamily\bfseries]
\tikzstyle{small node}=[font=\sffamily\bfseries]

\begin{tikzpicture}[->,>=stealth',shorten >=1pt,auto,node distance=4.5cm,thick]  
  
\node[source node] (source) at (0,0) {$s$};
\node[main node] (main) at (4,0) {$v$};
\node[dest node] (dest) at (8,0) {$d$};
\node[fake node] (source plus) at (0,0.1) {};
\node[fake node] (source minus) at (0,-0.1) {};
\node[fake node] (main plus) at (4,0.1) {};
\node[fake node] (main minus) at (4,-0.1) {};

\draw[-stealth] (source)[red] to node[above]{} node[below]{} (main);
\draw[-stealth] (main)[red, out=-30,in=-150] to node[above]{} node[below]{} (dest);
\node at (12,0) {$e_{1},e_{2}$};
\end{tikzpicture}

\begin{tikzpicture}[->,>=stealth',shorten >=1pt,auto,node distance=4.5cm,thick]  
  
\node[source node] (source) at (0,0) {$s$};
\node[main node] (main) at (4,0) {$v$};
\node[dest node] (dest) at (8,0) {$d$};
\node[fake node] (source plus) at (0,0.1) {};
\node[fake node] (source minus) at (0,-0.1) {};
\node[fake node] (main plus) at (4,0.1) {};
\node[fake node] (main minus) at (4,-0.1) {};

\draw[-stealth] (source)[blue] to node[above]{} node[below]{} (main);

\draw[-stealth] (main)[blue, out=30,in=150] to node[above]{} node[below]{} (dest);

\node at (12,0) {$e_{1},e_{3}$};
\end{tikzpicture}
\begin{tikzpicture}[->,>=stealth',shorten >=1pt,auto,node distance=4.5cm,thick]  
  
\node[source node] (source) at (0,0) {$s$};
\node[main node] (main) at (4,0) {$v$};
\node[dest node] (dest) at (8,0) {$d$};
\node[fake node] (source plus) at (0,0.1) {};
\node[fake node] (source minus) at (0,-0.1) {};
\node[fake node] (main plus) at (4,0.1) {};
\node[fake node] (main minus) at (4,-0.1) {};

\draw[-stealth] (source)[out=60,in=120] to node[above]{} node[below]{} (dest);

\node at (12,0) {$e_{4}$};
\end{tikzpicture}

\caption{\label{fi:seriesparallelcolorcode}Route color code}

\end{figure}
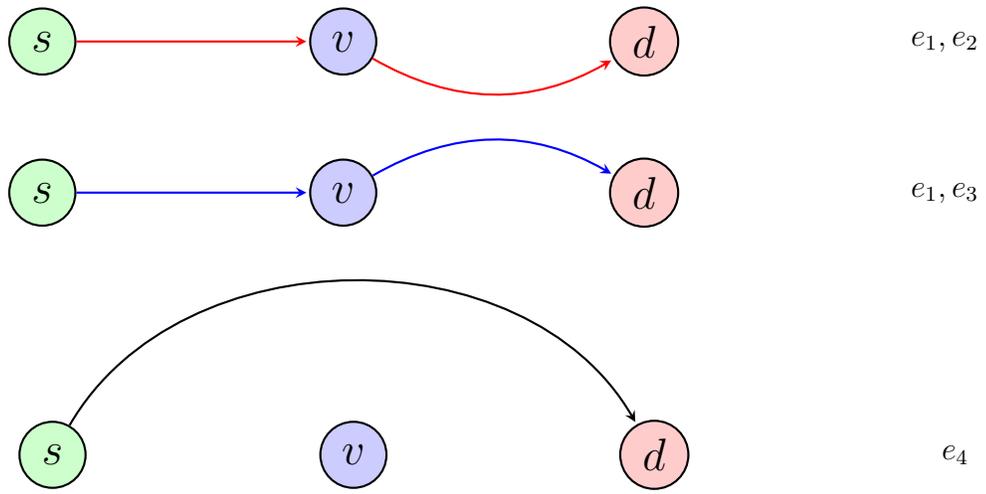

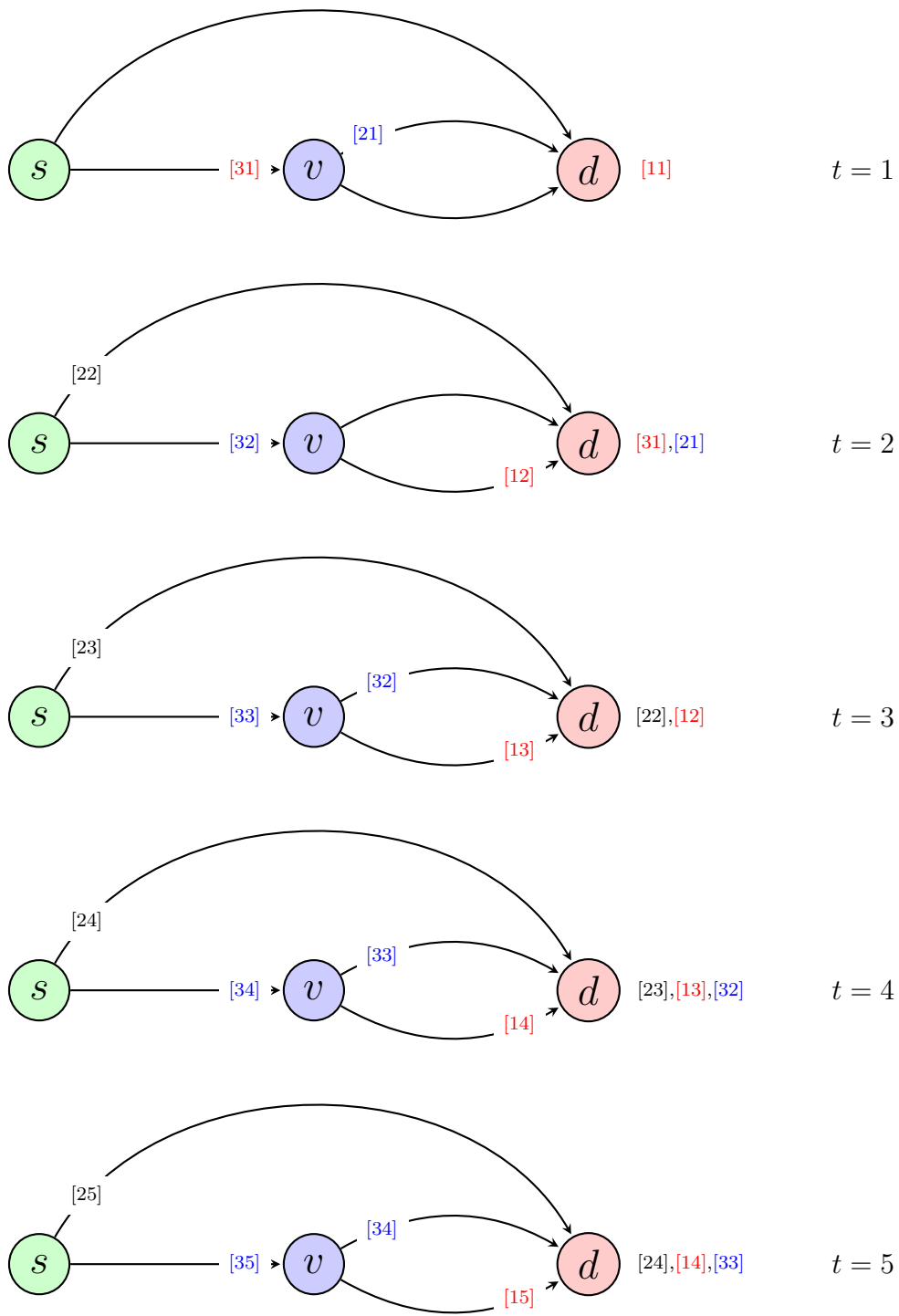
\begin{figure}[H]
\centering
\tikzstyle{main node}=[circle,fill=blue!20,draw,minimum size=25pt,font=\sffamily\Large\bfseries]
\tikzstyle{source node}=[circle,fill=green!20,draw,minimum size=25pt,font=\sffamily\Large\bfseries]
\tikzstyle{dest node}=[circle,fill=red!20,draw,minimum size=25pt,font=\sffamily\Large\bfseries]
\tikzstyle{fake node}=[circle,minimum size=22pt]
\tikzstyle{small blue node}=[blue,fill=white,font=\sffamily\bfseries]
\tikzstyle{small red node}=[red,fill=white,font=\sffamily\bfseries]
\tikzstyle{small black node}=[black,fill=white,font=\sffamily\bfseries]
\tikzstyle{small node}=[font=\sffamily\bfseries]

\begin{tikzpicture}[->,>=stealth',shorten >=1pt,auto,node distance=4.5cm,thick]  
  
\node[source node] (source) at (0,0) {$s$};
\node[main node] (main) at (4,0) {$v$};
\node[dest node] (dest) at (8,0) {$d$};
\draw[-stealth] (source)[out=60,in=120] to node[above]{} node[below]{} (dest);
\draw[-stealth] (source) to node[above]{} node[below]{} (main);
\draw[-stealth] (main)[out=30,in=150] to node[above]{} node[below]{} (dest);
\draw[-stealth] (main)[out=-30,in=-150] to node[above]{} node[below]{} (dest);
\node[small red node] at (3,0) {${}_{[31]}$};
\node[small blue node] at (4.8,.5) {${}_{[21]}$};
\node[small red node] at (9,0) {${}_{[11]}$};
\node at (12,0) {$t=1$};
\end{tikzpicture}

\begin{tikzpicture}[->,>=stealth',shorten >=1pt,auto,node distance=4.5cm,thick]  
\node[source node] (source) at (0,0) {$s$};
\node[main node] (main) at (4,0) {$v$};
\node[dest node] (dest) at (8,0) {$d$};
\draw[-stealth] (source)[out=60,in=120] to node[above]{} node[below]{} (dest);
\draw[-stealth] (source) to node[above]{} node[below]{} (main);
\draw[-stealth] (main)[out=30,in=150] to node[above]{} node[below]{} (dest);
\draw[-stealth] (main)[out=-30,in=-150] to node[above]{} node[below]{} (dest);
\node[small blue node] at (3,0) {${}_{[32]}$};
\node[small black node] at (0.7,1) {${}_{[22]}$};
\node[small red node] at (7,-.5) {${}_{[12]}$};
\node[small node] at (9.2,0) {${}_{\textcolor{red}{[31]},\textcolor{blue}{[21]}}$};
\node at (12,0) {$t=2$};
\end{tikzpicture}

\begin{tikzpicture}[->,>=stealth',shorten >=1pt,auto,node distance=4.5cm,thick]   
\node[source node] (source) at (0,0) {$s$};
\node[main node] (main) at (4,0) {$v$};
\node[dest node] (dest) at (8,0) {$d$};
\draw[-stealth] (source)[out=60,in=120] to node[above]{} node[below]{} (dest);
\draw[-stealth] (source) to node[above]{} node[below]{} (main);
\draw[-stealth] (main)[out=30,in=150] to node[above]{} node[below]{} (dest);
\draw[-stealth] (main)[out=-30,in=-150] to node[above]{} node[below]{} (dest);
\node[small blue node] at (3,0) {${}_{[33]}$};
\node[small black node] at (0.7,1) {${}_{[23]}$};
\node[small red node] at (7,-.5) {${}_{[13]}$};
\node[small blue node] at (5,.5) {${}_{[32]}$};
\node[small node] at (9.2,0) {${}_{[22],\textcolor{red}{[12]}}$};
\node at (12,0) {$t=3$};
\end{tikzpicture}

\begin{tikzpicture}[->,>=stealth',shorten >=1pt,auto,node distance=4.5cm,thick]  
  
\node[source node] (source) at (0,0) {$s$};
\node[main node] (main) at (4,0) {$v$};
\node[dest node] (dest) at (8,0) {$d$};
\draw[-stealth] (source)[out=60,in=120] to node[above]{} node[below]{} (dest);
\draw[-stealth] (source) to node[above]{} node[below]{} (main);
\draw[-stealth] (main)[out=30,in=150] to node[above]{} node[below]{} (dest);
\draw[-stealth] (main)[out=-30,in=-150] to node[above]{} node[below]{} (dest);
\node[small blue node] at (3,0) {${}_{[34]}$};
\node[small black node] at (0.7,1) {${}_{[24]}$};
\node[small red node] at (7,-.5) {${}_{[14]}$};
\node[small blue node] at (5,.5) {${}_{[33]}$};
\node[small node] at (9.5,0) {${}_{[23],\textcolor{red}{[13]},\textcolor{blue}{[32]}}$};
\node at (12,0) {$t=4$};
\end{tikzpicture}

\begin{tikzpicture}[->,>=stealth',shorten >=1pt,auto,node distance=4.5cm,thick]  
  
\node[source node] (source) at (0,0) {$s$};
\node[main node] (main) at (4,0) {$v$};
\node[dest node] (dest) at (8,0) {$d$};
\draw[-stealth] (source)[out=60,in=120] to node[above]{} node[below]{} (dest);
\draw[-stealth] (source) to node[above]{} node[below]{} (main);
\draw[-stealth] (main)[out=30,in=150] to node[above]{} node[below]{} (dest);
\draw[-stealth] (main)[out=-30,in=-150] to node[above]{} node[below]{} (dest);
\node[small blue node] at (3,0) {${}_{[35]}$};
\node[small black node] at (0.7,1) {${}_{[25]}$};
\node[small red node] at (7,-.5) {${}_{[15]}$};
\node[small blue node] at (5,.5) {${}_{[34]}$};
\node[small node] at (9.5,0) {${}_{[24],\textcolor{red}{[14]},\textcolor{blue}{[33]}}$};
\node at (12,0) {$t=5$};
\end{tikzpicture} 

\caption{Equilibrium~\eqref{eq:equilMackoSM}.}
\label{fi:Macko}
\end{figure}

\bigskip
\noindent
\emph{Details of Example~\ref{ex:Wheatstone}.}
Consider the Wheatstone network in Figure~\ref{fi:WheatstoneSM} with associated free-flow transit costs and capacity equal to $1$ for all edges. The capacity of the network is $2$.
\begin{figure}[h]
\centering
\begin{tikzpicture}[->,>=stealth',shorten >=1pt,auto,node distance=4.5cm,
  thick,main node/.style={circle,fill=blue!20,draw,minimum size=25pt,font=\sffamily\Large\bfseries},source node/.style={circle,fill=green!20,draw,minimum size=25pt,font=\sffamily\Large\bfseries},dest node/.style={circle,fill=red!20,draw,minimum size=25pt,font=\sffamily\Large\bfseries}, scale=.8]  

  \node[source node] (1)  at (0,3) {$s$};
  \node[main node] (2) at (-2.5,0) {$v$};
  \node[dest node] (3) at (0,-3) {$d$};
  \node[main node] (4) at (2.5,0) {$w$};

  \path[every node/.style={font=\sffamily\small}]
    (1) edge node [right] {$\tau_{2}=1$} (4)
        edge node[left,color=black] {$\tau_{1}=0$}  (2)
    (2) edge node [left] {$\tau_{4}=1$} (3)
        edge node {$\tau_{3}=0$} (4)
    (4) edge node [right,color=black] {$\tau_{5}=0$}  (3);

\end{tikzpicture}
~\vspace{0cm} \caption{\label{fi:WheatstoneSM} Wheatstone network.}
\end{figure}
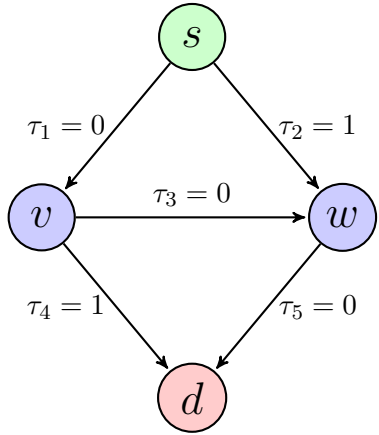

%\newpage

\begin{figure}[H]
\centering
\tikzstyle{main node}=[circle,fill=blue!20,draw,minimum size=25pt,font=\sffamily\Large\bfseries]
\tikzstyle{source node}=[circle,fill=green!20,draw,minimum size=25pt,font=\sffamily\Large\bfseries]
\tikzstyle{dest node}=[circle,fill=red!20,draw,minimum size=25pt,font=\sffamily\Large\bfseries]
\tikzstyle{fake node}=[circle,minimum size=22pt]
\tikzstyle{small blue node}=[blue,fill=white,font=\sffamily\bfseries]
\tikzstyle{small red node}=[red,fill=white,font=\sffamily\bfseries]
\tikzstyle{small black node}=[black,fill=white,font=\sffamily\bfseries]
\tikzstyle{small node}=[font=\sffamily\bfseries]

\begin{tikzpicture}[->,>=stealth',shorten >=1pt,auto,node distance=4.5cm,thick,scale=.8] 

  \node[source node] (1)  at (0,3) {$s$};
  \node[main node] (2) at (-2.5,0) {$v$};
  \node[dest node] (3) at (0,-3) {$d$};
  \node[main node] (4) at (2.5,0) {$w$};
  \node[fake node] (5) at (-.2,3.1) {};
  \node[fake node] (6) at (0,3) {};
  \node[fake node] (7) at (-2.7,0.1) {};
  \node[fake node] (8) at (-2.5,0) {};
  \node[fake node] (9) at (-.1,-3) {};
  \node[fake node] (10) at (.1,-3.1) {};
  \node[fake node] (11) at (2.6,-.1) {};
  \node[fake node] (12) at (2.4,0) {};
    
\draw[-stealth] (1)[] to  (2);

\draw[-stealth] (2)[] to  (3);

\node at (0,-4.5) {$e_{1},e_{4}$};
\node at (0,-5) {};
\end{tikzpicture}
\qquad
\begin{tikzpicture}[->,>=stealth',shorten >=1pt,auto,node distance=4.5cm,thick,scale=.8] 

  \node[source node] (1)  at (0,3) {$s$};
  \node[main node] (2) at (-2.5,0) {$v$};
  \node[dest node] (3) at (0,-3) {$d$};
  \node[main node] (4) at (2.5,0) {$w$};
  \node[fake node] (5) at (-.2,3.1) {};
  \node[fake node] (6) at (0,3) {};
  \node[fake node] (7) at (-2.7,0.1) {};
  \node[fake node] (8) at (-2.5,0) {};
  \node[fake node] (9) at (-.1,-3) {};
  \node[fake node] (10) at (.1,-3.1) {};
  \node[fake node] (11) at (2.6,-.1) {};
  \node[fake node] (12) at (2.4,0) {};
    
\draw[-stealth] (1)[red] to  (2);
\draw[-stealth] (2)[red] to  (4);
\draw[-stealth] (4)[red] to  (3);

\node at (0,-4.5) {$e_{1},e_{3},e_{5}$};
\node at (0,-5) {};
\end{tikzpicture}
\qquad
\begin{tikzpicture}[->,>=stealth',shorten >=1pt,auto,node distance=4.5cm,thick,scale=.8] 

  \node[source node] (1)  at (0,3) {$s$};
  \node[main node] (2) at (-2.5,0) {$v$};
  \node[dest node] (3) at (0,-3) {$d$};
  \node[main node] (4) at (2.5,0) {$w$};
  \node[fake node] (5) at (-.2,3.1) {};
  \node[fake node] (6) at (0,3) {};
  \node[fake node] (7) at (-2.7,0.1) {};
  \node[fake node] (8) at (-2.5,0) {};
  \node[fake node] (9) at (-.1,-3) {};
  \node[fake node] (10) at (.1,-3.1) {};
  \node[fake node] (11) at (2.6,-.1) {};
  \node[fake node] (12) at (2.4,0) {};
    
\draw[-stealth] (1)[blue] to  (4);
\draw[-stealth] (4)[blue] to  (3);

\node at (0,-4.5) {$e_{2},e_{5}$};
\node at (0,-5) {};
\end{tikzpicture}
~\vspace{0cm} \caption{\label{fi:Wheatstonecolorcode} Wheatstone network color code.}\end{figure}
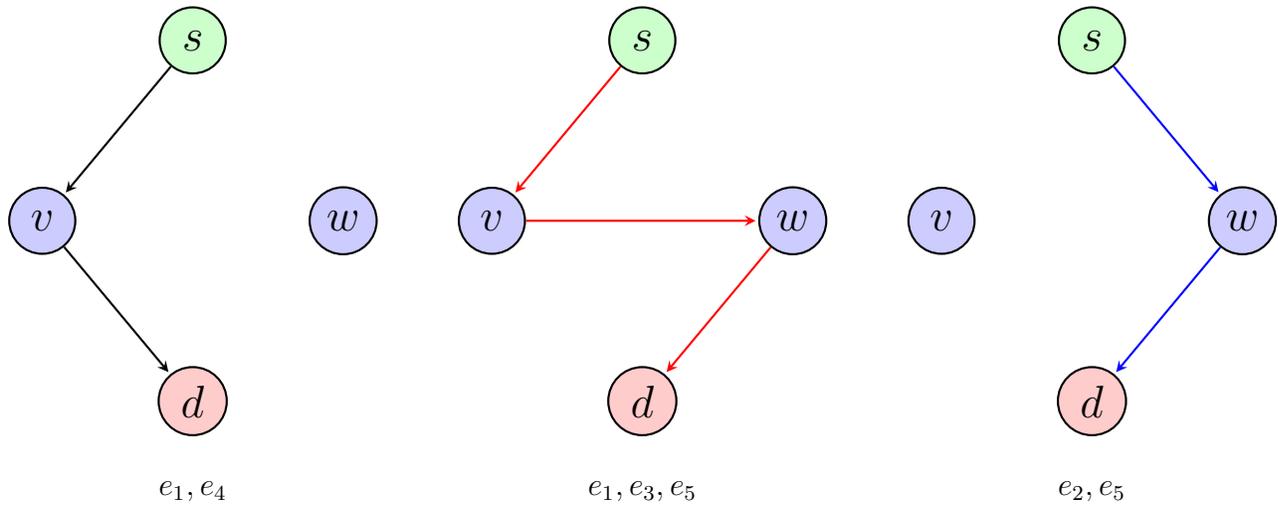

Consider the following equilibrium strategy 
\begin{equation}\label{eq:WheatstoneequilSM}
\sigma_{it}^{\Eq}=\begin{cases}
e_{1} e_{3} e_{5} &\text{ for } [it]=[i1],i=1,2,\\
e_{1} e_{3} e_{5} &\text{ for } [it]=[12],\\
e_{2} e_{5} &\text{ for } [it]=[22],\\
e_{1} e_{3} e_{5} &\text{ for } [it]=[13],\\
e_{1} e_{4} &\text{ for } [it]=[23],\\
e_{2} e_{5} &\text{ for } [it]=[14],\\
e_{1} e_{3} e_{5} &\text{ for } [it]=[24],\\
e_{1} e_{4} &\text{ for } [it]=[1t],t\geq5,\\
e_{2} e_{5} &\text{ for } [it]=[2t],t\geq5.
\end{cases}
\end{equation}
We refer to Figure \ref{fi:Wheatstoneequilibrium} to check that this is indeed an equilibrium.

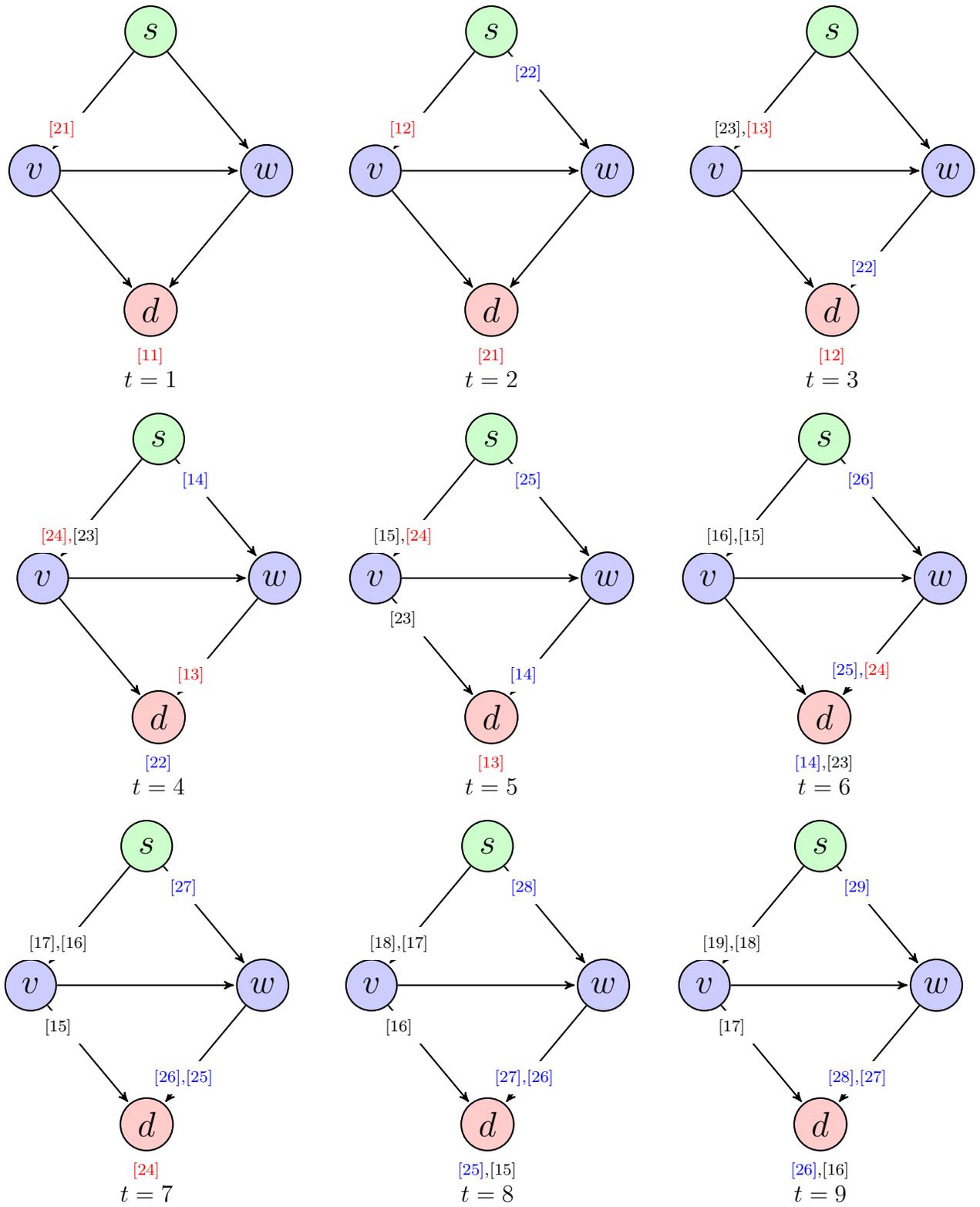
\begin{figure}[H]
\centering
\tikzstyle{main node}=[circle,fill=blue!20,draw,minimum size=25pt,font=\sffamily\Large\bfseries]
\tikzstyle{source node}=[circle,fill=green!20,draw,minimum size=25pt,font=\sffamily\Large\bfseries]
\tikzstyle{dest node}=[circle,fill=red!20,draw,minimum size=25pt,font=\sffamily\Large\bfseries]
\tikzstyle{fake node}=[circle,minimum size=22pt]
\tikzstyle{small blue node}=[blue,fill=white,font=\sffamily\bfseries]
\tikzstyle{small red node}=[red,fill=white,font=\sffamily\bfseries]
\tikzstyle{small black node}=[black,fill=white,font=\sffamily\bfseries]
\tikzstyle{small node}=[font=\sffamily\bfseries]

\begin{tikzpicture}[->,>=stealth',shorten >=1pt,auto,node distance=4.5cm,thick,scale=.8] 

  \node[source node] (1)  at (0,3) {$s$};
  \node[main node] (2) at (-2.5,0) {$v$};
  \node[dest node] (3) at (0,-3) {$d$};
  \node[main node] (4) at (2.5,0) {$w$};

  \path[every node/.style={font=\sffamily\small}]
    (1) edge node [right] {} (4)
        edge node[left,color=black] {}  (2)
    (2) edge node [left] {} (3)
        edge node {} (4)
    (4) edge node [right,color=black] {}  (3);
    
\node[small red node] at (0,-4) {${}_{[11]}$};
\node[small red node] at (-1.9,.9) {${}_{[21]}$};
\node at (0,-4.5) {$t=1$};
\node at (0,-5) {};
\end{tikzpicture}
\qquad
\begin{tikzpicture}[->,>=stealth',shorten >=1pt,auto,node distance=4.5cm,thick,scale=.8] 

  \node[source node] (1)  at (0,3) {$s$};
  \node[main node] (2) at (-2.5,0) {$v$};
  \node[dest node] (3) at (0,-3) {$d$};
  \node[main node] (4) at (2.5,0) {$w$};

  \path[every node/.style={font=\sffamily\small}]
    (1) edge node [right] {} (4)
        edge node[left,color=black] {}  (2)
    (2) edge node [left] {} (3)
        edge node {} (4)
    (4) edge node [right,color=black] {}  (3);
    
\node[small blue node] at (.8,2.1) {${}_{[22]}$};
\node[small red node] at (-1.9,.9) {${}_{[12]}$};
\node[small red node] at (0,-4) {${}_{[21]}$};
\node at (0,-4.5) {$t=2$};
\node at (0,-5) {};
\end{tikzpicture}
\qquad
\begin{tikzpicture}[->,>=stealth',shorten >=1pt,auto,node distance=4.5cm,thick,scale=.8] 

  \node[source node] (1)  at (0,3) {$s$};
  \node[main node] (2) at (-2.5,0) {$v$};
  \node[dest node] (3) at (0,-3) {$d$};
  \node[main node] (4) at (2.5,0) {$w$};

  \path[every node/.style={font=\sffamily\small}]
    (1) edge node [right] {} (4)
        edge node[left,color=black] {}  (2)
    (2) edge node [left] {} (3)
        edge node {} (4)
    (4) edge node [right,color=black] {}  (3);
    
\node[small black node] at (-1.9,.9) {${}_{[23],\textcolor{red}{[13]}}$};
\node[small blue node] at (.7,-2.1) {${}_{[22]}$};
\node[small red node] at (0,-4) {${}_{[12]}$};
\node at (0,-4.5) {$t=3$};
\node at (0,-5) {};
\end{tikzpicture}
\qquad
\begin{tikzpicture}[->,>=stealth',shorten >=1pt,auto,node distance=4.5cm,thick,scale=.8] 

  \node[source node] (1)  at (0,3) {$s$};
  \node[main node] (2) at (-2.5,0) {$v$};
  \node[dest node] (3) at (0,-3) {$d$};
  \node[main node] (4) at (2.5,0) {$w$};

  \path[every node/.style={font=\sffamily\small}]
    (1) edge node [right] {} (4)
        edge node[left,color=black] {}  (2)
    (2) edge node [left] {} (3)
        edge node {} (4)
    (4) edge node [right,color=black] {}  (3);

\node[small blue node] at (.8,2.1) {${}_{[14]}$};    
\node[small red node] at (-1.9,.9) {${}_{[24],\textcolor{black}{[23]}}$};
\node[small red node] at (.7,-2.1) {${}_{[13]}$};
\node[small blue node] at (0,-4) {${}_{[22]}$};
\node at (0,-4.5) {$t=4$};
\node at (0,-5) {};
\end{tikzpicture}\qquad
\begin{tikzpicture}[->,>=stealth',shorten >=1pt,auto,node distance=4.5cm,thick,scale=.8] 

  \node[source node] (1)  at (0,3) {$s$};
  \node[main node] (2) at (-2.5,0) {$v$};
  \node[dest node] (3) at (0,-3) {$d$};
  \node[main node] (4) at (2.5,0) {$w$};

  \path[every node/.style={font=\sffamily\small}]
    (1) edge node [right] {} (4)
        edge node[left,color=black] {}  (2)
    (2) edge node [left] {} (3)
        edge node {} (4)
    (4) edge node [right,color=black] {}  (3);

\node[small blue node] at (.8,2.1) {${}_{[25]}$};        
\node[small black node] at (-1.9,.9) {${}_{[15],\textcolor{red}{[24]}}$};
\node[small black node] at (-1.9,-.9) {${}_{[23]}$};
\node[small blue node] at (.7,-2.1) {${}_{[14]}$};
\node[small red node] at (0,-4) {${}_{[13]}$};
\node at (0,-4.5) {$t=5$};
\node at (0,-5) {};
\end{tikzpicture}\qquad
\begin{tikzpicture}[->,>=stealth',shorten >=1pt,auto,node distance=4.5cm,thick,scale=.8] 

  \node[source node] (1)  at (0,3) {$s$};
  \node[main node] (2) at (-2.5,0) {$v$};
  \node[dest node] (3) at (0,-3) {$d$};
  \node[main node] (4) at (2.5,0) {$w$};

  \path[every node/.style={font=\sffamily\small}]
    (1) edge node [right] {} (4)
        edge node[left,color=black] {}  (2)
    (2) edge node [left] {} (3)
        edge node {} (4)
    (4) edge node [right,color=black] {}  (3);

\node[small blue node] at (.8,2.1) {${}_{[26]}$};    
\node[small black node] at (-1.9,.9) {${}_{[16],[15]}$};
\node[small blue node] at (.8,-2) {${}_{[25],\textcolor{red}{[24]}}$};
\node[small node] at (0,-4) {${}_{\textcolor{blue}{[14]},[23]}$};
\node at (0,-4.5) {$t=6$};
\node at (0,-5) {};
\end{tikzpicture}
\qquad
\begin{tikzpicture}[->,>=stealth',shorten >=1pt,auto,node distance=4.5cm,thick,scale=.8] 

  \node[source node] (1)  at (0,3) {$s$};
  \node[main node] (2) at (-2.5,0) {$v$};
  \node[dest node] (3) at (0,-3) {$d$};
  \node[main node] (4) at (2.5,0) {$w$};

  \path[every node/.style={font=\sffamily\small}]
    (1) edge node [right] {} (4)
        edge node[left,color=black] {}  (2)
    (2) edge node [left] {} (3)
        edge node {} (4)
    (4) edge node [right,color=black] {}  (3);

\node[small blue node] at (.8,2.1) {${}_{[27]}$};    
\node[small black node] at (-1.9,.9) {${}_{[17],[16]}$};
\node[small blue node] at (.8,-2) {${}_{[26],[25]}$};
\node[small black node] at (-1.9,-.9) {${}_{[15]}$};
\node[small node] at (0,-4) {${}_{\textcolor{red}{[24]}}$};
\node at (0,-4.5) {$t=7$};
\node at (0,-5) {};
\end{tikzpicture}
\qquad
\begin{tikzpicture}[->,>=stealth',shorten >=1pt,auto,node distance=4.5cm,thick,scale=.8] 

  \node[source node] (1)  at (0,3) {$s$};
  \node[main node] (2) at (-2.5,0) {$v$};
  \node[dest node] (3) at (0,-3) {$d$};
  \node[main node] (4) at (2.5,0) {$w$};

  \path[every node/.style={font=\sffamily\small}]
    (1) edge node [right] {} (4)
        edge node[left,color=black] {}  (2)
    (2) edge node [left] {} (3)
        edge node {} (4)
    (4) edge node [right,color=black] {}  (3);

\node[small blue node] at (.8,2.1) {${}_{[28]}$};    
\node[small black node] at (-1.9,.9) {${}_{[18],[17]}$};
\node[small blue node] at (.8,-2) {${}_{[27],[26]}$};
\node[small black node] at (-1.9,-.9) {${}_{[16]}$};
\node[small node] at (0,-4) {${}_{\textcolor{blue}{[25]},[15]}$};
\node at (0,-4.5) {$t=8$};
\node at (0,-5) {};
\end{tikzpicture}\qquad
\begin{tikzpicture}[->,>=stealth',shorten >=1pt,auto,node distance=4.5cm,thick,scale=.8] 

  \node[source node] (1)  at (0,3) {$s$};
  \node[main node] (2) at (-2.5,0) {$v$};
  \node[dest node] (3) at (0,-3) {$d$};
  \node[main node] (4) at (2.5,0) {$w$};

  \path[every node/.style={font=\sffamily\small}]
    (1) edge node [right] {} (4)
        edge node[left,color=black] {}  (2)
    (2) edge node [left] {} (3)
        edge node {} (4)
    (4) edge node [right,color=black] {}  (3);

\node[small blue node] at (.8,2.1) {${}_{[29]}$};    
\node[small black node] at (-1.9,.9) {${}_{[19],[18]}$};
\node[small blue node] at (.8,-2) {${}_{[28],[27]}$};
\node[small black node] at (-1.9,-.9) {${}_{[17]}$};
\node[small node] at (0,-4) {${}_{\textcolor{blue}{[26]},[16]}$};
\node at (0,-4.5) {$t=9$};
\node at (0,-5) {};
\end{tikzpicture}
~\vspace{0cm} \caption{\label{fi:Wheatstoneequilibrium} Wheatstone network equilibrium \eqref{eq:WheatstoneequilSM}.}
\end{figure}

\begin{proof}[Proof of Proposition~\ref{pr:pospoa}]
For $k\in\mathbb{N}_{+}$, define  \emph{Braess's $k$-th graph} as follows.
Since this is just a graph and not a multigraph, edges are uniquely identified by their tail and head. Let 
\[
V^{k}=\{s,v_{1},\ldots,v_{k}, w_{1},\ldots,w_{k},d\}
\] 
be the set of $2k+2$ vertices and 
\[
E^{k}=\{(s,v_{i}),(v_{i},w_{i}),(w_{i},d)\mid1\leq i\leq k\}\cup \{(v_{i},w_{i-1})\mid2\leq i\leq k\}\cup \{(v_{1},d)\}\cup \{(s,w_{k})\}
\] 
the set of edges. See Figure~\ref{fi:braess1}.

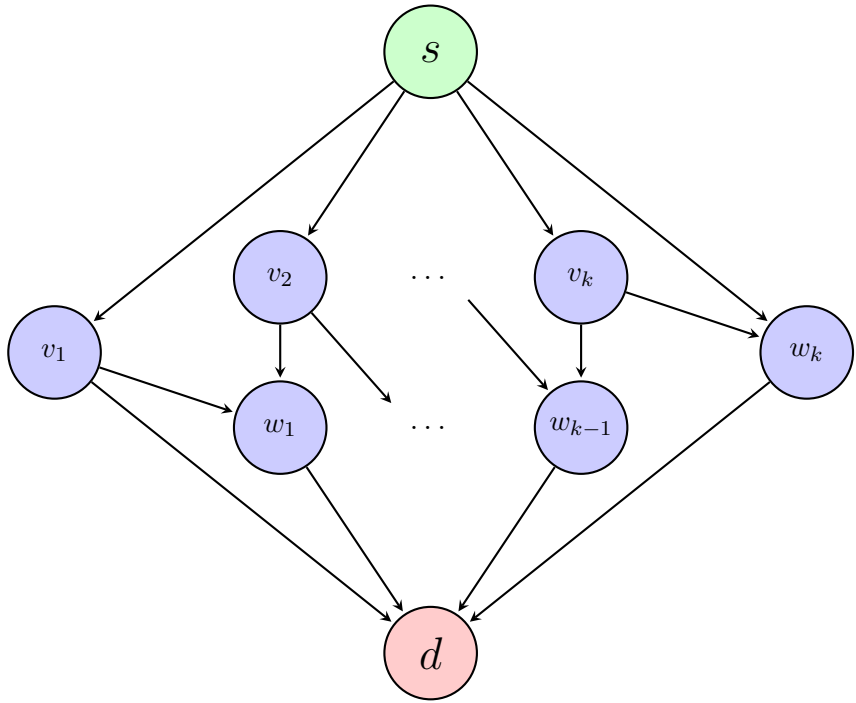
\begin{figure}[h]
\centering
\begin{tikzpicture}[->,>=stealth',shorten >=1pt,auto,node distance=5cm,
  thick,main node/.style={circle,fill=blue!20,draw,minimum size=35pt,font=\sffamily\small\bfseries},source node/.style={circle,fill=green!20,draw,minimum size=35pt,font=\sffamily\Large\bfseries},dest node/.style={circle,fill=red!20,draw,minimum size=35pt,font=\sffamily\Large\bfseries}]
	\node[source node] (source) at (0,4) {$s$};
	\node[main node] (left) at (-5,0) {$v_{1}$};
	\node[main node] (upl) at (-2,1) {$v_{2}$};
	\node[main node] (upr) at (2,1) {$v_{k}$};
	\node[main node] (right) at (5,0) {$w_{k}$};
	\node[main node] (downl) at (-2,-1) {$w_{1}$};
	\node[main node] (downr) at (2,-1) {$w_{k-1}$};
	\node[dest node] (sink) at (0,-4) {$d$};
	\node at (0,1) {\ldots};
	\node at (0,-1) {\ldots};
\draw[-stealth] (source) to (left);
\draw[-stealth] (source) to (right);
\draw[-stealth] (source) to (upl);
\draw[-stealth] (source) to (upr);
\draw[-stealth] (left) to (downl);
\draw[-stealth] (left) to (sink);
\draw[-stealth] (right) to (sink);
\draw[-stealth] (upl) to (downl);
\draw[-stealth] (upr) to (downr);
\draw[-stealth] (upr) to (right);
\draw[-stealth] (upl) to (-.5,-.7);
\draw[-stealth] (.5,.7) to (downr);
\draw[-stealth] (downl) to (sink);
\draw[-stealth] (downr) to (sink);
\end{tikzpicture}
\caption{Braess's $k$-th  graph.} \label{fi:braess1}
\end{figure}

Let $\gamma_{e}=1$ for all $e\in E^{k}$ and
\begin{align*}
\tau_{e}&=\begin{cases}
1 &\text{ if }e=(v_{1},d)\text{, }e=(s,w_{k})\text{ or }e=(v_{i},w_{i-1})\text{ for }2\leq i\leq k,\\
0 &\text{ otherwise,}
\end{cases}
\end{align*}
Notice that  Braess's $k$-th  graph has a capacity  $k+1$.

For $i=1,\ldots,k$, let $P_{i}$ denote the path $(s,v_{i})(v_{i},w_{i})(w_{i},d)$. Let $Q_{1}$ denote the path $(s,v_{1})(v_{1},d)$, for $i=2,\ldots,k$ let $Q_{i}$ denote the path $(s,v_{i})(v_{i},w_{i-1})(w_{i-1},d)$ and let $Q_{k+1}$ denote the path $(s,w_{k})(w_{k},d)$.

The optimal latency is achieved by the strategy profile in which each player of every generation chooses a different path $Q_{i}$ for $i=1,\ldots,k+1$. Hence
\[
\Opt(\mathcal{N},\gamma)=k+1.
\]

Consider the subnetwork $\mathcal{N}'$ obtained from $\mathcal{N}$ by deleting each edge $(v_{i},w_{i})$ for $i=1,\ldots,k$. 
Observe that  $\mathcal{N}'$ is  a parallel network and
the unique equilibrium latency is achieved by the strategy profile in which each player of every generation chooses a different path $Q_{i}$ for $i=1,\ldots,k+1$. Hence 
\[
\WEq(\mathcal{N}',\gamma)=k+1.
\]

Now, on $\mathcal{N}$ the best equilibrium latency is achieved by the following strategy profile.
\begin{enumerate}[(i)]
\item In the $j$-th period, where $1\leq j\leq k$, the first $k+1-j$ players choose path $P_{i}$ for $i=1,\ldots,k$ in increasing order and players $k+2-j$, \ldots, $k+1$ choose path $Q_{k+2-j}$, \ldots, $Q_{k+1}$, respectively.
\item From period $k+1$ onwards, each player chooses a path $Q_{i}$ for $i=1,\ldots,k+1$ in increasing order.
\end{enumerate}
Since no queue is created in any of the periods, 
\[
\BEq(\mathcal{N},\gamma)=k+1.
\]

The worst equilibrium latency is achieved by the following strategy profile. For all $[it]\in G$, choose a path $p_{it}$ with minimum latency that has no possibility of overtaking, according to the following preference relation over paths $P_{1}\succ\ldots\succ P_{k}\succ Q_{1}\succ Q_{k+1}\succ Q_{2}\succ\ldots\succ Q_{k}$. The idea is that the players in the transient states create queues on each $P_{i}$ for $i=1,\ldots,k$, in such a way that in the steady state,  in each generation exactly one player chooses the path $Q_{i}$ for $i=1,\ldots,k+1$, with a latency of $2k+1$. For $k=1$, this strategy profile is illustrated in Example \ref{ex:Wheatstone}. For $k\geq2$, queues grow as follows.

\begin{enumerate}[(i)]
\item In the first $k$ periods, each player chooses a path $P_{i}$ for $i=1,\ldots,k$ in increasing order such that after $k$ periods, each edge $(s,v_{i})$ for $i=1,\ldots,k$ has a waiting cost of one.

\item Partition the following $(k-1)\cdot2k$ periods into $k-1$ sets of $2k$ periods. Each set of $2k$ periods consists of two subsets of $k$ periods such that players in the second $k$ periods choose the same routes as players in the first $k$ periods. During the $j$-th set of $2k$ periods, where $1\leq j\leq k-1$, all players in the first $k-j$ periods choose a path $P_{i}$ for $i=1,\ldots,k$ in increasing order and create a queue on $(s,v_{i})$ for $i=1,\ldots,k-j$. In the next $j$ periods, the UFR property implies that a path $P_{i}$ is replaced by a path $Q_{i+1}$ for $i=k,\ldots,k+1-j$ such that a queue grows on $(w_{i},d)$ instead of on $(s,v_{i})$.

So during the $j$-th set of $2k$ periods, where $1\leq j\leq k-1$, the queue on each edge $(s,v_{i})$ for $i=1,\ldots,k-j$ has increased by two, and the queue on $(w_{i},d)$ for $i=k,\ldots,k+1-j$ has increased by two.

\item After $(k-1)\cdot2k+k$ periods, the length of the queue on each edge $(s,v_{i})$ for $i=1,\ldots,k$ is $2\cdot(k-i)+1$, and the length of the queue on each edge $(w_{i},d)$ for $i=2,\ldots,k$ is $2\cdot(i-1)$. The UFR property implies that in the following $k$ periods a queue grows on each edge $(w_{i},d)$ for $i=1,\ldots,k$.

\item The subsequent $k+1$ periods are summarized as follows. First, a queue grows on $(s,v_{1})$, then a queue grows on $(s,v_{i})$ and $(w_{i-1},d)$ for $i=2,\ldots,k$, finally a queue grows on $(w_{k},d)$. Summarizing, each path $P_{i}$ for $i=1,\ldots,k$ has a latency of $2k+2$ and each path $Q_{i}$ for $i=1,\ldots,k+1$ has a latency of $2k+1$.

\item In all of the upcoming periods, each player chooses a different path $Q_{i}$ for $i=1,\ldots,k+1$. The UFR property guarantees that queues cannot grow any further.
\end{enumerate}

Hence 
\[
\WEq(\mathcal{N},\gamma)=(k+1)\cdot(2k+1).
\]

Concluding, we found that
\begin{align*}
\PoS(\mathcal{N},\gamma)&=\frac{k+1}{k+1}=1,\\
\PoA(\mathcal{N},\gamma)&=\BR(\mathcal{N},\gamma)=\frac{(k+1)\cdot(2k+1)}{(k+1)}=2k+1=n-1. \qedhere
\end{align*}
\end{proof}

\begin{remark}
For the $k$-th Braess's graph, the Nash latency of the non-atomic game is achieved by the following strategy profile. First, a queue of length one is created on each $(s,v_{i})$ for $i=1,\ldots,k$. Then congestion
occurs on each  $(s,v_{i})$ for $i=1,\ldots,k-1$ and $(w_{k},d)$ until there is an additional queue of one. This process, where more queues grow on $(w_{i},d)$ instead of on $(s,v_{i})$ for $i=1,\ldots,k$, continues until each path has the same latency equal to $k+1$. Hence the non-atomic game has a price of anarchy of $n/2$.
\end{remark}

\subsection*{Section~\ref{se:seasonal}}
\begin{proof}[Proof of Proposition~\ref{pr:parallelperiod}]
We first prove the formula for the optimum cost,
\begin{equation*}
\Opt(\mathcal{N},K, \boldsymbol{\delta})=K\sum_{e\in E}\gamma_{e}\tau_{e}+
D(\boldsymbol{\delta}).
\end{equation*}
We prove this result by induction on $D( \boldsymbol{\delta})$. For $D( \boldsymbol{\delta})=0$, note that $ \boldsymbol{\delta}= \boldsymbol{\gamma}$ and the result is obvious.

Suppose the result is true for $ \boldsymbol{\delta}'$ with $D( \boldsymbol{\delta}')\in\mathbb{N}$ and let $ \boldsymbol{\delta}\rightarrow \boldsymbol{\delta}'$. Since $ \boldsymbol{\delta}\rightarrow \boldsymbol{\delta}'$, there is some $k\in\{1,\ldots,K\}$ such that $\delta_k>\gamma$, $\delta'_k=\delta_k-1$, $\delta'_{k+1}=\delta_{k+1}+1$ and $\delta'_{\ell}=\delta_{\ell}$ for all $\ell\notin\{k,k+1\}$, where $k+1$ is considered modulo $K$.

At each stage $t_k$ such that $t=k \mod K$, players depart above capacity under $\boldsymbol{\delta}$. This implies that there is at least one player $[jt_k]$ who  sees a queue on his route, and thus who adds one unit of waiting time to the total cost. Denote $[j^*t_k]$ such a player with the highest index, i.e., the player with the   lowest priority. Consider the relaxed optimization problem where the planner postpones the departure of this player by one stage and let her depart as the first player of the next generation, that is, to transform $\boldsymbol{\delta}$ into $\boldsymbol{\delta}'$. 

By the choice of $j^*$ (the last one in the generation who sees a queue), the postponing of this player does not affect the costs nor the choices of the other players. This is clear for those who have higher priority. For those who have lower priority, this player will be ahead of them in the queue in both cases. So the choice of strategy for player $[j^*t_k]$ that has to be made by the social planner is the same in $\boldsymbol{\delta}$ as in $\boldsymbol{\delta}'$. Hence all players, including player $[j^*t_k]$, choose the same strategy, and thus in each period, one unit of waiting cost is saved by postponing the departure of player $[j^*t_k]$. This concludes the proof.

\bigskip
 
We now turn to the proof of the formula for the equilibrium. We start by showing some preliminary results. The first claim shows that for computing equilibrium costs, without loss of generality  all capacities can be assumed to be $1$.

Let $\mathcal{N}=(\mathcal{G}, (\tau_e)_{e\in E}, (\gamma_e)_{e\in E})$ be a network. Given $e\in E$, let  $\mathcal{N}^{e}$ be the network obtained from $\mathcal{N}$ by replacing the edge $e$ of capacity $\gamma_e$, by a set $E(e)$ of $\gamma_e$ parallel edges of capacity 1, with the same head and tail, and same length as $e$.

\begin{claim}\label{cl:cap1} 
Every equilibrium of $\mathcal{N}$ (resp. $\mathcal{N}^{e}$) can be mapped to an equilibrium of $\mathcal{N}^{e}$ (resp. $\mathcal{N}$) with the same total cost.
\end{claim}

\begin{proof}
We index the $\gamma_e$ edges $E(e)$ by the integers $\{1,\dots,\gamma_e\}$. We consider an equilibrium $\sigma$ of $\mathcal{N}$ and construct an equilibrium of $\mathcal{N}^{e}$ with the same total cost.

Suppose that under $\sigma$,  at generation $t$, $n$ players numbered $[i_1t],\dots,[i_nt]$ enter edge $e$.  

First assume that there is no initial queue on $e$ at the beginning of stage $t$. For each $k=1,\dots, n$, assign player $[i_kt]$ to the $q$-th edge if $k=q\mod\gamma_e$. In words, take those $n$ players and assign them to the edges according to their  priority, following the numbering of edges: player $[i_1t]$ is assigned to edge 1, \dots, player $[i_kt]$ is assigned to edge $k$ for $k\leq  \gamma_e$. If $n>\gamma_e$, then  player $[i_{\gamma_{e}+1}t]$ is assigned to edge 1, and so on. By construction, both ways, player $[i_kt]$ will queue on $e$ if $k>\gamma_e$, and if $k=w\gamma_e+r$ ($w, r$ integers, $r<\gamma_e$),  player $[i_kt]$ will queue for $w-1$ units of time. Therefore the total cost paid by this player on $e$ is the same in both cases. This defines an equilibrium for those players: there is no point in deviating to a route feasible in $\mathcal{N}$, as it would imply a profitable deviation from $\sigma$. By construction, each player is assigned to an edge of $E(e)$ which is fastest, given the priorities.

With this  construction, the queues left by this generation to the next one has the following structure: there exist $w$ and $q^*\leq \gamma_e$ such that all edges of $E(e)$ numbered $1,\dots, q^*$ have a queue of length $w$, and edges numbered $q^*+1,\dots, \gamma_e$ have a queue of length $w-1$ (if $q^*=\gamma_e$, all queues have the same length).

Second, suppose again that at generation $t$, $n$ players numbered $[i_1t],\dots,[i_nt]$ enter edge $e$ under $\sigma$, but that on $E(e)$, they see queues with the above structure. If $q^*<\gamma_e$, then let the first $\gamma_e-q^*$ players fill the edges numbered $q^*+1,\dots,\gamma_e$ in an orderly fashion, according to priorities. The remaining $n-(\gamma_e-q^*)$ choose edges as in the previous case: the first player chooses the first edge, and so on.

As in the previous case, since the choice of edges in $E(e)$ respects the priorities, the waiting time is the same for each player on both networks. Also for the same reason as in the previous case, this is an equilibrium choice. Note that the above structure of queues is preserved from one generation to the next, so that the analysis can be iterated. We have thus constructed an equilibrium of $\mathcal{N}^{e}$ with the same total cost as $\sigma$.

Conversely, take an equilibrium $\sigma^{e}$ of $\mathcal{N}^{e}$. For each route in $\mathcal{N}^{e}$ that uses an edge $f\in E(e)$, there is a unique corresponding route in $\mathcal{N}$ which uses  edge $e$. This maps uniquely the strategy profile $\sigma^{e}$ in a strategy profile $\sigma$ on $\mathcal{N}$. Then $\sigma$ has to be an equilibrium. Actually, a  deviation in $\mathcal{N}$ is also feasible  in  $\mathcal{N}^{e}$, so a profitable deviation from $\sigma$ would imply a profitable deviation from $\sigma^{e}$.

 From the fact that $\sigma^{e}$ is an equilibrium of $\mathcal{N}^{e}$, the queues on the edges of $E(e)$ (if at all) must have a structure as above, there is an integer $w$ such that each edge in an non-empty subset of $E(e)$ has a queue of length $w$, and all other edges in $E(e)$ have a queue of length $w-1$. 
Therefore, the waiting time of a player on edge $e$ is the same as on $E(e)$.
\end{proof}

Applying this result iteratively we can transform any network into another where all capacities are one.

\begin{lemma}\label{le:worsteq}
Let $\mathcal{N}$ be a parallel network. In a worst equilibrium of $\Gamma(\mathcal{N},\mathcal{D})$, whenever a player is indifferent between several edges, she chooses one where there is a queue, if there is one.
\end{lemma}

\begin{proof} Using Claim \ref{cl:cap1}, we assume that all capacities are one.
For a parallel network, arriving at intermediary nodes is not an issue. Therefore, each generation of players chooses an equilibrium as in a game where no subsequent generations exist.
More precisely, consider a parallel network with edges $e\in E$ and lengths $(\tau_e)_{e\in E}$. Consider the game where at stage 1 a generation of $\delta$ players enter the network and where there are no subsequent players. Denote $W(\delta,(\tau_e)_{e\in E})$ the worst equilibrium total cost of this game.

\begin{claim}\label{cl:monotonic}
$W(\delta,(\tau_e)_{e\in E})$ is weakly increasing in free-flow transit costs. That is, if, for all $e\in E$, $\tau_e\leq\tau_e'$, then $W(\delta,(\tau_e)_{e\in E})\leq W(\delta,(\tau'_e)_{e\in E})$.
\end{claim}

\begin{proof}
Since we are dealing with just one generation, we denote the players $1,\dots, \delta$. Take an equilibrium $\sigma$  and let $\ell_i(\sigma)$ be the equilibrium latency of player $i$. A simple remark is that $\ell_i(\sigma)$ weakly increases with $i$ and that from one player to the next, it can only increase by one unit. Precisely, there exist integers $1\leq k_1<k_2<\cdots<k_n\leq\delta$ such that
\begin{itemize}
\item whenever $1\leq i\leq k_1$, we have $\ell_i(\sigma)=\min_e\tau_e$,
%\item  $\forall i, k_1<i\leq k_2, \ell_i(\sigma)=\min_e\tau_e+1$,
\item if $1\leq m<n$, then, whenever $k_m<i\leq k_{m+1}$, we have $\ell_i(\sigma)=\min_e\tau_e+m$.
\end{itemize}
To see this, note  first that if $i<j$, then $\ell_i(\sigma)\leq\ell_j(\sigma)$. Otherwise, player $i$ who has priority over $j$ could profitably imitate $j$. Second, $\ell_{i+1}(\sigma)\leq \ell_i(\sigma)+1$. Otherwise, player $i+1$ could profitably imitate $i$ and pay $\ell_i(\sigma)+1$.

If follows directly that if we increase $\min_e\tau_e$, then the equilibrium costs of all players are pushed (weakly) upwards. Suppose now that we increase by one unit the length of an edge which is used in equilibrium. That is, take an edge $f$ with length $\tau_f=\min_e\tau_e+m$ for some $m$, with $1\leq m<n$ as above, and replace it by an edge with length $\tau_f+1$. In this new situation, we have the same number of players who pay $\tau_f-1$ at most. Among the players who paid $\tau_f$ in the old situation, one has to pay now $\tau_f+1$ (the one with the lowest priority; whether she chooses edge $f$ or another one with the same total cost). Subsequent players have to pay weakly more. So, all costs are weakly pushed upwards.
\end{proof}

 Consider now a worst equilibrium of $\Gamma(\mathcal{N},\mathcal{D})$. The first generation chooses an equilibrium for the network with  lengths $(\tau_e)_{e\in E}$. Let $n_1(e)$ be the number of players of the first generation who choose edge $e$. If $n_1(e)\leq 1$, then the first generation leaves no queue on $e$ for the next one. If $n_1(e)>1$, then the first player in the second generation meets a queue of $r_1(e)=n_1(e)-1$ on edge $e$. Iteratively, let $n_t(e)$ denote the number of players of generation $t$ who choose edge $e$ and $r_t(e)$ the queue that the first player in generation $t+1$ meets on $e$. We have the following recursion for $t>1$,
 \begin{equation*}
 r_t(e)=(r_{t-1}(e)+n_t(e)-1)_+,
 \end{equation*}
where $x_+=\max\{x,0\}$.

Then, generation $t+1$ chooses an equilibrium for the network $(\tau ^{t+1}_e)_{e\in E}$ with $\tau_e^{t+1}:=\tau_e+r_t(e)$.

Now, suppose that there is a generation $t$, a player $[it]$ and two edges $e,f$ such that player $[it]$ is indifferent between $e$ and $f$ and, there is a queue on $e$ but not on $f$. There are two equilibrium scenarios. In the best scenario (BS) player $[it]$ chooses $f$, in the worst scenario (WS) player $[it]$ chooses $e$. We argue that the queues left for future generations are all weakly higher in WS than in BS. 

Consider  first the case where player $[it]$ is the last in generation $t$. Then by choosing $f$, she leaves no queue on $f$ for the next generation and the queue on $e$ decreases by one unit. If she chooses $e$, she recreates the queue on $e$, there is still no queue on $f$. For all other edges, the queue is the same under both scenarios.

The second case is when player $[it]$ is not the last in her generation. Let $p=\delta_t-i+1$ be the number of players who come weakly after player $i$ in generation $t$, and let $q$ be the number of edges that have the same total cost as $e$ for player $[it]$. If $p>q$, then one player must choose $e$ and another one must choose $f$, no matter what player $[it]$ does, so the queues are the same under both scenarios. If $p\leq q$, then at most one player will choose $f$ (so no queue is created there) and no player chooses the same edge as $[it]$. Therefore, if she chooses $e$ she maintains the queue there, whereas she creates no queue by choosing $f$.

We conclude that whenever a player is indifferent between queuing or not, choosing the edge with the queue weakly increases all queues for the next generation. From the recursion $r_t(e)=(r_{t-1}(e)+n_t(e)-1)_+$, this weakly increases queues for all future generations. From Claim \ref{cl:monotonic}, the conclusion follows.
\end{proof}

We now turn to the proof of the formula for the equilibrium,
\begin{equation*}
\WEq(\mathcal{N}, K, \boldsymbol{\delta})=K\gamma\max_{e\in E}\tau_{e}+D(\boldsymbol{\delta}).
\end{equation*}

We prove it by induction on $D(\boldsymbol{\delta})$, the result being obvious for $D(\boldsymbol{\delta})=0$. We assume that the result is true for $\boldsymbol{\delta}'$ with $D(\boldsymbol{\delta}')\in\mathbb{N}$ and let $\boldsymbol{\delta}\rightarrow\boldsymbol{\delta}'$. We take  $k\in\{1,\ldots,K\}$ such that $\delta_k>\gamma$, $\delta'_k=\delta_k-1$, $\delta'_{k+1}=\delta_{k+1}+1$ and $\delta'_{\ell}=\delta_{\ell}$ for all $\ell\notin\{k,k+1\}$, where $k+1$ is considered modulo $K$.

First, we show that there is an equilibrium for $\Gamma(\mathcal{N},K,\boldsymbol{\delta})$ with costs equal to  $\WEq(\mathcal{N},K,\boldsymbol{\delta}')+1$. This implies that $\WEq(\mathcal{N},K,\boldsymbol{\delta})\geq \WEq(\mathcal{N},K,\boldsymbol{\delta}')+1$.

Let $\sigma'$ be the strategy profile as defined for $\Gamma(\mathcal{N},K,\boldsymbol{\delta}')$ yielding the worst equilibrium latency. We construct a strategy profile $\sigma$ for 
$\Gamma(\mathcal{N},K,\boldsymbol{\delta})$ corresponding to $\sigma'$ such that the same queues are created. The definition of the strategy is iterative. We indicate below how the construction works for one period and how the iteration proceeds to the next.

\begin{enumerate}[(I)]
\item\label{it:algo-I}
Let $k<K$ and $t\in\mathbb{N}$. If $t<k$, then let each player $[it]$ choose the same edge as in $\sigma'$. If $t=k$, then let each player $[it]$ with $i<\delta_k$ choose the same edge as in $\sigma'$.

\begin{enumerate}[(1)]
\item\label{it:algoI-1} 
If the  edge chosen by player $[1,k+1]$ in $\sigma'$ has minimum latency and waiting costs for player $[\delta_k k]$, then let $[\delta_k k]$ choose this edge and let each player $[it]$ with $t<k+K$ choose the same edge as in $\sigma'$. In this case, bringing forward a player does not affect the choice of the other players, because for them there is no difference whether the player waits a stage in a queue or whether the player waits a stage to depart. From stage $k+K$ onwards, go to \ref{it:algo-I} and iterate.

\item\label{it:algoI-2} If the  edge chosen by player $[1,k+1]$ in $\sigma'$ has either no minimum latency or no waiting costs for player $[\delta_k k]$, then let $[\delta_k k]$ choose an edge with minimum latency and no waiting costs (observe that in the former case, the edge has waiting costs and thus there must be a different edge with minimum latency but no waiting costs) and let each player $[i,k+1]$ with $i\leq\delta_{k+1}$ choose the same edge as in $\sigma'$.

\begin{enumerate}[(a)]
\item\label{it:algoI2-a} 
Either there is a first generation $G_s$ with $k+1\leq s<k+K$ which has a last player with no waiting costs in $\sigma'$. Let each player $[it]$ with $k+1<t\leq s$ choose the same edge as the player departing before $[it]$ in $\sigma'$ and let each player $[it]$ with $s<t<k+K$ choose the same edge as in $\sigma'$. In this case, $[\delta'_s s]$ does not affect the other players in $\sigma'$. So it is no problem if he does not depart. From stage $k+K$ onwards, go to \ref{it:algo-I}.

\item\label{it:algoI2-b} 
Or all generations $G_t$ with $k+1\leq t<k+K$ have a last player with waiting costs in $\sigma'$, then let each player $[it]$ with $k+1<t\leq k+K$ choose the same edge as the player departing before $[it]$ in $\sigma'$. From stage $k+K+1$ onwards, go to \ref{it:algo-II}.
\end{enumerate}
\end{enumerate}

\item\label{it:algo-II}
Let $k=K$ and $t\in\mathbb{N}$. Let each player $[i1]$ with $i\leq\delta_k$ choose the same edge as in $\sigma'$.

\begin{enumerate}[(1)]
\item\label{it:algoII-1} 
Either, there is a first generation $G_s$ with $1\leq s<K$ which has a last player with no waiting costs in $\sigma'$. Let each player $[it]$ with $1<t\leq s$ choose the same edge as the player departing before $[it]$ in $\sigma'$ and let each player $[it]$ with $s<t<K$ choose the same edge as in $\sigma'$. In this case, $[\delta'_ss]$ does not affect the other players in $\sigma'$. So it is no problem if she does not depart. From stage $K$ onwards, go to \ref{it:algo-I}.
\item\label{it:algoII-2}
Or, all generations $G_t$ with $1\leq t<K$ have a last player with waiting costs in $\sigma'$, then let each player $[it]$ with $1<t\leq K$ choose the same edge as the player departing before $[it]$ in $\sigma'$. From stage $K+1$ onwards, go to \ref{it:algo-II}.
\end{enumerate}
\end{enumerate}

Notice that $\sigma$ is defined in such a way that queues have the same length as under $\sigma'$,  either at the beginning or at end of stage $k+K$ . Queues have the same length at the beginning of stage $k+K$ in cases where the algorithm goes to \ref{it:algo-I}, and queues have the same length at the end of stage $k+K$ in cases where the algorithm goes to \ref{it:algo-II}.

Now, in order to compute the long-run latency, let us  focus on the steady state. 
We know that with uniform departures there is a $t_0$ such that for all generations $t\geq t_0$, queues are such that $\gamma_e$ players choose edge $e$ for all $e\in E$. Recall that $\delta_k>\gamma$. By construction, for each generation $t=k\mod{K}$, where $t\geq t_0+K$, and for all edges $e$ at least $\gamma_e$ players choose $e$. 
This implies that player $[\delta_k t]$ must wait for at least one period. So the waiting costs for $[\delta_k t]$ increases by one unit compared to $\sigma'$.

\medskip

Second, we show that there is an equilibrium of $\Gamma(\mathcal{N},K,\boldsymbol{\delta}')$ with costs equal to $\WEq(\mathcal{N},K,\boldsymbol{\delta})-1$. This implies that $\WEq(\mathcal{N},K,\boldsymbol{\delta})\leq \WEq(\mathcal{N},K,\boldsymbol{\delta}')+1$.

Fix a worst equilibrium $\sigma$ of $\Gamma(\mathcal{N}, K, \boldsymbol{\delta})$ and consider stage $k$. Since $\delta_k>\gamma$, queues must be created on some edges. Let $i^{*}$ be the maximal index such that player $[i^{*}k]$ meets a queue under $\sigma$. Then, by the choice of $i^{*}$, it must be that each subsequent player meets no queue. Further, each such player pays the same cost as $[i^{*}k]$. Indeed, if for $j>i^{*}$, player $[jk]$ pays less than $[i^{*}k]$, then $[i^{*}k]$ has a profitable deviation by imitating $[jk]$. If $[jk]$ pays more than $[i^{*}k]$, then she would meet a queue. Recall from Lemma \ref{le:worsteq} that in a worst equilibrium, in case of indifference, players choose an edge with a queue over an edge without a queue. So $[jk]$ would imitate $[i^{*}k]$, paying the cost of $[i^{*}k]$ plus $1$. This contradicts the definition of $i^{*}$.

Now, consider the game where $[i^{*}k]$ is postponed by one stage, starting as the first player of the next generation. In this game, if the postponed player chooses the exact same strategy, she  pays one unit less, since queues have decreased by  one. She cannot pay less than that, since that would have offered a profitable deviation for $[i^{*}k]$ in the original game. So it is an equilibrium. Since the two situations are identical for all other players, this reasoning can be iterated at each period.
\medskip

Hence, combining the previous two results yields
\begin{align*}
\WEq(\mathcal{N},K,\boldsymbol{\delta})&=\WEq(\mathcal{N},K,\boldsymbol{\delta}')+1,\\
&=\WEq(\mathcal{N},K,\gamma)+D(\boldsymbol{\delta}')+1,\\
&=\WEq(\mathcal{N},K,\gamma)+D(\boldsymbol{\delta}),
\end{align*}
where the second equality follows from the induction hypothesis.

\end{proof}

\subsection*{Nash equilibria}

The following example shows that there are Nash (but not UFR) equilibria for chain-of-parallel networks, where players may end up paying strictly more than the cost of the costlier route. The reason is that in a Nash equilibrium players need not arrive at intermediate vertices as early as possible and this may create additional queues.

\begin{example}\label{ex:verybadNash}
Consider the chain-of-parallel network in Figure~\ref{fi:chainofparallel2SM} with associated free-flow transit costs and capacity equal to $1$ for all edges. The capacity $\gamma^{(*)}$ of the network is $2$. 
\begin{figure}[h]
\centering
\begin{tikzpicture}[->,>=stealth',shorten >=1pt,auto,node distance=5cm,
  thick,main node/.style={circle,fill=blue!20,draw,minimum size=25pt,font=\sffamily\Large\bfseries},source node/.style={circle,fill=green!20,draw,minimum size=25pt,font=\sffamily\Large\bfseries},dest node/.style={circle,fill=red!20,draw,minimum size=25pt,font=\sffamily\Large\bfseries}]

  \node[source node] (1) {$s$};
  \node[main node] (2) [right of=1] {$v$};
  \node[dest node] (3) [right of=2] {$d$};
  
  \path[every node/.style={font=\sffamily\small}]
    (1) edge [bend right = 30] node[above] {$\tau_{1}^{(1)}=1$} (2)
        edge [bend left = 30] node[above] {$\tau_{2}^{(1)}=2$} (2);

  \path[every node/.style={font=\sffamily\small}]
    (2) edge [bend right = 30] node[above] {$\tau_{1}^{(2)}=1$} (3)
        edge [bend left = 30] node[above] {$\tau_{2}^{(2)}=2$} (3);

\end{tikzpicture}
~\vspace{0cm} 
\caption{\label{fi:chainofparallel2SM}Chain-of-parallel network.}
\end{figure}
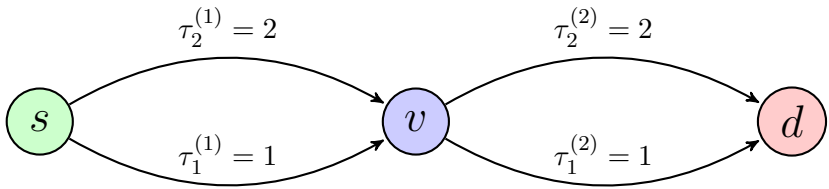

Consider the following equilibrium strategy profile
\begin{align*}
\sigma_{it}^{\Eq}=\begin{cases}
e_{1}^{(1)} e_{1}^{(2)} &\text{ for }[it]=[1t], t\leq2,\\
e_{1}^{(1)} e_{1}^{(2)} &\text{ for }[it]=[2t], t\leq2,\\
e_{2}^{(1)} e_{1}^{(2)} &\text{ for }[it]=[13],\\
e_{1}^{(1)} e_{1}^{(2)} &\text{ for }[it]=[23],\\
e_{2}^{(1)} e_{2}^{(2)} &\text{ for }[it]=[1t], t\geq4,\\
e_{1}^{(1)} e_{1}^{(2)} &\text{ for }[it]=[2t], t\geq4,\\
\end{cases}
\end{align*}
In the first two periods a queue of length two is created on $e_{1}^{(1)}$. Note that player $[22]$ cannot be overtaken as the next player departs in the following period. In period three, a new queue starts on $e_{1}^{(2)}$ and in stage four we reach the steady state.

The latency of this strategy profile equals $9$, which means that the last player of each generation pays more than the maximum free-flow transit costs.
\end{example}

\end{document}